\documentclass[a4paper,USenglish,numberwithinsect]{lipics-v2019}

\usepackage{microtype}
\usepackage{amsmath}
\usepackage{mathtools}
\usepackage{proof}

\hideLIPIcs

\usepackage{xspace}
\usepackage{mathpartir}
\usepackage{bussproofs}

\usepackage[table]{xcolor}

\usepackage[T1]{fontenc}
\usepackage{lmodern}

\usepackage{stmaryrd}

\newcommand{\balan}[1]{\ensuremath{\mathsf{balanced}(#1)}}
\newcommand{\subst}[2]{\{^{#1}\!/{\scriptstyle #2}\}}

\newcommand{\lrangle}[1]{\langle #1 \rangle}
\newcommand{\blrangle}[1]{\big\langle #1 \big\rangle}

\newcommand{\brtext}[1]{[\textrm{\small #1}]}

\newcommand{\orule}[1]{{\scriptsize{\brtext{#1}}}}

\newcommand{\bnfis}{\;\;::=\;\;}
\newcommand{\bnfbar}{\;\;\;|\;\;\;}
\newcommand{\sbnfbar}{\;\;|\;\;}

\newcommand{\set}[1]{\{#1\}}
\newcommand{\es}{\emptyset}

\newcommand{\dom}[1]{\mathtt{dom}(#1)}

\newcommand{\freev}[1]{\lrangle{#1}}
\newcommand{\boundv}[1]{(#1)}

\newcommand{\inact}{\mathbf{0}}

\newcommand{\Par}{\;|\;}
\newcommand{\news}[1]{(\nu\, #1)}
\newcommand{\newsp}[2]{(\nu\, #1)(#2)}

\newcommand{\varp}[1]{#1}
\newcommand{\rvar}[1]{#1}
\newcommand{\recp}[2]{\mu \rvar{#1}. #2}

\newcommand{\fn}[1]{\mathtt{fn}(#1)}
\newcommand{\rfn}[1]{\mathtt{rn}(#1)}

\newcommand{\fv}[1]{\mathtt{fv}(#1)}

\newcommand{\fs}[1]{\mathtt{fs}(#1)}

\newcommand{\scong}{\equiv}

\newcommand{\red}{\longrightarrow}

\def\subst#1#2{\{\raisebox{.5ex}{\small$#1$}\! / \mbox{\small$#2$}\}}

\newcommand{\vart}[1]{\mathsf{#1}}

\newcommand{\shsep}{.}
\newcommand{\outses}{!}
\newcommand{\inpses}{?}
\newcommand{\selses}{\triangleleft}
\newcommand{\brases}{\triangleright}
\newcommand{\dual}[1]{\overline{#1}}
\newcommand{\cat}{,}

\newcommand{\bout}[2]{#1 \outses \freev{#2} \shsep}
\newcommand{\about}[2]{#1 \outses \freev{#2}}

\newcommand{\bbout}[2]{#1 \outses \blrangle{#2} \shsep}
\newcommand{\abbout}[2]{#1 \outses \blrangle{#2}}

\newcommand{\abinp}[2]{#1 \inpses \boundv{#2}}
\newcommand{\binp}[2]{#1 \inpses \boundv{#2} \shsep}
\newcommand{\bsel}[2]{#1 \selses #2 \shsep}

\newcommand{\bbra}[2]{#1 \brases \set{#2}}

\newcommand{\tfont}[1]{\mathtt{#1}}
\newcommand{\tsep}{;}

\newcommand{\chtype}[1]{\lrangle{#1}}

\newcommand{\outtype}{\outses}
\newcommand{\inptype}{\inpses}

\newcommand{\trec}[2]{\mu\vart{#1}.#2}
\newcommand{\tvar}[1]{\vart{#1}}

\newcommand{\tinact}{\tfont{end}}

\newcommand{\btout}[1]{\outtype \lrangle{#1} \tsep}

\newcommand{\btinp}[1]{\inptype (#1) \tsep}

\newcommand{\btsel}[1]{\oplus \set{#1}}

\newcommand{\btbra}[1]{\& \set{#1}}

\newcommand{\proves}{\vdash}
\newcommand{\hastype}{\triangleright}

\makeatletter
\newcommand{\xMapsto}[2][]{\ext@arrow 0599{\Mapstofill@}{#1}{#2}}
\def\Mapstofill@{\arrowfill@{\Mapstochar\Relbar}\Relbar\Rightarrow}
\makeatother

\newcommand{\map}[1]{\ensuremath{\llbracket#1\rrbracket}}
\newcommand{\dmap}[1]{\ensuremath{\{\!|#1|\!\}}}

\newcommand{\dualof}{\ \mathsf{dual}\ }

\newcommand{\HOp}{\ensuremath{\mathsf{HO}\pi}\xspace}
\newcommand{\HO}{\ensuremath{\mathsf{HO}}\xspace}

\newcommand{\Proc}{\ensuremath{\diamond}}

\newcommand{\appl}[2]{#1\, {#2}}
\newcommand{\abs}[2]{\lambda #1.\,#2}

\newcommand{\lollipop}{\multimap}
\newcommand{\sharedop}{\rightarrow}

\newcommand{\lhot}[1]{#1\! \lollipop\! \diamond}
\newcommand{\shot}[1]{#1\! \sharedop\! \diamond}
\newcommand{\lhotup}[1]{#1^{\lollipop}}
\newcommand{\shotup}[1]{#1^{\sharedop}}

\newtheorem{notation}{Notation}

\newif\ifny\nyfalse

\newif\ifdm\dmtrue

\newif\ifrhu\rhutrue

\newif\ifjp\jptrue

\newif\ifjp\jptrue

\newcommand{\newj}[1]{{#1}}
\newcommand{\newjb}[1]{{#1}}

\newcommand{\AT}[2]{#1 \! : \! #2}

\newcommand{\secref}[1]{\S\,\ref{#1}}
\newcommand{\defref}[1]{Def.~\ref{#1}}

\newcommand{\figref}[1]{Fig.~\ref{#1}}
\newcommand{\thmref}[1]{Thm.~\ref{#1}}

\newcommand{\exref}[1]{Exam.~\ref{#1}}

\newcommand{\lemref}[1]{Lem.~\ref{#1}}
\newcommand{\remref}[1]{Rem.~\ref{#1}}

\newcommand{\tabref}[1]{Tab.~\ref{#1}}

\definecolor{lightgray}{gray}{0.75}

\makeatletter
\newcommand*{\rom}[1]{({\expandafter\romannumeral #1})}
\makeatother

\newcommand{\defas}{\triangleq}

\newcommand{\B}[3]{\mathcal{B}^{#1}_{#2}\big(#3\big)}
\newcommand{\mB}[3]{\mathsf{B}^{#1}_{#2}\big(#3\big)}
\newcommand{\V}[3]{\mathcal{V}^{#1}_{#2}\big(#3\big)}
\newcommand{\mV}[3]{\mathsf{V}^{#1}_{#2}\big(#3\big)}
\newcommand{\Gt}[1]{\mathcal{G}(#1)}
\newcommand{\Rt}[1]{\mathcal{R}(#1)}
\newcommand{\Rts}[3]{\mathcal{R}^{\star}(#3)}
\newcommand{\envR}{\Delta_\mu}
\newcommand{\envPropR}{\Phi}
\newcommand{\D}[1]{\mathcal{D}(#1)}
\newcommand{\mD}[1]{\mathsf{D}(#1)}
\newcommand{\len}[1]{|#1|}
\newcommand{\incrname}[2]{\subst{#1_{#2+1}}{#1_#2}}

\newcommand{\prop}{c}
\newcommand{\apropinp}[1]{\abinp{\prop_{#1}}}
\newcommand{\propinp}[1]{\binp{\prop_{#1}}}

\newcommand{\propout}[1]{\bout{\dual{\prop_{#1}}}}
\newcommand{\apropout}[1]{\about{\dual{\prop_{#1}}}}
\newcommand{\proprinp}[1]{\binp{\prop^{#1}}}
\newcommand{\proprout}[1]{\bout{\prop^{#1}}}

\newcommand{\propbout}[1]{\bbout{\dual{\prop_{#1}}}}
\newcommand{\apropbout}[1]{\abbout{\dual{\prop_{#1}}}}
\newcommand{\slhotup}[1]{{#1}^{\leadsto}}
\newcommand{\slhot}[1]{{#1} \leadsto \diamond}
\newcommand{\degree}{l}
\newcommand{\thunk}[1]{\{\!\{#1\}\!\}}
\newcommand{\appthunk}[1]{\mathtt{run}\,{#1}}
\newcommand{\appendx}[1]{#1}

\newcommand{\mugamma}{\gamma}
\newcommand{\misty}{\textsf{MISTY}\xspace}

\EventEditors{Alastair F. Donaldson}
\EventNoEds{1}
\EventLongTitle{33rd European Conference on Object-Oriented Programming (ECOOP 2019)}
\EventShortTitle{ECOOP 2019}
\EventAcronym{ECOOP}
\EventYear{2019}
\EventDate{July 15--19, 2019}
\EventLocation{London, United Kingdom}
\EventLogo{}
\SeriesVolume{134}
\ArticleNo{11}

\title{Minimal Session Types}

\titlerunning{Minimal Session Types}

\author{Alen Arslanagi\'{c}}{University of Groningen, The Netherlands}{}{https://orcid.org/0000-0002-0292-478X}{}
\author{Jorge A. P\'{e}rez}{University of Groningen, The Netherlands}{}{https://orcid.org/0000-0002-1452-6180}{}
\author{Erik Voogd}{University of Groningen, The Netherlands}{}{}{}

\authorrunning{A.\,Arslanagi\'{c}, J.\,A.\,P\'{e}rez, and E.\,Voogd}

\Copyright{A.\,Arslanagi\'{c}, J.\,A. P\'{e}rez, and E.\,Voogd}



\funding{Work partially supported by the Netherlands Organization for Scientific Research (NWO) under the VIDI Project No.\ 016.Vidi.189.046 (Unifying Correctness for Communicating Software).}

\acknowledgements{
We are grateful to the anonymous reviewers for their remarks and questions. 
P\'{e}rez is also with CWI, Amsterdam and the NOVA Laboratory for Computer Science and Informatics (FCT grant NOVA LINCS PEst/UID/CEC/04516/2013), Universidade Nova de Lisboa, Portugal.}

\nolinenumbers 

\begin{document}

\maketitle
\begin{abstract}
Session types are a type-based approach to the verification of message-passing programs.
They have been
much studied
as type systems for the $\pi$-calculus and for languages such as Java.
A session type specifies what and when should be exchanged through a channel.
Central to session-typed languages are constructs
in types and processes that specify \emph{sequencing} in
 protocols.

Here we study \emph{minimal session types}, session types without sequencing.
This is arguably the simplest form of session types.
By relying on a core process calculus with sessions and higher-order concurrency (abstraction-passing),
we prove that every process typable with usual (non minimal) session types can be compiled down into a process typed
with minimal session types. This means that having sequencing constructs in both processes and session types is redundant; only
sequentiality in processes is indispensable, as it can precisely codify sequentiality in types.

Our developments draw inspiration from work by Parrow on behavior-preserving decompositions of untyped processes.
By casting Parrow's results in the realm of typed processes, our results reveal a conceptually simple formulation of session types and a
principled avenue to the integration of session types into languages without sequencing in types.
\end{abstract}

\section{Introduction}
\label{s:intro}
Session types are a type-based approach to the verification of
message-passing programs. A session type specifies what and when should be exchanged through a channel;
this makes them a useful tool to enforce safety and liveness properties related to communication correctness.
Session types have had a significant impact on the foundations of programming languages~\cite{DBLP:journals/csur/HuttelLVCCDMPRT16}, but also on their practice~\cite{DBLP:journals/ftpl/AnconaBB0CDGGGH16}. In particular, the interplay of session types and
 object-oriented languages has received much attention (cf.~\cite{DBLP:conf/ecoop/Dezani-CiancagliniMYD06,DBLP:conf/fmco/Dezani-CiancagliniGDY06,DBLP:conf/ecoop/HuYH08,DBLP:conf/popl/GayVRGC10,DBLP:conf/agere/BalzerP15,DBLP:conf/ppdp/KouzapasDPG16,DBLP:conf/ecoop/ScalasY16}).
 In this work, our goal is to understand to what extent session types can admit simpler, more fundamental formulations. This foundational question has concrete practical ramifications, as we discuss next.


In session-typed languages,   \emph{sequencing} constructs in types and processes specify the intended structure of  message-passing protocols.
In the session type $S = \btinp{\mathsf{Int}} \btinp{\mathsf{Int}} \btout{\mathsf{Bool}} \tinact$,
sequencing (denoted `$;$') allows us to
specify a protocol for a channel that \emph{first} receives (?) two integers, \emph{then} sends (!) a Boolean, and \emph{finally} ends.
As such, $S$ could type a service that checks for integer equality.
Sequencing in types goes hand-in-hand with sequencing in processes, which is specified using   prefix constructs (denoted~`$.$').
The $\pi$-calculus process $P =  \binp{s}{x_1}  \binp{s}{x_2} \bout{s}{b} \inact$
 is an implementation of the equality service:
it \emph{first} expects two values on name $s$, \emph{then} outputs a Boolean on $s$, and \emph{finally} stops.
Thus, name $s$ in $P$ conforms to  type~$S$.
Session types
can also specify sequencing within labeled choices and recursion; these typed constructs are also in close match with their respective process expressions.

Originally developed on top of the $\pi$-calculus
for the analysis  of message-passing protocols between exactly two parties~\cite{honda.vasconcelos.kubo:language-primitives},
  session types have been extended in many directions.
  We find, for instance,
  multiparty session types~\cite{HYC08} and
  extensions with dependent types, assertions, exceptions, and time (cf.~\cite{DL10,DBLP:journals/csur/HuttelLVCCDMPRT16} for surveys).
All these extensions seek to address natural research questions on the expressivity and applicability of session types theories.

Here we address a different, if opposite, question: \emph{is there a minimal formulation of session types?}
This is an appealing question from a theoretical perspective, but seems particularly relevant to the practice of session types:
identifying the ``core'' of session types could enable their integration in languages
whose type systems do not have
advanced
constructs 
present in session types (such as sequencing).
For instance, the Go programming language offers  primitive support for message-passing concurrency; it comes with a static verification mechanism which can only enforce that messages exchanged along channels correspond with their declared payload types---it cannot ensure essential correctness properties associated to the structure of protocols.
This observation has motivated the development of advanced static verification tools based on session types for Go programs~\cite{DBLP:conf/cc/NgY16,DBLP:conf/icse/LangeNTY18}.


This paper identifies an elementary formulation of session types and studies its properties. We call them \emph{minimal session types}: these are session types without sequencing. That is, in
session types such as `$\btout{U} S$' and `$\btinp{U} S$', we decree that $S$ can only correspond to $\tinact$, the type of the terminated protocol.

Adopting this elementary formulation entails dispensing with sequencing, which is one of the most distinctive features of session types.
While this may appear as a far too drastic restriction, it turns out that it is not:
our main result is that for every process $P$ that is well-typed under standard (non minimal) session types, there is a  \emph{process decomposition} $\D{P}$ that is well-typed using minimal session types.
Intuitively,  $\D{P}$ codifies the sequencing information given by the session types (protocols) of $P$ using additional synchronizations.
This shows that having sequencing in both types and processes is redundant; only sequencing at the level of processes is truly fundamental.
To define $\D{P}$ we draw inspiration from a known result by Parrow~\cite{DBLP:conf/birthday/Parrow00}, who proved that untyped $\pi$-calculus processes can be decomposed as a collection of \emph{trios processes}, i.e., processes with at most three nested prefixes~\cite{DBLP:conf/birthday/Parrow00}.

The question of how to relate session types with other type systems has attracted interest in the past.
Session types have been encoded into generic types~\cite{DBLP:journals/corr/GayGR14} and linear types~\cite{DBLP:conf/concur/DemangeonH11,DBLP:conf/ppdp/DardhaGS12,DBLP:journals/iandc/DardhaGS17}.
As such, these prior studies concern the \emph{relative expressiveness} of session types: where the expressivity of session types stands with respect to that of some other type system.
In sharp contrast, we study the \emph{absolute expressiveness} of session types: how session types can be explained in terms of themselves.
To our knowledge, this is the first study of its kind. 

The process language that we consider for decomposition into minimal session types is \HO, the core process calculus for session-based concurrency studied by Kouzapas et al.~\cite{DBLP:conf/esop/KouzapasPY16,KouzapasPY17}.
\HO is a very small language: it supports abstraction-passing only and lacks name-passing and recursion; still, it is also very expressive, because both features can be expressed in it in a fully abstract way.
As such, \HO is an excellent candidate for a decomposition. Being a higher-order language, \HO is very different from the (untyped, first-order) $\pi$-calculus considered by Parrow in~\cite{DBLP:conf/birthday/Parrow00}. 
Also, the session types of \HO severely constrain the range and kind of conceivable decompositions.
Therefore, our results are not an expected consequence of Parrow's: essential aspects of our decomposition into processes typable with minimal session types are only possible in a higher-order setting, not considered in~\cite{DBLP:conf/birthday/Parrow00}.

\smallskip

Summing up, in this paper we make the following contributions:
\begin{enumerate}
\item We identify the class of \emph{minimal session types} as a simple fragment of standard session types that retains its absolute expressiveness.
\item We show how to decompose processes typable with standard session types into processes typable with minimal session types. We prove that this decomposition satisfies a typability result for a rich typed language that includes labeled choices and recursive types.
\item We develop optimizations of our decomposition that bear witness to its robustness.
\end{enumerate}

\noindent
The rest of the paper is organized as follows.
\secref{s:lang} summarizes the syntax, semantics, and session type system for \HO, the core process calculus for session-based concurrency.
\secref{s:decomp} presents the decomposition of well-typed \HO processes into minimal session types. The decomposition is presented incrementally, starting with a core fragment that is later extended with further features.
\secref{s:opt} presents optimizations of the decomposition.
\secref{s:rw} elaborates further on related works and 
\secref{s:concl} concludes.
\appendx{The appendix contains omitted definitions and proofs.}

\section{The Source Language}
\label{s:lang}
We recall the syntax, semantics, and type system for \HO, the higher-order process calculus for session-based concurrency studied by Kouzapas et al.~\cite{DBLP:conf/esop/KouzapasPY16,KouzapasPY17}.\footnote{We summarize the content from~\cite{DBLP:conf/esop/KouzapasPY16,KouzapasPY17} that concerns \HO; the notions and results given in~\cite{DBLP:conf/esop/KouzapasPY16,KouzapasPY17} are given for \HOp, a super-calculus of \HO.}
\HO is arguably the simplest language for session types: it supports
passing of abstractions (functions from names to processes)
but does not support name-passing nor process recursion.
Still, \HO is very expressive:  it can encode name-passing, recursion, and polyadic communication via type-preserving encodings that are fully-abstract with respect to  contextual equivalence~\cite{DBLP:conf/esop/KouzapasPY16}.

\subsection{Syntax and Semantics}
The syntax of names, variables, values, and \HO processes  is defined as follows:
		\begin{align*}
			n, m  & \bnfis a,b \bnfbar s, \dual{s}
			\qquad
			u,w   \bnfis n \bnfbar x,y,z
			\qquad
			V,W   \bnfis {x,y,z} \bnfbar {\abs{x}{P}}
			\\[1mm]
			P,Q
			 & \bnfis
			\bout{u}{V}{P}  \sbnfbar  \binp{u}{x}{P} \sbnfbar
			\bsel{u}{l} P \sbnfbar \bbra{u}{l_i:P_i}_{i \in I} \sbnfbar {\appl{V}{u}}
			 \sbnfbar P\Par Q \sbnfbar \news{n} P
			\sbnfbar \inact
		\end{align*}

\noindent
We use
$a,b,c, \dots$ to
range over \emph{shared  names},
and
$s, \dual{s}, \dots$
to range over \emph{session names}.
Shared names are used for unrestricted, non-deterministic interactions;
session names are used for linear, deterministic interactions.
We write $n, m$ to denote session or shared names,
and assume that the sets of session and shared names are disjoint.
The \emph{dual} of   $n$ is denoted $\dual{n}$; we define
$\dual{\dual{s}} = s$ and $\dual{a} = a$, i.e., duality is only relevant for session names.
Variables are denoted with $x, y, z, \dots$.
An abstraction 
$\abs{x}{P}$ is a process $P$ with parameter $x$.
 \emph{Values} $V,W, \ldots$ include
variables 
and
abstractions,  
but not names.
A tuple of variables $(x_1, \ldots, x_k)$ is denoted $\widetilde{x}$ (and similarly for names and values). 
We use $\epsilon$ to denote the empty tuple.


Processes $P, Q, \ldots$
include usual $\pi$-calculus
output and input prefixes,
denoted $\bout{u}{V}{P}$ and $\binp{u}{x}{P}$, respectively.
Processes $\bsel{u}{l} P$ and $\bbra{u}{l_i: P_i}_{i \in I}$ are selecting and branching constructs, respectively, commonly used in session calculi to express deterministic choices~\cite{honda.vasconcelos.kubo:language-primitives}.
Process
$\appl{V}{u}$
is the application
which substitutes name $u$ on abstraction~$V$.
Constructs for
inaction $\inact$,  parallel composition $P_1 \Par P_2$, and
name restriction $\news{n} P$ are standard.
\HO lacks name-passing and recursion, but they are expressible in the language (see \exref{ex:np} below).

We sometimes omit trailing $\inact$'s, so we may write, e.g., 
$\about{u}{V}{}$ instead of  $\bout{u}{V}{\inact}$.
Also, we write 
$\bout{u}{}{P}$ and $\binp{u}{}{P}$ whenever the exchanged value is not relevant (cf. \remref{r:prefix}).

Session name restriction $\news{s} P$ simultaneously binds session names $s$ and $\dual{s}$ in $P$.
Functions $\fv{P}$, $\fn{P}$, and $\fs{P}$ denote, respectively, the sets of free
variables, names, and session names in $P$, and are defined as expected.
If $\fv{P} = \emptyset$, we call $P$ {\em closed}.
 We write $P\subst{u}{y}$ (resp.,\,$P\subst{V}{y}$) for the capture-avoiding substitution of name $u$ (resp.,\, value $V$) for $y$ in process $P$.
We identify processes up to consistent renaming of bound names, writing $\scong_\alpha$ for this congruence.
We shall rely on Barendregt's variable convention, which ensures that free and bound names are different in every mathematical context.

The  operational semantics of \HO is defined in terms of a \emph{reduction relation},
denoted $\red$.
Reduction is closed under \emph{structural congruence},
denoted $\scong$, which is defined as the smallest congruence on processes such that:
	\begin{align*}
		P \Par \inact \scong P
		\quad
		P_1 \Par P_2 \scong P_2 \Par P_1
		\quad
		P_1 \Par (P_2 \Par P_3) \scong (P_1 \Par P_2) \Par P_3
		\quad
		\news{n} \inact \scong \inact
		\\[1mm]
		P \Par \news{n} Q \scong \news{n}(P \Par Q)
		\ (n \notin \fn{P})
		\quad
		P \scong Q \textrm{ if } P \scong_\alpha Q
	\end{align*}
	We assume the expected extension of $\scong$ to values $V$.
The reduction relation expresses the behavior of processes; it is defined as follows:
	\begin{align*}
		\appl{(\abs{x}{P})}{u}   & \red  P \subst{u}{x}
		& \orule{App} &
		\\
		\bout{n}{V} P \Par \binp{\dual{n}}{x} Q & \red  P \Par Q \subst{V}{x}
		& \orule{Pass} &
		\\
		\bsel{n}{l_j} Q \Par \bbra{\dual{n}}{l_i : P_i}_{i \in I} & \red Q \Par P_j
~~(j \in I)~~
		& \orule{Sel} &
		\\
		P \red P'\Rightarrow  \news{n} P   & \red    \news{n} P'
		& \orule{Res} &
		\\
		P \red P'  \Rightarrow    P \Par Q  & \red   P' \Par Q
		& \orule{Par} &
		\\
		P \scong Q \red Q' \scong P'  \Rightarrow  P  & \red  P'
		& \orule{Cong} &
	\end{align*}
 Rule $\orule{App}$ defines  name application.
Rule $\orule{Pass}$ defines a shared or session interaction, depending on the nature of $n$.
Rule $\orule{Sel}$ is the standard rule for labelled choice/selection. 
Other rules are standard $\pi$-calculus rules.
We write 
$\red^k$ for a $k$-step reduction, and 
$\red^\ast$ for the reflexive, transitive closure of $\red$.

We illustrate \HO processes and their semantics by means of an example.

\begin{example}[Encoding Name-Passing]
\label{ex:np}
\HO lacks name-passing, and so the reduction
\begin{align}
\bout{n}{m} P \Par \binp{\dual{n}}{x} Q  \red  P \Par Q \subst{m}{x}
\label{red:ex}
\end{align}
is not supported by the language. Still, as explained in~\cite{DBLP:conf/esop/KouzapasPY16}, name-passing can be encoded in a fully-abstract way using abstraction-passing, by
``packing'' the name $m$ in an abstraction. Let $\map{\cdot}$ be the encoding defined as 
\begin{align*}
  \map{\bout{n}{m} P}	&= \bbout{n}{ \abs{z}{\,\binp{z}{x} (\appl{x}{m})} } \map{P} \\
  \map{\binp{n}{x} Q}	&=	 \binp{n}{y} \newsp{s}{\appl{y}{s} \Par \about{\dual{s}}{\abs{x}{\map{Q}}}}
\end{align*}
and as an homomorphism for the other constructs.
Reduction \eqref{red:ex} can be mimicked as
\begin{align*}
\map{\bout{n}{m} P \Par \binp{\dual{n}}{x} Q} =~~ & \bbout{n}{ \abs{z}{\,\binp{z}{x} (\appl{x}{m})} } \map{P}  \Par \binp{n}{y} \newsp{s}{\appl{y}{s} \Par \about{\dual{s}}{\abs{x}{\map{Q}}}}
 \\
 \red~~ &  \map{P}  \Par \newsp{s}{\appl{\abs{z}{\,\binp{z}{x} (\appl{x}{m})}}{s} \Par \about{\dual{s}}{\abs{x}{\map{Q}}}}
 \\
 \red~~ &  \map{P}  \Par \newsp{s}{\binp{s}{x} (\appl{x}{m}) \Par \about{\dual{s}}{\abs{x}{\map{Q}}}}
 \\
 \red~~ &  \map{P}  \Par \appl{(\abs{x}{\map{Q}})}{m}
 \\
 \red~~ &  \map{P}  \Par \map{Q}\subst{m}{x}
\end{align*}\hspace*{\fill} $\lhd$
\end{example}

\begin{remark}[Polyadic Communication]
\label{r:poly}
\HO as presented above allows only for \emph{monadic communication}, i.e., the exchange of tuples of values with length 1.
We will find it convenient to use \HO with \emph{polyadic communication}, i.e., the exchange of tuples of values $\widetilde{V}$, with length $k \geq 1$.
In \HO, polyadicity appears in session synchronizations and applications, but not in synchronizations on shared names. This entails having the following reduction rules:
	\begin{align*}
		\appl{(\abs{\widetilde{x}}{P})}{\widetilde{u}}   & \red  P \subst{\widetilde{u}}{\widetilde{x}}
		\\
		\bout{s}{\widetilde{V}} P \Par \binp{\dual{s}}{\widetilde{x}} Q & \red  P \Par Q \subst{\widetilde{V}}{\widetilde{x}}
\end{align*}
where the simultaneous substitutions
 $P\subst{\widetilde{u}}{\widetilde{x}}$
 and $P\subst{\widetilde{V}}{\widetilde{x}}$
 are as expected.
This polyadic \HO can be readily encoded into (monadic) \HO~\cite{KouzapasPY17}; for this reason, by a slight
 abuse of notation we will often write \HO when we actually mean ``polyadic \HO''.
\end{remark}

\subsection{Session Types for \HO}
\label{sec:types}

We give essential definitions and properties for the session type system for \HO, following~\cite{KouzapasPY17}.

\begin{definition}[Session Types for \HO~\cite{KouzapasPY17}]
\label{d:types}
Let us write $\Proc$ to denote the process type.
The syntax of types for \HO is defined as follows: 
	\begin{align*}
U & \bnfis		\shot{C} \bnfbar \lhot{C}
		\\
		C  & \bnfis		S \bnfbar  {\chtype{U}}
\\
		S & \bnfis 	\tinact  \bnfbar  \btout{U} S \bnfbar \btinp{U} S 
		 \bnfbar \btsel{l_i:S_i}_{i \in I} \bnfbar \btbra{l_i:S_i}_{i \in I}
		 \bnfbar  \trec{t}{S} \bnfbar \vart{t}
	\end{align*}
\end{definition}
\noindent
Value types $U$ include $\shot{C}$ and $\lhot{C}$, which denote {\em shared} and {\em linear} higher-order types, respectively.
Shared channel types are denoted $\chtype{S}$ and $\chtype{U}$.
Session types, denoted by $S$, follow the standard binary session type syntax~\cite{honda.vasconcelos.kubo:language-primitives}.
Type $\tinact$ is the termination type.
The {\em output type}
$\btout{U} S$ 
first sends a value of type $U$ and then follows the type described by $S$.
Dually, $\btinp{U} S$ denotes an {\em input type}.
The {\em branching type} $\btbra{l_i:S_i}_{i \in I}$ and the {\em selection type} $\btsel{l_i:S_i}_{i \in I}$ are used to type the branching and selection constructs that define the labeled choice.
We assume the {\em recursive type} $\trec{t}{S}$ is guarded,
i.e.,  type $\trec{t}{\vart{t}}$ is not allowed.

In session types theories \emph{duality} is a key notion:
implementations derived from dual session types will respect their protocols at run-time, avoiding communication errors.
Intuitively, duality is  obtained by
exchanging $!$ by $?$ (and vice versa) and  $\oplus$ by $\&$ (and vice versa),
including the fixed point construction.
We write $S \dualof T$ if session types $S$ and $T$ are dual according to this intuition;
the formal definition is coinductive, and given in~\cite{KouzapasPY17}.

We consider shared, linear, and session
\emph{environments}, denoted $\Gamma$, $\Lambda$, and $\Delta$, \newjb{resp.}:
\begin{align*}
		\Gamma  & \bnfis  \emptyset \bnfbar \Gamma \cat \varp{x}: \shot{C} 
		\bnfbar \Gamma \cat u: \chtype{U}
\qquad
		\Lambda  \bnfis  \emptyset \bnfbar \Lambda \cat \AT{x}{\lhot{C}}
		 \\
		\Delta    & \bnfis   \emptyset \bnfbar \Delta \cat \AT{u}{S}
\end{align*}
$\Gamma$ maps variables and shared names to value types; 
$\Lambda$ maps variables to
linear
higher-order
types.
 $\Delta$  maps
session names to session types.
While $\Gamma$ admits weakening, contraction, and exchange principles,
both $\Lambda$ and $\Delta$
are
only subject to exchange.
The domains of $\Gamma,
\Lambda$, and $\Delta$ are assumed pairwise distinct.
$\Delta_1\cdot \Delta_2$ is the disjoint union of $\Delta_1$ and $\Delta_2$.

We write $\Gamma\backslash x$ to denote
$\Gamma\backslash \{x:C\}$, i.e.,
the environment obtained from $\Gamma$ by removing the assignment $x : \shot{C}$, for some $C$.
Notations
$\Delta\backslash u$
and
$\Gamma\backslash \widetilde{x}$ will have expected readings.
With a slight abuse of notation, given a tuple of variables $\widetilde x$,
we sometimes write $(\Gamma, \Delta)(\widetilde x)$ to denote the tuple of types assigned to variables in $\widetilde x$.

The typing judgements for values $V$ and processes $P$ are denoted
\[
\Gamma; \Lambda; \Delta \proves V \hastype U \qquad \text{and} \qquad\Gamma; \Lambda; \Delta \proves P \hastype \Proc
\]

\begin{figure}[t]
\[
	\begin{array}{c}
		\inferrule[(Prom)]{
			\Gamma; \emptyset; \emptyset \proves V \hastype
                         \lhot{C}
		}{
			\Gamma; \emptyset; \emptyset \proves V \hastype
                         \shot{C}
		}
		\quad
		\inferrule[(EProm)]{
		\Gamma; \Lambda \cat x : \lhot{C}; \Delta \proves P \hastype \Proc
		}{
			\Gamma \cat x:\shot{C}; \Lambda; \Delta \proves P \hastype \Proc
		} \quad 
		\inferrule[(Abs)]{
			\Gamma; \Lambda; \Delta_1 \proves P \hastype \Proc
			\quad
			\Gamma; \es; \Delta_2 \proves x \hastype C
		}{
			\Gamma\backslash x; \Lambda; \Delta_1 \backslash \Delta_2 \proves \abs{{x}}{P} \hastype \lhot{{C}}
		}
		\quad
		\\[4mm]
		\inferrule[(App)]{
				\Gamma; \Lambda; \Delta_1 \proves V \hastype C \leadsto \diamond \quad
				\leadsto \in \{\multimap,\rightarrow \}
				\quad
				\Gamma; \es; \Delta_2 \proves u \hastype C
		}{
			\Gamma; \Lambda; \Delta_1 \cat \Delta_2 \proves \appl{V}{u} \hastype \Proc
		}
		\\[4mm]
		\inferrule*[left=(Send)]{
					u:S \in \Delta_1 \cat \Delta_2 
					\quad
			\Gamma; \Lambda_1; \Delta_1 \proves P \hastype \Proc
			\quad
			\Gamma; \Lambda_2; \Delta_2 \proves V \hastype U
		}{
			\Gamma; \Lambda_1 \cat \Lambda_2; ((\Delta_1 \cat \Delta_2) \setminus u:S) \cat u:\btout{U} S \proves \bout{u}{V} P \hastype \Proc
		}
		\\[4mm]
		\inferrule*[left=(Rcv)]{
		\begin{array}{c}
			\Gamma; \Lambda_1; \Delta \cat u: S \proves P \hastype \Proc
			\quad
			\Gamma; \Lambda_2; \es \proves {x} \hastype {U}
			\end{array}
		}{
			\Gamma \backslash x; \Lambda_1\backslash \Lambda_2; \Delta \cat u: \btinp{U} S \vdash \binp{u}{{x}} P \hastype \Proc
		}
		\\[2mm] 
		\inferrule[(Req)]{
			\begin{array}{c}
				\Gamma; \es; \es \proves u \hastype \chtype{U}
				\quad
				\Gamma; \Lambda; \Delta_1 \proves P \hastype \Proc
				\\
				\Gamma; \es; \Delta_2 \proves V \hastype U
			\end{array}
		}{
			\Gamma; \Lambda; \Delta_1 \cat \Delta_2 \proves \bout{u}{V} P \hastype \Proc
		}
		~~
		\inferrule[(Acc)]{
			\begin{array}{c}
				\Gamma; \emptyset; \emptyset \proves u \hastype \chtype{U}
				\quad
				\Gamma; \Lambda_1; \Delta \proves P \hastype \Proc
				\\
				\Gamma; \Lambda_2; \es \proves x \hastype U\\
	               \end{array}
		}{
			\Gamma\backslash x; \Lambda_1 \backslash \Lambda_2; \Delta \proves \binp{u}{x} P \hastype \Proc
		}
		\end{array}
\]
\caption{Selected Typing Rules for $\HO$.
See~\cite{KouzapasPY17} for a full account.
\label{fig:typerulesmys}}
\end{figure}

\noindent
\figref{fig:typerulesmys} shows
selected typing rules;
see~\cite{KouzapasPY17}
for a full account.
The shared type $\shot{C}$ 
is derived using Rule~\textsc{(Prom)} only
if the value has a linear type with an empty linear
environment.
Rule~\textsc{(EProm)} allows us to freely use a \newj{shared
type variable as linear}.
Abstraction values are typed with Rule~\textsc{(Abs)}.
Application typing
is governed by Rule~\textsc{(App)}:  
the type $C$ of an application name $u$
must match the type of the application variable $x$ ($\lhot{C}$ or $\shot{C}$).
%
In Rule~\textsc{(Send)},
the type $U$ of value $V$ should appear as a prefix in the session type $\btout{U} S$ of $u$.
Rule~\textsc{(Rcv)} is its dual.
Rules~\textsc{(Req)} and~\textsc{(Acc)} type interaction along shared names;
the type of the sent/received object $V$
(i.e., $U$) should
match the type of the subject $s$
($\chtype{U}$).



To state type soundness, we require two auxiliary definitions on session environments.
First, a session environment $\Delta$ is {\em balanced} (written $\balan{\Delta}$) if whenever $s: S_1, \dual{s}: S_2 \in \Delta$ then $S_1 \dualof S_2$.
Second, we define the reduction relation $\red$ on session environments as:
	\begin{eqnarray*}
			\Delta \cat s: \btout{U} S_1 \cat \dual{s}: \btinp{U} S_2 & \red &
			\Delta \cat s: S_1 \cat \dual{s}: S_2\\
			\Delta \cat s: \btsel{l_i: S_i}_{i \in I} \cat \dual{s}: \btbra{l_i: S_i'}_{i \in I} &\red&
			 \Delta \cat s: S_k \cat \dual{s}: S_k' \ (k \in I)
		\end{eqnarray*}


\begin{theorem}[Type Soundness~\cite{KouzapasPY17}]\label{t:sr}\rm
%
			Suppose $\Gamma; \es; \Delta \proves P \hastype \Proc$
			with
			$\balan{\Delta}$.
			Then $P \red P'$ implies $\Gamma; \es; \Delta'  \proves P' \hastype \Proc$
			and $\Delta = \Delta'$ or $\Delta \red \Delta'$
			with $\balan{\Delta'}$.
\end{theorem}

\begin{remark}[Typed Polyadic Communication]
When using processes with polyadic communication (cf.~\remref{r:poly}),  
we shall assume the extension of the type system defined in~\cite{KouzapasPY17}.
\end{remark}

\begin{notation}[Type Annotations]
We shall often annotate bound names and variables with their respective type.
We will write, e.g., $\news{s:S}P$ to denote that the type of $s$ in $P$ is $S$.
Similarly for values: we shall write $\abs{u:C}{P}$.
Also, letting $\leadsto \in \{\lollipop, \sharedop\}$, we may write
$\abs{u:\slhotup{C}}{P}$ to denote that the value is linear (if $\leadsto = \lollipop$) or shared (if $\leadsto = \sharedop$).
That is, we write $\abs{u:\slhotup{C}}{P}$ if
$\Gamma ; \Lambda; \Delta \proves \abs{{u}}{P} \hastype \slhot{{C}}$, for some $\Gamma$, $\Lambda$, and $\Delta$.
\end{notation}

Having introduced the core session process language \HO, we now move to detail its type-preserving decomposition into minimal session types.

\section{Decomposing Session-Typed Processes}
\label{s:decomp}

\subsection{Key Ideas}
Our goal is to transform an \HO process $P$, typable with the session types in \defref{d:types}, into
another \HO process, denoted $\D{P}$, typable using \emph{minimal session types} (cf. \defref{d:mtypesi} below).
By means of this transformation on processes, which we call a \emph{decomposition}, the sequencing in session types for $P$ is codified in $\D{P}$ by using additional actions.
To ensure that this transformation on $P$ is sound, we must also decompose its session types;
our main result says that if $P$ is well-typed under session types $S_1, \ldots, S_n$,
then $\D{P}$ is typable using the minimal session types $\Gt{S_1}, \ldots, \Gt{S_n}$, where $\Gt{\cdot}$ is a decomposition function that ``slices'' a session type (as in \defref{d:types}) into a \emph{list} of minimal session types (cf. \defref{def:typesdecomp} below).

To define  the decomposition $\D{P}$, in \defref{def:decomp} we rely on a
\emph{breakdown function} that translates $P$ into a composition of \emph{trios processes} (or simply \emph{trios}).
A {trio} is a process with exactly three nested prefixes.
Roughly speaking, if $P$ is a sequential process with $k$ nested actions, then $\D{P}$ will contain $k$  trios running in parallel: each trio in $\D{P}$ will enact exactly one prefix from $P$; the breakdown function must be carefully designed to ensure that trios trigger each other in such a way that $\D{P}$ preserves the prefix sequencing in $P$.

We borrow from Parrow~\cite{DBLP:conf/birthday/Parrow00}
some useful terminology and notation on trios.
The \textit{context} of a trio is a tuple of variables $\widetilde x$, possibly empty, which  makes variable bindings explicit.
We use a reserved set of \textit{propagator names} (or simply \emph{propagators}), denoted with $\prop_k, \prop_{k+1}, \ldots$, to carry contexts and trigger the
subsequent trio.
A process with less than three sequential prefixes is called a \textit{degenerate trio}.
Also, a \emph{leading trio} is the one that receives a context, performs an action, and triggers the next trio; a \emph{control trio} only activates other trios.

The breakdown function works on both processes and values.
The breakdown of process $P$ is denoted by $\B{k}{\tilde x}{P}$, where $k$ is the
index for the propagators $\prop_k$, and $\widetilde x$ is the context to be received by the previous trio.
Similarly,
the breakdown  of a value $V$ is denoted by
$\V{k}{\tilde x}{V}$.

We present the decomposition of well-typed \HO processes (and its associated typability results) incrementally---this is useful to gradually illustrate our ideas and highlight the several ways in which our developments differ from Parrow's.
In \secref{ss:core}, we consider a ``core fragment'' of \HO, which contains output and input prefixes,
application, restriction, parallel composition, and inaction.
Hence, this fragment does not have labeled choice and recursion, nor  recursive types.
In \secref{ss:exti} we shall extend the decomposition functions with selection and branching;
an extension that supports names with recursive types is presented in \secref{ss:extii}.

\subsection{The Core Fragment}
\label{ss:core}
We present our approach for a core fragment of \HO. We start introducing some preliminary definitions, including the definition of breakdown function. Then we give our main result:
\thmref{t:decompcore} (Page~\pageref{t:decompcore}) asserts that if process $P$ is well-typed with standard session types, then  $\D{P}$ is well-typed with minimal session types.
This theorem relies crucially on \thmref{t:typecore} (Page~\pageref{t:typecore}), which specifies the way in which the breakdown function preserves typability.

\subsubsection{Preliminaries}

We start by introducing minimal session types as a fragment of \defref{d:types}:
\begin{definition}[Minimal Session Types]
\label{d:mtypesi}
The syntax of \emph{minimal session types} for \HO is defined as follows: 
	\begin{align*}
U & \bnfis		\shot{\widetilde{C}} \bnfbar \lhot{\widetilde{C}}
		\\
		C  & \bnfis		M  \bnfbar  {\chtype{U}}
\\
		M & \bnfis 	\tinact  \bnfbar  \btout{\widetilde{U}} \tinact \bnfbar \btinp{\widetilde{U}} \tinact
	\end{align*}
\end{definition}
Clearly, this minimal type structure induces a reduced set of typable \HO processes.
We shall implicitly assume a type system for \HO based on these minimal session types by considering the expected specializations of the notions, typing rules, and results summarized in \secref{sec:types}.

We now define how to ``slice'' a session type into a \emph{list} of
minimal session types.

\begin{definition}[Decomposing Session Types] \label{def:typesdecomp}
Let $S$ be a session type, $U$ be a higher-order type, $C$ be a name type, and
$\chtype{U}$ be a shared type, all as in \defref{d:types}.  The \emph{type decomposition
function}  $\Gt{\cdot}$ is defined as: \smallskip
\begin{align*}
  \Gt{\btout{U}{S}} &=
  \begin{cases}
  \btout{\Gt{U}}{\tinact}  &  \text{if $S = \tinact$} \\
  \btout{\Gt{U}}{\tinact}\, ,\Gt{S} & \text{otherwise}
  \end{cases} \\
  \Gt{\btinp{U}{S}} &=
  \begin{cases}
  \btinp{\Gt{U}}{\tinact}  &  \text{if $S = \tinact$} \\
  \btinp{\Gt{U}}{\tinact}\, , \Gt{S} & \text{otherwise}
  \end{cases} \\
    \Gt{\tinact} &= \tinact \\
  \Gt{\lhot{C}} &= \lhot{\Gt{C}} \\
  \Gt{\shot{C}} &= \shot{\Gt{C}} \\
  \Gt{\chtype{U}} &= \chtype{\Gt{U}}\\
  \Gt{S_1,\ldots,S_n} &=
  \Gt{S_1},\ldots,\Gt{S_n}
\end{align*}
\end{definition}
\noindent
Thus, intuitively, if a session type $S$ contains $k$ input/output actions, the list $\Gt{S}$ will contain $k$ minimal session types.
We write $\len{\Gt{S}}$ to denote the length of  $\Gt{S}$.

\begin{example}
Let $S = \btinp{\mathsf{Int}} \btinp{\mathsf{Int}} \btout{\mathsf{Bool}} \tinact$ be the session type given in \secref{s:intro}.
Then $\Gt{S}$ is the list of minimal session types given by $\btinp{\mathsf{Int}} \tinact \, ,  \btinp{\mathsf{Int}} \tinact \, , \btout{\mathsf{Bool}} \tinact$. 
 \hspace*{\fill} $\lhd$
\end{example}

\noindent
The breakdown function $\B{k}{\tilde x}{\cdot}$ will operate on processes with \emph{indexed} names (cf. \defref{d:counterinit}).
Indexes are relevant for session names: a name $s_i$ will execute the $i$-th  action in session $s$. For this reason, to extend the decomposition function $\Gt{\cdot}$ to typing environments, we consider names $u_i$ in $\Gamma$ and $\Delta$. 
To define the decomposition of environments, we rely on the following notation.
Given 
a tuple of names 
$\widetilde{s} = s_1, \ldots, s_n$ and 
a tuple of (session) types
$\widetilde{S} = S_1, \ldots, S_n$ of the same length, 
we write $\widetilde{s}:\widetilde{S}$ to denote a list of typing assignments $s_1:S_1, \ldots, s_n:S_n$.

\begin{definition}[Decomposition of Environments] \label{def:typesdecompenv}
  Let $\Gamma$, $\Lambda$, and $\Delta$ be typing environments.
  We define $\Gt{\Gamma}$, $\Gt{\Lambda}$, and $\Gt{\Delta}$ inductively as follows:
  \begin{align*}
        \Gt{\Delta \cat u_i:S} &= \Gt{\Delta},(u_i,\ldots,u_{i+\len{\Gt{S}}-1}) : \Gt{S} \\
      \Gt{\Gamma \cat u_i:\chtype{U}} &= \Gt{\Gamma} \cat u_i : \Gt{\chtype{U}} \\
    \Gt{\Gamma \cat x:U} &= \Gt{\Gamma} \cat x:\Gt{U} \\
    \Gt{\Lambda \cat x:U} &= \Gt{\Lambda} \cat x:\Gt{U} \\
\Gt{\emptyset} &=  \emptyset
  \end{align*}
\end{definition}

\noindent
In order to determine the required number of propagators ($\prop_k, \prop_{k+1}, \ldots$) required in the breakdown of processes and  values, we mutually define their \emph{degree}:
\begin{definition}[Degree of a Process and Value]
  \label{def:sizeproc}
	Let $P$ be an \HO process.
	The \emph{degree} of $P$, denoted $\len{P}$, is inductively defined as follows:
	$$
	\len{P} =
	\begin{cases}
	\len{V} + \len{Q} + 1 & \text{if $P = \bout{u_i}{V}Q$}
	\\
	\len{Q} + 1 & \text{if $P =\bout{u_i}{y}Q$ or $P=\binp{u_i}{y}Q$}
	\\
	\len{V} + 1 & \text{if $P = \appl{V}{u_i}$}
	\\
	\len{P'} & \text{if $P = \news{s:S}P'$}
	\\
	\len{Q} + \len{R} + 1 & \text{if $P = Q \Par R$}
	\\
	1 & \text{if $P = \appl{y}{u_i}$ or $P = \inact$}
	\end{cases}
	$$
%
	  The \emph{degree} of a value $V$,
	denoted $\len{V}$, is defined as follows:
	$$
	\len{V} = \begin{cases} \len{P} & \text{if $V = \abs{x:\lhotup{C}}{P}$} \\
	0 & \text{if $V = \abs{x:\shotup{C}}{P}$ or $V=y$} 
	\end{cases}
	$$
\end{definition}


\noindent
We define an auxiliary function that ``initializes'' the indices of a tuple of names.

%

\begin{definition}[Name and Process Initialization]
  \label{d:counterinit}
Let $\widetilde u = (a,b,s,s',\ldots)$ be a finite tuple of names.
We shall write $\mathsf{init}(\widetilde u)$ to denote the tuple
$(a_1,b_1,s_1,s'_1,\ldots)$.
We will say that a process has been \emph{initialized} if all of its names have some index.
\end{definition}


\begin{remark}\label{r:prefix}
Recall that we write
`$\apropinp{k}{}$' and  `$\apropout{k}{}$'
to denote input and output prefixes in which the value communicated along $\prop_k$ is not relevant.
While `$\apropinp{k}{}$' stands for `$\apropinp{k}{x}$', 
`$\apropout{k}{}$' stands for `$\apropout{k}{\abs{x}{\inact}}$'.
Their corresponding minimal types are 
$\btinp{\shot{\tinact}}\tinact$ 
and 
$\btout{\shot{\tinact}}\tinact$, 
which are 
denoted by
$\btinp{\cdot}\tinact$ and  $\btout{\cdot}\tinact$, respectively.
\end{remark}

Recall that  $P$ is \emph{closed} if $\fv{P} = \emptyset$.
We now define the \emph{decomposition} of a process. 

\begin{definition}[Decomposing Processes]
  \label{def:decomp}
	Let $P$ be a closed \HO process such that $\widetilde u = \fn{P}$.
  The \emph{decomposition} of $P$, denoted $\D{P}$, is
defined as:
  $$
  \D{P} = \news{\widetilde \prop}\big(\propout{k}{} \inact \Par \B{k}{\epsilon}{P\sigma}\big)
  $$
  where:
  $k > 0$;
  $\widetilde \prop = (\prop_k,\ldots,\prop_{k+\len{P}-1})$;
  $\sigma = \subst{\mathsf{init}(\tilde u)}{\widetilde u}$;
  and
  the \emph{breakdown function} $\B{k}{\tilde x}{\cdot}$, where $\widetilde{x}$ is a tuple of variables, is defined inductively in
  \secref{ss:bdown}.
  \end{definition}

\noindent
The bulk of the decomposition of a process is given by the breakdown function, detailed next.

\subsubsection{The Breakdown Function}
\label{ss:bdown}
Given a context $\widetilde x$ and a  $k>0$,
the breakdown function   $\B{k}{\tilde x}{\cdot}$ is defined on the structure of initialized processes, relying on the  breakdown function on values  $\V{k}{\tilde y}{\cdot}$.
The definition relies on type information; we describe each of its cases next.

\begin{description}
\item[Output]
The decomposition of $\bout{u_i}{V}Q$ is the most interesting case: an output prefix sends a value $V$ (i.e., an
abstracted process) that has to be broken down as well.
We then have:
$$
\B{k}{{\tilde x}}{\bout{u_i}{V}Q} =
  \propinp{k}{\widetilde x}
			\bbout{u_i}{\V{k+1}{\tilde y}{V\sigma}}
			\apropout{k+l+1}{\widetilde z}  \Par
\B{k+\degree+1}{\tilde z}{Q\sigma}
  $$
Process $\B{k}{{\tilde x}}{\bout{u_i}{V}Q}$
consists of a leading trio that mimics an output action
in parallel with the breakdown of the continuation $Q$.
The context  $\widetilde x$ must include the
free variables of $V$ and $Q$, denoted $\widetilde y$ and
$\widetilde z$, respectively.
These tuples are not necessarily disjoint:
variables with shared types can appear free in both $V$ and $Q$.
The output object $V$ is then broken down with parameters $\widetilde y$ and $k+1$; the latter serves
to consistently generate propagators for the trios in the breakdown of $V$, denoted $\V{k+1}{\tilde y}{V\sigma}$ (see below for its definition).
The  substitution $\sigma$  increments the index of session names; it is
applied to both $V$ and $Q$ before they are broken down.
We then distinguish two cases: 
\begin{itemize}
\item If name $u_i$ is linear (i.e., it has a session type) then its future occurrences
are renamed into $u_{i+1}$, and $\sigma = \subst{u_{i+1}}{u_i}$; 
\item Otherwise, if $u_i$ is not linear, then $\sigma = \{\}$.
\end{itemize}
Note that if $u_i$ is
linear then it appears either in $V$ or $Q$  and $\sigma$ affects
only one of them.

The last prefix in the leading trio  activates
the breakdown of  $Q$ with its corresponding context
$\widetilde z$.
To avoid name conflicts with the  propagators used in the breakdown of $V$, we use
$\dual{c_{k+l+1}}$,
with $\degree = \len{V}$ as a trigger for the continuation.

\smallskip

We remark that the same breakdown strategy is used when $V$ stands for a variable $y$. 
Since by definition $\len{y} = 0$, $\V{k}{\tilde y}{y} = y$, and $y \sigma = y$, we have: 
$$
\B{k}{\tilde x}{\bout{u_i}{y}Q} =
\binp{c_k}{\widetilde x} \bout{u_i}{y} \apropout{k+1}{\widetilde z}
\Par \B{k+1}{\tilde z}{Q\sigma}
$$
We may notice that
variable $y$ is not propagated further if it does not appear in $Q$. 

\item[Input]
The breakdown of an input prefix is defined as follows:
$$
\B{k}{\tilde x}{\binp{u_i}{y}Q} =
\propinp{k}{\widetilde x}\binp{u_i}{y} \apropout{k+1}{\widetilde x'}
     \Par \B{k+1}{\tilde x'}{Q\sigma}
$$
\noindent where $\widetilde x' = \fv{Q}$.
A leading trio mimics the input action and possibly extends the context with the received variable $y$.
The substitution $\sigma$ is defined as in the output case.

\item[Application] The breakdown of $\appl{V}{u_i}$ is as follows: 
$$
\B{k}{\tilde x}{\appl{V}{u_i}}
=
\propinp{k}{\widetilde x}\appl{\V{k+1}{\tilde x}{V}}{\widetilde m}
$$

A degenerate trio receives a context $\widetilde x$ and 
then proceeds with
the application. 
We break down $V$  with $\widetilde x$ as a context
since these variables need to be propagated to the abstracted process.
We use $k+1$ as a parameter to avoid name conflicts.
Name $u_i$ is decomposed into a tuple $\widetilde m$ using type information:
if  $u_i:C$ then
$\widetilde m = (u_i,\ldots,u_{i+\len{\Gt{C}}-1})$ and so
the length of $\widetilde m$ is $\len{\Gt{C}}$; each name in $\widetilde m$ will perform exactly one action.
When $V$ is a variable $y$, we have: 
$$
\B{k}{\tilde x}{\appl{y}{u_i}}
=
\propinp{k}{y}\appl{y}{\widetilde m}
$$
Notice that by construction $\widetilde x = y$.

\item[Restriction] We define the breakdown of a restricted process as follows:
$$
\B{k}{\tilde x}{\news{s:C}{P'}}
=
\news{\widetilde{s}:\Gt{C}}{\,\B{k}{\tilde x}
	{P'\sigma}}
$$
By construction, $\widetilde x = \fv{P'}$.
Similarly as in the decomposition of $u_i$ into $\widetilde m$ discussed above, we use the type $C$ of $s$ to
obtain the tuple $\widetilde s$
of length $\len{\Gt{C}}$.
We initialize the index of $s$ in $P'$ by applying the substitution $\sigma$.
This substitution
depends on $C$: if it is a shared type then $\sigma = \subst{s_1}{s}$; otherwise, if $C$ is a session type, then $\sigma = \subst{s_1\dual{s_1}}{s\dual{s}}$.

\item[Composition] The breakdown of a process $Q \Par R$ is as follows:
$$
\B{k}{\tilde x}{Q \Par R}
=
 \propinp{k}{\widetilde x} \propout{k+1}{\widetilde y}
    \apropout{k+\degree+1}{\widetilde z}   \Par
\B{k+1}{\tilde y}{Q} \Par \B{k+\degree+1}{\tilde z}{R}
$$
A control trio triggers the breakdowns of $Q$ and $R$; it
 does not mimic any action of the source process.
The tuple $\widetilde y \subseteq \widetilde x$ (resp. $\widetilde z \subseteq \widetilde x$)
collects the free variables in  $Q$ (resp. $R$).
To avoid name conflicts, the trigger for the breakdown of  $R$
is $\dual {c_{k+\degree+1}}$, with $\degree = \len{Q}$.

\item[Inaction]
To breakdown $\inact$, we define a degenerate trio with only one input prefix
that receives a context that by construction will always be empty, i.e., $\widetilde x = \epsilon$:
$$
\B{k}{\tilde x}{\inact}
=
	\propinp{k}{}\inact
$$

\item[Value] In defining the breakdown function for values 
we distinguish two main cases:
\begin{itemize}
\item If $V = \abs{y:\slhotup{C}}P$, where $\leadsto \in \{\multimap, \sharedop \}$, then we have:
$$
\V{k}{\tilde x}{\abs{y:\slhotup{C}}P}
=
\abs{\widetilde{y}:{\slhotup{\Gt{C}}}}{\news{\widetilde \prop}\big(\apropout{k}{\widetilde x}
   \Par \B{k}{\tilde x}{P \subst{y_1}{y}}}\big)
$$
We use  type $C$ to decompose $y$ into the tuple   $\widetilde{y}$.
We
abstract over $\widetilde{y}$; the body of the abstraction is the  composition of a control trio and the breakdown of
$P$, with name index initialized with the substitution $\subst{y_1}{y}$. 
If $\leadsto = \rightarrow$ then we restrict 
the propagators $\widetilde \prop = (\prop_k,\ldots,\prop_{k+\len{P}-1})$: this enables us to type
the value in a shared environment. When $\leadsto = \multimap$
 we do not have to restrict the 
propagators, and   $\widetilde \prop = \epsilon$. 

\item If $V = y$, then the breakdown function is the identity: 
$\V{k}{\tilde x}{y} = y$.
\end{itemize}

\end{description}

\noindent
\tabref{t:bdowncore} summarizes 
the definition of the breakdown, spelling out the side conditions involved.  
We illustrate it by means of an example:

  \begin{table}[!t]
\begin{tabular}{ |l|l|l|}
  \rowcolor{gray!25}
  \hline
  $P$ &
    \multicolumn{2}{l|}{
  \begin{tabular}{l}
    \noalign{\smallskip}
    $\B{k}{\tilde x}{P}$
    \smallskip
  \end{tabular}
}  \\
  \hline

$\bout{u_i}{V}{Q}$ &
    \begin{tabular}{l}
      \noalign{\smallskip}
      $\propinp{k}{\widetilde x}
			\bbout{u_i}{\V{k+1}{\tilde y}{V\sigma}}
			\apropout{k+l+1}{\widetilde z}  \Par$ \\
      $\B{k+\degree+1}{\tilde z}{Q\sigma}$
     \smallskip
  \end{tabular}
  &
  \begin{tabular}{l}
    \noalign{\smallskip}
    $\widetilde y = \fv{V}, \ \widetilde z = \fv{Q}$ \\
    $\degree = \len{V}$ \\
		$\sigma = \begin{cases}
		\incrname{u}{i} & \text{if } u_i:S \\
		\{\} & \text{otherwise}
		\end{cases}$
    \smallskip
  \end{tabular}
  \\
%
\hline 
$\binp{u_i}{y}Q$
&

  \begin{tabular}{l}
      \noalign{\smallskip}
      $\propinp{k}{\widetilde x}\binp{u_i}{y} \apropout{k+1}{\widetilde x'}
     \Par \B{k+1}{\tilde x'}{Q\sigma}$
      \smallskip
  \end{tabular}
  &
  \begin{tabular}{l}
    \noalign{\smallskip}
	$\widetilde x' = \fv{Q} $ \\
		$\sigma = \begin{cases}
		\incrname{u}{i} & \text{if } u_i:S \\
		\{\} & \text{otherwise}
		\end{cases}$
    \smallskip
  \end{tabular}
    \\
  \hline

    $\appl{V}{u_i}$ &
  \begin{tabular}{l}
    \noalign{\smallskip}
    $\propinp{k}{\widetilde x}\appl{\V{k+1}{\tilde x}{V}}{\widetilde m}$
    \smallskip
  \end{tabular}
  &
  \begin{tabular}{l}
  	\noalign{\smallskip}
  	$u_i:C$ \\
  	$\widetilde x = \fv{V}$ \\
    $\widetilde{m} = (u_i,\ldots,u_{i+\len{\Gt{C}}-1})$
    \smallskip
  \end{tabular}
  \\
  \hline

  $\news{s:C}{P'}$ &
  \begin{tabular}{l}
    \noalign{\smallskip}
 	$\news{\widetilde{s}:\Gt{C}}{\,\B{k}{\tilde x}
 			{P'\sigma}}$
    \smallskip
  \end{tabular}
  &
  \begin{tabular}{l}
    \noalign{\smallskip}
    $\widetilde x = \fv{P'}$ \\
    $\widetilde{s} = (s_1,\ldots,s_{\len{\Gt{C}}})$ \\
	$\sigma = \begin{cases}
				\subst{s_1 \dual{s_1}}{s \dual{s}} & \text{if} \ C = S \\
				\subst{s_1}{s} & \text{if} \ C = \chtype{U}
\end{cases}
$
    \smallskip
  \end{tabular}
  \\
  \hline
  $Q \Par R$ &
  \begin{tabular}{l}
    \noalign{\smallskip}
    $\propinp{k}{\widetilde x} \propout{k+1}{\widetilde y}
    \apropout{k+\degree+1}{\widetilde z}   \Par$ 
    \\
    $\B{k+1}{\tilde y}{Q} \Par \B{k+\degree+1}{\tilde z}{R}$
    \smallskip
  \end{tabular}
  &
  \begin{tabular}{l}
    \noalign{\smallskip}
    $\widetilde y  = \fv{Q}$ \\
    $\widetilde z = \fv{R}$ \\
    $\degree = \len{Q}$
    \smallskip
  \end{tabular}
      \\
  \hline

  $\inact$ &
\begin{tabular}{l}
  \noalign{\smallskip}
  $\propinp{k}{}\inact$
 \smallskip
\end{tabular}
&
\begin{tabular}{l}
	\noalign{\smallskip}
\smallskip
\end{tabular}
  \\
  \hline

\rowcolor{gray!25}
$V$ &
  \multicolumn{2}{l|}{
\begin{tabular}{l}
  \noalign{\smallskip}
  $\V{k}{\tilde x}{V}$
  \smallskip
\end{tabular} } 
\\
\hline
$y$ &
\begin{tabular}{l}
  \noalign{\smallskip}
  $y$
  \smallskip
\end{tabular}
&
\begin{tabular}{l}
  \noalign{\smallskip}
  \smallskip
\end{tabular}
\\
\hline
%
%

$\abs{u:\slhotup{C}}P$ &
\begin{tabular}{l}
  \noalign{\smallskip}
  \begin{tabular}{l}
  $\abs{\widetilde{y}:{\slhotup{\Gt{C}}}}{\news{\widetilde \prop}\big(\apropout{k}{\widetilde x}
   \Par$ 
   \\ 
   \qquad \qquad \qquad \quad \quad $\B{k}{\tilde x}{P \subst{y_1}{y}}}\big)$
   \end{tabular}
  \smallskip
  \\
  $\widetilde{\prop} = \begin{cases} 
        \epsilon & \text{if} \ \leadsto = \multimap \\
     (\prop_k,\ldots,\prop_{k+\len{P}-1}) & \text{if} \ \leadsto = \rightarrow
 \end{cases}$ 
  \smallskip
\end{tabular}
&
\begin{tabular}{l}
  \noalign{\smallskip}
  $\widetilde x = \fv{V}$ \\
  $\widetilde{y} = (y_1,\ldots,y_{\len{\Gt{C}}})$ 
\end{tabular}
\\
\hline
\end{tabular}
\caption{The breakdown function for processes and values (core fragment). \label{t:bdowncore}}
\end{table}

\begin{example}[Breaking Down Name-Passing]
\label{ex:bdnp}
Consider the following process $P$, in which a channel $m$ is passed, through which a Boolean value is sent back: 
\begin{align*}
	P = \news{u}(\bout{u}{m} \binp{\dual m}{b} \inact \Par
		\binp{\dual u}{x} \bout{x}{\textsf{true}} \inact)
\end{align*}
$P$ is not an \HO process as it features name-passing. We then use the encoding described in \exref{ex:np} to construct its encoding into \HO. We thus obtain $\llbracket P \rrbracket = \news{u} (Q \Par R)$, where 
\begin{align*}
	& Q = \bout{u}{V}
		\binp{\dual m}{y}\news{s}{(\appl{y}{s} \Par
		\bout{\dual s}{\abs{b}{\inact}}\inact
		)} & V = \abs{z}\binp{z}{x}(\appl{x}{m})\\
	& R = \binp{\dual u}{y}\news{s}{(\appl{y}{s} \Par
		\bout{\dual s}{W}\inact)} & W =\abs{x}{\bout{x}{W'}\inact} \text{ with }W' = \abs{z}\binp{z}{x}(\appl{x}{\textsf{true})}
\end{align*}
By \exref{ex:np}, we know that $\map{\cdot}$ requires exactly four reduction steps to mimic a name-passing synchronization. We show here part of the reduction chain of $\map{P}$:
\begin{align}
\label{ex:reduction-chain}
&\map{P} \red^4  \map { \binp{\dual m}{b} \inact \Par
\bout{m}{\textsf{true}} \inact } \red^4 \inact
\end{align}
We will now investigate the decomposition of $\llbracket P \rrbracket$ and its reduction chain. 
First, we use \defref{def:sizeproc} to compute $\len V=\len {W'} =2$, and so $\len W = 4$. 
Then $\len{Q}=\len{\appl{y}{s} \Par
	\bout{\dual s}{\abs{b}{\inact}}\inact} + \len{V} + 2 = 9$, and similarly,  $\len R = 9$. 
	Therefore, $\len {\map{P}} = 19$.
Following \defref{def:decomp}, we see that $\sigma = \subst{m_1\dual{m_1}}{m\dual m}$, which we silently apply. Using $k=1$, we then have
the decomposition shown in \tabref{t:examp}.

\begin{table}[t!]
\begin{align*}
	 \D{\map P} & = \news{c_1,\ldots, c_{19}} \Big( \apropout{1}{}   \Par 
	\news{u_1}\big( \propinp{1}{} \propout{2}{} \apropout{11}{}   \Par \B{2}{\epsilon}{Q} \Par \B{11}{\epsilon}{R} \big) \Big)
	\\
	 \B{2}{\epsilon}{Q} & = \propinp{2}{} \bout{u_1}{\V{3}{\epsilon}{V}}
	\apropout{5}{}  \Par \propinp{5}{} \binp{\dual{m_1}}{y} \apropout{6}{y}  \Par 
	\\ 
	&   \quad \news{s_1}(\propinp{6}{y}\propout{7}{y}\apropout{8}{}  \Par 
	\propinp{7}{y}(\appl{y}{s_1}) \Par 
	\propinp{8}{}\bout{\dual {s_1}}{\V{9}{\epsilon}{\abs{b}{\inact}}}\apropout{10}{}  \Par \apropinp{10}{}  )
	\\
	\B{11}{\epsilon}{R} & = \propinp{11}{}\binp{\dual {u_1}}{y}\apropout{12}{y}  \Par  
	\\
	& \quad \news{s_1}\big( \propinp{12}{y}\propout{13}{y}\apropout{14}{} 
	\Par \propinp{13}{y}(\appl{y}{s_1}) \Par 
	\propinp{14}{}\bout{\dual{s_1}}{\V{15}{\epsilon}{W}}\apropout{19}{}  \Par \apropinp{19}{}  \big) 
	\\
	\V{3}{\epsilon}{V} & = \abs{z_1}{(\apropout{3}{}  \Par 
	\propinp{3}{} \binp{z_1}{x} \apropout{4}{x}   \Par 
	\propinp{4}{x} (\appl{x}{m_1}) )}
	\\
	 \V{9}{\epsilon}{\abs{b}{\inact}} & = \abs{b_1}{(\apropout{9}{}  \Par 
		\apropinp{9}{})} 
		\\
	 \V{15}{\epsilon}{W} & = \abs{x_1}{( \apropout{15}{}   \Par \propinp{15}{} \bout{x_1}{\V{16}{\epsilon}{W'}} \apropout{18}{}   \Par \apropinp{18}{}  )} 
	 \\
	 \V{16}{\epsilon}{W'} &  = \abs{z_1}{( \apropout{16}{}  \Par \propinp{16}{} \binp{z_1}{x} \apropout{17}{x}  \Par \propinp{17}{x} (\appl{x}{\textsf{true}}) )}
\end{align*}
\caption{The process decomposition discussed in \exref{ex:bdnp}. \label{t:examp}}
\end{table}
\tabref{t:examp} we have omitted substitutions that have no effect and trailing $\inact$s. The first interesting process appears after synchronizations on $c_1$, $c_2$, and $c_{11}$. At that point, the process will be ready to mimic the first action that is performed by  $\map P $, i.e., $u_1$ will send $\V{3}{\epsilon}{V}$, the breakdown of $V$. Next, $c_{12}$, $c_{13}$, and $c_{14}$ will synchronize, and  $\V{3}{\epsilon}{V}$ is passed further along, until $s_1$ is ready to be applied to it in the breakdown of $R$. At this point, we know that $\map{P} \red^{7} \news{\widetilde c}P'$, where $\widetilde c = (c_3,\ldots,c_{10},c_{15},\ldots,c_{19})$, and 
\begin{align*}
 P' = &~
\propout{5}{}\inact \Par \propinp{5}{} \binp{ \dual{m_1}}{y} \propout{6}{y} \inact 
\\
& \Par  
\news{s_1}(\propinp{6}{y}\propout{7}{y}\propout{8}{}\inact \Par 
\propinp{7}{y}\appl{y}{s_1} \Par 
\propinp{8}{}\bout{\dual {s_1}}{\V{9}{\epsilon}{\abs{b}{\inact}}}\propout{10}{} \inact \Par \propinp{10}{} \inact)\\
& \Par
\news{s_1}\big(  \appl{\V{3}{\epsilon}{V}}{s_1} \Par 
\bout{\dual{s_1}}{\V{15}{\epsilon}{W}}\propout{19}{}\inact \Par \propinp{19}{}\inact \big) 
\end{align*} 
After $s_1$ is applied, the trio guarded by $c_3$ will be activated, where $z_1$ has been substituted by $s_1$. Then $\dual{s_1}$ and $s_1$ will synchronize, and the breakdown of $W$ is passed along. Then $c_4$ and $c_{19}$ synchronize, and now $m_1$ is ready to be applied to $\V{15}{\epsilon}{W}$, which was the input for $c_4$ in the breakdown of $V$. After this application, $c_5$ and $c_{15}$ can synchronize with their duals, and we know that $\news{\widetilde c}P' \red^8 \news{\widetilde{c}'}P''$, where $\widetilde{c}'=(c_6,\ldots,c_{10},c_{16},c_{17},c_{18})$, and
\begin{align*}
P'' = & ~\binp{\dual{m_1}}{y} \propout{6}{y} \inact  
\Par
 \bout{m_1}{\V{15}{\epsilon}{W'}} \propout{17}{} \inact \Par \propinp{17}{}\inact
\\
& \Par \news{s_1}(\propinp{6}{y}\propout{7}{y}\propout{8}{}\inact \Par 
\propinp{7}{y}\appl{y}{s_1} \Par 
\propinp{8}{}\bout{\dual {s_1}}{\V{9}{\epsilon}{\abs{b}{\inact}}}\propout{10}{} \inact \Par \propinp{10}{} \inact)
\end{align*}
Remarkably, $P''$ is standing by to mimic the encoded exchange of value \textsf{true}. Indeed, the decomposition of the four-step reduced process in \eqref{ex:reduction-chain} will reduce in three steps to a process that is equal (up to $\scong_\alpha$) to the process we obtained here. This strongly suggests  a tight operational correspondence between a process and its decomposition. \hspace*{\fill} $\lhd$
\end{example}


We may now state our technical results:

\begin{theorem}[Typability of Breakdown]
\label{t:typecore}
   Let $P$ be an initialized process and $V$ be a value.
  \begin{enumerate}
  \item
  If $\Gamma;\Lambda;\Delta \proves P \hastype \Proc$ then
  $$\Gt{\Gamma_1};\es;\Gt{\Delta} \cat
  \Theta \proves \B{k}{\tilde x}{P} \hastype \Proc 
  \quad 
  (k > 0)
  $$ where:
   $\widetilde x = \fv{P}$;
   $\Gamma_1=\Gamma \setminus \widetilde x$;
   and
   $\balan{\Theta}$ with
    $\text{dom}(\Theta) = \{\prop_k,\prop_{k+1},\ldots,\prop_{k+\len{P}-1}\} \cup \{\dual{\prop_{k+1}},\ldots,\dual{\prop_{k+\len{P}-1}}\}$ and $\Theta(\prop_k)=
  \btinp{\widetilde M} \tinact$,
  where $\widetilde M = (\Gt{\Gamma}\cat\Gt{\Lambda})(\widetilde x)$.

  \item If $\Gamma;\Lambda;\Delta \proves V \hastype \lhot{C}$ then
  $$\Gt{\Gamma};\Gt{\Lambda};\Gt{\Delta} \cat \Theta \proves \V{k}{\tilde
  x}{V} \hastype \lhot{\Gt{C}}
    \quad 
  (k > 0)
  $$ where:
  $\widetilde x = \fv{V}$;
  and $\balan{\Theta}$ with
  $\text{dom}(\Theta) =  \{\prop_k,\ldots,\prop_{k+\len{V}-1}\} \cup \{\dual{\prop_{k}},\ldots,\dual{\prop_{k+\len{V}-1}}\}$
  and $\Theta(\prop_k)= \btinp{\widetilde M} \tinact$
  and $\Theta(\dual{\prop_k})= \btout{\widetilde M} \tinact$,
  where $\widetilde M = (\Gt{\Gamma}\cat\Gt{\Lambda})(\widetilde x)$.

  \item If $\Gamma;\es;\es \proves V \hastype \shot{C}$ then
  $\Gt{\Gamma};\es;\es \proves \V{k}{\tilde x}{V} \hastype \shot{\Gt{C}}$, where $\widetilde x = \fv{V}$
  and $k > 0$.

\end{enumerate}
\end{theorem}
\begin{proof}
By mutual induction on the structure of $P$ and $V$. 
\appendx{See Appendix~\ref{app:typecore} for details.}
\end{proof}

Using the above theorem, we can prove our main result:

\begin{theorem}[Typability of the Decomposition]
\label{t:decompcore}
	Let $P$ be a closed \HO process with $\widetilde u = \fn{P}$.
	If $\Gamma;\es;\Delta \proves P \hastype \Proc$ then
	$\Gt{\Gamma \sigma};\es;\Gt{\Delta \sigma} \proves \D{P} \hastype \Proc$, where
	$\sigma = \subst{\mathsf{init}(\widetilde u)}{\widetilde u}$.
\end{theorem}

\begin{proof}
Direct from the definitions, using \thmref{t:typecore}. 
\appendx{See Appendix~\ref{app:decompcore} for details.}
\end{proof}

\subsection{Extensions (I): Select and Branching}
\label{ss:exti}

We now show how to extend the decomposition to handle select and branch processes,
which implement labeled (deterministic) choice in session protocols,
 as well as their corresponding session types. As we will see, in formalizing this extension we shall appeal to the expressive power of abstraction-passing.
We start by extending the syntax of minimal session types:

\begin{definition}[Minimal Session Types (with Labeled Choice)]
\label{d:mtypesii}
The syntax of \emph{minimal session types} for \HO is defined as follows: 
	\begin{align*}
		M & \bnfis 	\tinact  \bnfbar  \btout{\widetilde{U}} \tinact \bnfbar \btinp{\widetilde{U}} \tinact
		\bnfbar \btsel{l_i : M_i}_{i \in I} \bnfbar \btbra{l_i : M_i}_{i \in I}
	\end{align*}
	where
	$U$ and $C$ are defined as in \defref{d:mtypesi}.
\end{definition}

We may then extend \defref{def:typesdecomp} to branch and select types as follows:
  \begin{definition}[Decomposing Session Types, Extended (I)]
  \label{def:typesdecomp-bra}
  The decomposition function on types as given in
  \defref{def:typesdecomp} is
  extended as follows:
    \begin{align*}
      \Gt{\btbra{l_i : S_i}_{i \in I}} &= \btbra{l_i:
\btout{\lhot{\Gt{S_i}}} \tinact}_{i \in I}
\\
      \Gt{\btsel{l_i : S_i}_{i \in I}} &=
      \btsel{l_i:\btinp{\lhot{\Gt{\dual {S_i}}}} \tinact}_{i \in I}
      \end{align*}
    \end{definition}
    The above definition for decomposed types already suggests our strategy to breakdown branching and selection processes: we will exploit abstraction-passing to exchange one abstraction per each branch of the labeled choice. This   intuition will become clearer shortly.

We now extend the definition of the degree of a process/value (cf. \defref{def:sizeproc}) to account for branch and select processes:
\begin{definition}[Degree of a Process, Extended]
	\label{def:sizeproc-bra}
	The \emph{degree} of a process $P$, denoted $\len{P}$, is as given in \defref{def:sizeproc}, extended as follows:
	$$
	\len{P} =
	\begin{cases}
	1 & \text{if $P = \bbra{u_i}{l_j:P_j}_{j \in I}$}
	\\
	\len{P'} + 2 & \text{if $P = \bsel{u_i}{l_j}P'$}
	\end{cases}
	$$
\end{definition}
The definition of process decomposition (cf. \defref{def:decomp}) does not require modifications; it relies on the extended definition of the breakdown function for processes $\B{k}{\tilde x}{\cdot}$ that combines the definitions in \tabref{t:bdowncore} with those in \tabref{t:bdown-selbra} (see below).
The breakdown of values $\V{k}{\tilde x}{\cdot}$ is as before, and relies on the extended definition of $\B{k}{\tilde x}{\cdot}$.

We now present and describe the breakdown of branching and selection processes:
%
\begin{description}

\item[Branching]
The breakdown of a branching process $\bbra{u_i}{l_j:P_j}_{j \in I}$ is as follows:
 \begin{align*}
 \B{k}{\tilde x}{\bbra{u_i}{l_j:P_j}_{j \in I}} = ~&
 \propinp{k}{\widetilde x} \bbra{u_i}{l_j :
   \abbout{u_i}{N_{u,j}
  	}
  }_{j \in I} 
  \\
  & \text{where} 
  \ \ N_{u,j} =  \abs{\widetilde y^{u}_{j} : \Gt{S_j}}{
  	\news{\widetilde \prop_j}
  	 \big(\apropout{k+1}{\widetilde x} 
  	\Par
  	\B{k+1}{\tilde x}{P_j \subst{y^{u}_{1}}{u_i}}\big)}
 \end{align*}
The first prefix receives the context $\widetilde x$.
The next two prefixes are along $u_i$: the first one mimics the branching action of $P$, whereas
the second outputs an abstraction $N_{u,j}$.
This output does not have a counterpart in $P$; it is meant to synchronize with an input in the breakdown of the corresponding selection process (see below).
$N_{u,j}$ encapsulates the breakdown of subprocess $P_j$.
It  has the same structure as the breakdown of a value $\abs{y:\shotup{C}}{P}$ in \tabref{t:bdowncore}:
it is a composition of a control trio and the breakdown of $P_j$; the generated propagators, denoted $\widetilde \prop_j$, are restricted.
We use types to define $N_{u,j}$:
we assume $S_j$ is the session type of $u_i$ in the $j$-th branch of $P$.
We abstract over  $\widetilde y^{u}_j = (y^{u}_1,\ldots,y^{u}_{\len{\Gt{S_j}}})$.
We substitute $u_i$ with $y^{u}_1$ in $P_j$ before breaking it down: this way, $u_i$ is decomposed and bound by abstraction.

\item[Selection]
The breakdown of a selection process $\bsel{u_i}{l_j}P'$ is as follows:
\begin{align*}
\B{k}{\tilde x}{\bsel{u_i}{l_j}P'} = ~&
 	\propinp{k}{\widetilde x}\apropbout{k+1}{M_j}
  \Par \news{\widetilde u:\Gt{S_j}}{(\propinp{k+1}{y}{\appl{y}{\widetilde{\dual u}}} \Par
  \B{k+2}{\tilde x}{P'\incrname{u}{i}})}
  \\
  &\text{where~} M_j = \abs{\widetilde y}{\bsel{u_i}{l_j}{\binp{u_i}{z}	\propout{k+2}{\widetilde x}\appl{z}{\widetilde y}}}
\end{align*}

 After receiving the context $\widetilde x$, 
 the abstraction $M_j$ is sent along $\dual{\prop_{k+1}}$, and is to be received by the second subprocess in the composition.
 This sequence of actions allows us to preserve the intended trio structure. We use $S_j$,  the type of $u_i$ in $P'$, 
  to construct a corresponding tuple $\widetilde u$, with type $\Gt{S_j}$.
 We apply the abstraction $M_j$, received along $\prop_{k+1}$, to $\widetilde{\dual u}$ (the duals of $\widetilde u$).
 At this point, the selection action in $P$ can be mimicked, and so label $l_j$ is chosen from the breakdown of a corresponding branching process. As discussed above, such a breakdown will send an abstraction $N_{u,j}$  with type $\lhot{\dual{S_j}}$, which encapsulates the breakdown of the chosen subprocess. Before running $N_{u,j}$ with names $\widetilde{\dual u}$, we trigger the breakdown of $P'$ with an appropriate substitution.
\end{description}

\begin{table}[!t]
\begin{tabular}{ |l|l|}
  \rowcolor{gray!25}
  \hline
  \multicolumn{2}{|l|}{
  \begin{tabular}{l}
  \noalign{\smallskip}
    $\B{k}{\tilde x}{\bbra{u_i}{l_j{:}P_j}_{j \in I}}$
    \smallskip
  \end{tabular} 
  \begin{tabular}{l}
  \end{tabular} 
  }
  \\
  \hline 
  \begin{tabular}{l}
    \noalign{\smallskip}
    $\propinp{k}{\widetilde x} \bbra{u_i}{l_j :
   \abbout{u_i}{N_{u,j}
  	}
  }_{j \in I}$ 
  \\
  where: 
  \\ $N_{u,j} =  \abs{\widetilde y^{u}_{j} : \Gt{S_j}}{
  	\news{\widetilde \prop_j}
  	 \big(\apropout{k+1}{\widetilde x} 
  	\Par
  	\B{k+1}{\tilde x}{P_j \subst{y^{u}_{1}}{u_i}}\big)}$ 
  \smallskip
  \end{tabular} & 
  \begin{tabular}{l}
      $\widetilde y^{u}_{j} = (y^{u}_{1},\ldots,y^{u}_{\len{\Gt{S_j}}})$ \\
    $\widetilde \prop_j = (\prop_{k+1},\ldots,\prop_{k+\len{P_j}})$
  \end{tabular} 
  \\
  \hline
  \rowcolor{gray!25}
    \multicolumn{2}{|l|}{
  \begin{tabular}{l}
    \noalign{\smallskip}
    $\B{k}{\tilde x}{\bsel{u_i}{l_j}P'}$
    \smallskip 
  \end{tabular}  
  }
  \\
  \hline 
  \begin{tabular}{l}
    \noalign{\smallskip}
  $\propinp{k}{\widetilde x}\apropbout{k+1}{M_j}
  \Par$ \\
  $\news{\widetilde u:\Gt{S_j}}{\big(\propinp{k+1}{y}{\appl{y}{\widetilde{\dual u}}} \Par
  \B{k+2}{\tilde x}{P'\incrname{u}{i}}\big)}$
  \\
  where: 
  \\
  $M_j = \abs{\widetilde y}{\bsel{u_i}{l_j}{\binp{u_i}{z}\propout{k+2}{\widetilde x}\appl{z}{\widetilde y}}}$
  \smallskip
 \end{tabular} & 
  \begin{tabular}{l}
      	\noalign{\smallskip}
  	$\widetilde y = (y_1,\ldots,y_{\len{\Gt{S_j}}})$ \\
	$\widetilde u=(u_{i+1},\ldots,u_{i+\len{\Gt{S_j}}})$
  \\
  $\widetilde {\dual {u}}=(\dual {u_{i+1}},\ldots,\dual {u_{i+\len{\Gt{S_j}}}})$ 
  \smallskip
  \end{tabular} \\
  \hline 
\end{tabular}
\caption{The breakdown function for processes  (extension with selection and branching).}
\label{t:bdown-selbra}
\end{table}

\noindent
Summing up, our strategy for breaking down labeled choices exploits higher-order concurrency to uniformly handle the fact that
the subprocesses of a branching process have a different session type and degree.
Interestingly, it follows the intuition that branching and selection correspond to a form of output and input actions involving labels, respectively.

\begin{remark}
Theorems~\ref{t:typecore} and \ref{t:decompcore} hold also for the extension with selection and branching
\appendx{(see
Appendix \ref{app:exti} for details)}.
\end{remark}


\begin{example}[Breaking down Selection and Branching]
	We illustrate the breaking down of selection and branching processes by considering a basic mathematical server $Q$ that allows clients to add or subtract two integers. The server contains two {branches}: one sends an abstraction $V_+$ that implements integer {addition}, the other sends an abstraction $V_-$ implementing {subtraction}. A client $R$ {selects} the first option to add  integers 16 and 26:
	\begin{align*}
	& Q \defas \bbra{u}{\textsf{add} : \bout{u}{V_+} \inact, \textsf{sub} : \bout{u}{V_-} \inact} \\
	& R \defas \bsel{\dual{u}}{\textsf{add}} \binp{\dual {u}}{x} \appl{x}{(\text{16,26})}
	\end{align*}
	The composition $P \defas \news{u}(Q \Par R)$ reduces in two steps to a process $\appl{V_+}{(\text{16,26})}$:
	\begin{align}\label{ex:reduction}
	 P ~\red~ \news{u} ( \bout{u}{V_+} \inact \Par \binp{\dual {u}}{x} \appl{x}{(\text{16,26})}) ~\red~ 
	 \appl{V_+}{(\text{16,26})}
	 \end{align}
	We will investigate the decomposition of $P$, and its reduction chain. 
	First, by \defref{def:sizeproc} and \defref{def:sizeproc-bra}, we have: $\len{Q}=1$, $\len{R}=4$, and $\len{P}=6$. 
	Following the extension of \defref{def:decomp}, using $k=1$, and observing that $\sigma_1 = \{\}$, we obtain:
	\begin{align*}
	\D{P} = \news{c_1\ldots c_6} \big( \propout{1}{} \inact \Par \news{u_1} (\propinp{1}{} \propout{2}{} \propout{3}{} \inact \Par \B{2}{\epsilon}{Q\sigma_2} \Par \B{3}{\epsilon}{R\sigma_2}) \big)
	\end{align*}
	where $\sigma_2=\subst{u_1\dual{u_1}}{u\dual u}$. 
	The breakdown of $Q$ is obtained by applying the first rule in \tabref{t:bdown-selbra}:
	\begin{align*}
	\B{2}{\epsilon}{Q\sigma_2} = \propinp{2}{} u_1 \brases \{  \textsf{add}: & ~\news{c_{3}c_{4}} u_1 \outses \big\langle \abs{y_1}{
		\propout{3}{} \inact
		\Par \B{3}{\epsilon}{\bout{u_1}{V_+} \inact \subst{y_1}{u_1}}
	} \big\rangle \shsep
	\inact,  \\
	\textsf{sub}: & ~\news{c_{3}c_{4}} u_1 \outses \big\langle \abs{y_1}{
		\propout{3}{} \inact
		\Par \B{3}{\epsilon}{\bout{u_1}{V_-} \inact \subst{y_1}{u_1}}
	} \big\rangle \shsep
	\inact \}
	\end{align*}
	The breakdown of $R$ is obtained by applying the second rule in \tabref{t:bdown-selbra}:
	\begin{align*}
	\B{3}{\epsilon}{R\sigma_2} = ~& \propinp{3}{}\propout{4}{\abs{y_1}{\bsel{\dual{u_1}}{\textsf{add}}{\binp{\dual{u_1}}{z}\propout{5}{}\appl{z}{y_1}}}}\inact \\
	& \Par \news{u_2}{(\propinp{4}{y}{\appl{y}{\dual {u_2}}} \Par
		\B{5}{\epsilon}{ \binp{\dual{u_1}}{x} \appl{x}{(\text{16,26})} \subst{u_2}{\dual{u_1}}})}
	\end{align*}
	We will now follow the chain of reductions of the process $\D{P}$. First, $c_1$, $c_2$, and $c_3$ will synchronize, after which $c_4$ will pass the abstraction. Let $\D{P} \red^4 P'$, then we know: \\
		\setlength{\abovedisplayskip}{0pt} \setlength{\abovedisplayshortskip}{0pt}
		\setlength{\belowdisplayskip}{0pt} \setlength{\belowdisplayshortskip}{0pt}
		\begin{align*}
	P' = \news{c_5 c_6} \news{u_1} \big(
	u_1 \brases \{  \textsf{add}: & ~\news{c_{3}c_{4}} u_1 \outses \big\langle \abs{y_1}{
		\propout{3}{} \inact
		\Par \B{3}{\epsilon}{\bout{y_1}{V_+} \inact}
	} \big\rangle \shsep
	\inact,  \\
	\textsf{sub}: & ~\news{c_{3}c_{4}} u_1 \outses \big\langle \abs{y_1}{
		\propout{3}{} \inact
		\Par \B{3}{\epsilon}{\bout{y_1}{V_-} \inact}
	} \big\rangle \shsep
	\inact \}
	\end{align*}
	\begin{align*}
	\qquad \qquad \qquad \quad \Par \news{u_2}{({\appl{\abs{y_1}{\bsel{\dual{u_1}}{\textsf{add}}{\binp{\dual{u_1}}{z}\propout{5}{}\appl{z}{y_1}}}}{\dual {u_2}}} \Par
		\B{5}{\epsilon}{ \binp{u_2}{x} \appl{x}{(\text{16,26})} })} \big )
	\end{align*} \\
	\setlength{\abovedisplayskip}{12pt} \setlength{\abovedisplayshortskip}{12pt}
	\setlength{\belowdisplayskip}{12pt} \setlength{\belowdisplayshortskip}{12pt}
	In $P'$, $\dual{u_2}$ will be applied to the abstraction with variable $y_1$. After that, the choice for the process labeled by $\textsf{add}$ is made. Process $P'$ will reduce further as  $P' \red^2 P'' \red^2 P'''$, where: \\
	\setlength{\abovedisplayskip}{0pt} \setlength{\abovedisplayshortskip}{0pt}
	\setlength{\belowdisplayskip}{0pt} \setlength{\belowdisplayshortskip}{0pt}
	\begin{align*}
	P'' = \news{c_5 c_6} \news{u_1} \big( &
	\news{c_{3}c_{4}} u_1 \outses \big\langle \abs{y_1}{
		\propout{3}{} \inact
		\Par \B{3}{\epsilon}{\bout{y_1}{V_+} \inact}
	} \big\rangle \shsep
	\inact \\
	& \Par \news{u_2}{(\binp{\dual{u_1}}{z}\propout{5}{}\appl{z}{\dual {u_2}} \Par
		\B{5}{\epsilon}{ \binp{u_2}{x} \appl{x}{(\text{16,26})} })} \big ) 
	\end{align*} 
	\begin{align*}
	P''' = \news{c_3c_4c_5 c_6} \big(  \news{u_2}{\propout{5}{}
			\propout{3}{} \inact
			\Par \B{3}{\epsilon}{\bout{u_2}{V_+} \inact }
		  \Par
		\B{5}{\epsilon}{ \binp{u_2}{x} \appl{x}{(\text{16,26})} })} \big ) 
	\end{align*}\\
	\setlength{\abovedisplayskip}{12pt} \setlength{\abovedisplayshortskip}{12pt}
	\setlength{\belowdisplayskip}{12pt} \setlength{\belowdisplayshortskip}{12pt}
	Interestingly,  $P'''$ strongly resembles a decomposition of the one-step reduced process in \eqref{ex:reduction}. This advocates the operational correspondence between a process and its decomposition. 
	\hspace*{\fill} $\lhd$
\end{example}

\subsection{Extensions (II): Recursion}
\label{ss:extii}

We extend the decomposition to handle \HO processes in which names can be typed with recursive session types $\trec{t}{S}$. 
We consider recursive types which are \textit{simple} and \textit{contractive}, i.e., in $\trec{t}{S}$, the body $S \neq \vart{t}$ does not contain recursive types. Unless stated otherwise, we shall handle \emph{tail-recursive} session types such as, e.g., 
$S = \trec{t}\btinp{\mathsf{Int}}\btinp{\mathsf{Bool}}\btout{\mathsf{Bool}}\vart{t}$.  
Non-tail-recursive session types such as $\trec{t}{\btinp{\shot{(\widetilde T,\vart{t})}}\tinact}$, which is essential in the 
fully abstract encoding of \HOp into \HO~\cite{DBLP:conf/esop/KouzapasPY16}, can also be accommodated; see \remref{r:ntrsts} below. 

We start by extending minimal session types
(\defref{d:mtypesi}) with 
minimal recursive types: 

\begin{definition}[Minimal Recursive Session Types]
	\label{d:mtypesi-rec}
	The syntax of \emph{minimal recursive session types} 
	for \HO is defined as follows: 
	 \begin{align*}
	 M & 
	 \bnfis 	\mugamma  \bnfbar  \btout{\widetilde{U}} \mugamma 			\bnfbar \btinp{\widetilde{U}} \mugamma \bnfbar \trec{t}{M }	\\
	 \mugamma & \bnfis \tinact \bnfbar \vart{t}
	 \end{align*}


\end{definition}
Thus, types  such as $\trec{t}{\btout{U}\vart{t}}$ and
$\trec{t}{\btinp{U}\vart{t}}$
are minimal recursive session types: in fact they are  tail-recursive session types with exactly one 
session prefix.    
We extend \defref{def:typesdecomp} as follows:

\begin{definition}[Decomposing Session Types, Extended (II)]
 \label{def:recurdecomptypes}
Let $\trec{t}{S}$ be a recursive session type.
The  decomposition
function given in \defref{def:typesdecomp} is extended as:
	\label{def:decomptyp-rec}
 \begin{align*}
 	\Gt{\vart{t}} &= \vart{t}
	& 
    \Gt{\trec{t}{S}} &= \begin{cases} 	\Rt{S} & \text{if $\trec{t}{S}$ is tail-recursive} \\ 
    \trec{t}{\Gt{S}} & \text{otherwise}
 \end{cases} 
	\\
\Rt{\vart{t}} &= \epsilon 
&
		\Rt{\btout{U}S} &= \trec{t}{\btout{\Gt{U}}} \tvar{t}, \Rt{S} 
		\\
& & \Rt{\btinp{U}S} &= \trec{t}{\btinp{\Gt{U}}} \tvar{t}, \Rt{S} 
 \end{align*}
We shall also use the function $\Rts{}{s}{\cdot}$, which is defined as follows:
 \begin{align*}
 \Rts{}{s}{\btinp{U}S} = \Rts{}{s}{S} \qquad
\Rts{}{s}{\btout{U}S} = \Rts{}{s}{S} \qquad
\Rts{}{s}{\trec{t}{S}} = \Rt{S}
 \end{align*}
\end{definition}

Hence, $\Gt{\trec{t}{S}}$ is 
 a list of minimal recursive session types, obtained using the auxiliary function $\Rt{\cdot}$ on   $S$: if $S$ has $k$ prefixes then the list   
 $\Gt{\trec{t}{S}}$ will contain $k$ minimal recursive session types. 
 The auxiliary function $\Rts{}{s}{\cdot}$ decomposes
 \emph{guarded}  recursive session types: it skips session prefixes until a type of form $\trec{t}{S}$ is encountered; 
 when that occurs, the recursive type is decomposed using $\Rt{\cdot}$. 
 We illustrate \defref{def:recurdecomptypes} with two examples:

\begin{example}[Decomposing a Recursive Type]
\label{ex:rtype}
Let
$S = \trec{t}S'$ be a recursive session type, with $S'=\btinp{\mathsf{Int}}\btinp{\mathsf{Bool}}\btout{\mathsf{Bool}}\vart{t}$.
By \defref{def:recurdecomptypes},
 since $S$ is tail-recursive,  $\Gt{S} = \Rt{S'}$. 
Further, 
$\Rt{S'} = \trec{t}\btinp{\mathsf{\Gt{Int}}} \vart{t}, \Rt{\btinp{\mathsf{Bool}}\btout{\mathsf{Bool}}\vart{t}}$. 
By definition of $\Rt{\cdot}$, we obtain $\Gt{S} = \trec{t}\btinp{\mathsf{Int}} \vart{t}, \trec{t}\btinp{\mathsf{Bool}} \vart{t}, 
\trec{t}\btout{\mathsf{Bool}} \vart{t}, \Rt{t}$ (using $\Gt{\mathsf{Int}} = \mathsf{Int}$ and $\Gt{\mathsf{Bool}} = \mathsf{Bool}$).
Since 
$\Rt{\vart{t}} = \epsilon$, we obtain 
$\Gt{S} = \trec{t}\btinp{\mathsf{Int}} \vart{t}, \trec{t}\btinp{\mathsf{Bool}} \vart{t}, 
\trec{t}\btout{\mathsf{Bool}} \vart{t}$.  	
\hspace*{\fill} $\lhd$
\end{example}


\begin{example}[Decomposing an Unfolded Recursive Type]
 Let $T = \btinp{\mathsf{Bool}}\btout{\mathsf{Bool}}S$ be a derived unfolding of  $S$ from 
  \exref{ex:rtype}. Then, by \defref{def:recurdecomptypes}, $\Rts{}{s}{T}$ is the list of minimal recursive 
  types obtained as follows:  first, 
  $\Rts{}{s}{T} = \Rts{}{s}{\btout{\mathsf{Bool}}\trec{t}S'}$ and after one more step, $\Rts{}{s}{\btout{\mathsf{Bool}}\trec{t}S'} = \Rts{}{s}{\trec{t}S'}$. Finally, we have $\Rts{}{s}{\trec{t}S'} = \Rt{S'}$. 
  We get the same list of minimal types as in \exref{ex:rtype}: $\Rts{}{s}{T} = \trec{t}{\btinp{\mathsf{Int}}\vart{t}}, \trec{t}{\btinp{\mathsf{Bool}}\vart{t}}, \trec{t}{\btout{\mathsf{Bool}}
  \vart{t}}$. 
  \hspace*{\fill} $\lhd$
\end{example}


%

%
%

We now explain how to decompose processes whose names are typed with recursive types. 
In the core fragment, we decompose a name $u$ into a sequence of names $\tilde u = (u_1,\ldots,u_n)$: 
each $u_i \in \tilde u$ is used exactly by one trio to perform exactly one action; 
the session associated to $u_i$ ends after its single use, as prescribed by 
its minimal session type. The situation is different when names can have recursive types, 
for the names $\tilde u$ should be propagated in order to be used infinitely many times. 
As a simple example, consider the process
\begin{align*}
R = \binp{r}{x}\bout{r}{x}\appl{V}{r}    
\end{align*}
\noindent where name $r$ has type $S=\trec{t}{\btinp{\mathsf{Int}}\btout{\mathsf{\mathsf{Int}}}}{\vart{t}}$
and the higher-order type of $V$ is $\shot{S}$. 
Processes of this form are key in the encoding of recursion given in~\cite{DBLP:conf/esop/KouzapasPY16}.
A naive decomposition of $R$, using the approach we defined 
for processes without recursive types, would result into 
\begin{align*}
\B{1}{\epsilon}{R}= \propinp{1}{}&\binp{r_1}{x}\propout{2}{x}\inact \Par 
 \propinp{2}{x}\bout{r_2}{x}\propout{3}{}\inact \Par 
 \propinp{3}{}\appl{\V{}{\epsilon}{V}}{(r_3,r_4)} 
\end{align*}
There are several issues with this breakdown. One of them is typability: we have that 
$r_1 : \trec{t}{\btinp{\mathsf{Int}}}\vart{t}$, but subprocess $\propout{2}{x}\inact$  is not typable under a linear environment containing such a judgment. Another, perhaps more central, issue concerns $\tilde r$: 
the last trio (which mimics application) should apply to the sequence of names 
$(r_1,r_2)$, rather than to $(r_3,r_4)$.
We address both issues by devising a mechanism that propagates names with recursive types (such as $(r_1,r_2)$) among the trios that use some of them. This entails decomposing $R$ in such a way that the first two trios 
propagate $r_1$ and $r_2$ after they have used them; the trio simulating $\appl{V}{r}$ should then have a way to access the propagated names $(r_1,r_2)$.

We illustrate the key insights underpinning our solution by means of two examples.
The first one illustrates how to break down input and output actions on names with recursive types (the ``first part'' of 
$R$). 
The second example shows how to break down an application where a value is applied 
to a tuple of names with recursive names (the ``second part'' of $R$). 
 
\begin{example}[Decomposing Processes with Recursive Names (I)]
\label{ex:recdec}    
Let $P = \binp{r}{x}\bout{r}{x}P'$ be a process 
where $r$ has type $S=\trec{t}{\btinp{\mathsf{Int}}\btout{\mathsf{\mathsf{Int}}}}{\vart{t}}$
and $r \in \fn{P'}$. To define $\B{1}{\epsilon}{P}$ in a compositional way, names $(r_1,r_2)$
should be provided to its first trio; they cannot be known beforehand. 
To this end, we introduce a new  control trio that will hold these names:
\begin{align*} 
    \proprinp{r}{b}{\appl{b}{(r_1, r_2)}}
\end{align*}
where the \emph{shared} name $\prop^r$ provides a decomposition of the (recursive) name $r$.
The intention is that each name with a recursive type $r$ will get its own dedicated propagator channel $\prop^r$. 
Since there is only one recursive name in $P$, its decomposition will be of the following form:
\begin{align*}
	\D{P} = \news{\tilde \prop}\news{\prop^r} 
	\big( \proprinp{r}{b}{\appl{b}{(r_1,r_2)}} \Par
\propout{1}{}\inact \Par  \B{1}{\epsilon}{P} \big)
\end{align*}
The new control trio can be seen as a server that provides names: each trio that mimics some action on $r$ should request the sequence $\tilde r$ from the server on $\prop^r$. This request will be realized by a higher-order communication: trios should 
send an abstraction to the server; such an abstraction will contain further actions of a trio and it will be applied to the sequence $\tilde r$. Following this idea, 
we may refine the definition of $\D{P}$ by expanding $\B{1}{\epsilon}{P}$:
\begin{align*}
	\D{P} = \news{\tilde \prop}\news{\prop^r} 
	\big( \proprinp{r}{b}{\appl{b}{(r_1,r_2)}} \Par
\propout{1}{}\inact \Par  \propinp{1}{} \proprout{r}{N_1} \inact \Par 
	\propinp{2}{y} \proprout{r}{N_2} \inact  
	\Par \B{3}{\epsilon}{P'}
\big)
\end{align*}
The trios involving names with recursive types have now a different shape. 
After being triggered by a previous trio, rather than immediately mimicking an action, they will send
an abstraction to the server available on $\prop^r$. The abstractions $N_1$ and $N_2$ are defined as follows:
\begin{align*}
N_1 = \abs{(z_1, z_2)}\binp{z_1}{x}\propout{2}{x}\proprinp{r}{b}{
\appl{b}{(z_1,z_2)}}	 
\quad
N_2 = \abs{(z_1, z_2)}\bout{z_2}{x}\propout{3}{}
\proprinp{r}{b}{\appl{b}{(z_1,z_2)}}	
\end{align*}
Hence, the formal arguments for these values are meant to correspond to $\tilde r$. The server on name $\prop^r$ will appropriately instantiate these names. Notice that all names in $\tilde r$ are propagated, even if the abstractions only use some of them. For instance, $N_1$ only uses $r_1$, whereas $N_2$ uses $r_2$. After simulating an action on $r_i$ and activating the next trio, these values reinstate the server on $\prop^r$ for the benefit of future trios mimicking actions on $r$. 
		\hspace*{\fill} $\lhd$
\end{example}

\begin{example}
[Decomposing Processes with Recursive Names (II)]
\label{ex:rec2}
Let 
$S = \trec{\vart{t}}\btinp{\mathsf{Int}}\btout{\mathsf{Int}} \vart{t}$
and
$T = \trec{\vart{t}}\btinp{\mathsf{Bool}}\btout{\mathsf{Bool}} \vart{t}$, and 
define $Q = \appl{V}{(u,v)}$ as a process where $u : S$ and $v : T$, where $V$ is some value of type $\shot{(S, T)}$.  
The  decomposition of $Q$ is as in the previous example, except that now we need two servers, one for $u$ and one for $v$:
\begin{align*}
	\D{Q} = \news{\prop_1 \tilde \prop}\news{\prop^u \prop^v} \big( 
	\binp{\prop^u}{b}\appl{b}{(u_1,u_2)} \Par 
	\binp{\prop^v}{b}\appl{b}{(v_1,v_2)} \Par 
	\propout{1}{}\inact \Par 
	\B{1}{\epsilon}{Q}
	\big)	
\end{align*}

\noindent where $\tilde \prop = (\prop_2,\ldots,\prop_{\len{Q}})$. 
We should break down $Q$ in such a way that it could communicate 
with both servers to collect sequences $\tilde u$ and $\tilde v$.  
To  this end, we define a process in which abstractions are nested using output prefixes and whose innermost process is an application. After successive communications with multiple servers this innermost application will have collected all names in $\tilde u$ and $\tilde v$. 
We apply this idea to breakdown $Q$:
\begin{align*}
	\B{1}{\epsilon}{Q} = \propinp{1}{}\bbout{\prop^u}{\abs{(x_1,x_2)}{
	\proprout{v}{\abs{(y_1,y_2)}{\appl{\V{2}{\epsilon}{V}}{(x_1,x_2,y_1,y_2)}}}\inact}}\inact
\end{align*}
Observe that we use two nested outputs, one for each name with recursive types in $Q$. 
We now look at the reductions of $\D{Q}$ to analyze how the communication of nested abstractions allows us to collect all name sequences needed. After the first reduction along $\prop_1$ we have: 
\begin{align*}
	\D{Q} \red & \news{\tilde \prop}\news{\prop^u \prop^v} \big( 
	\binp{\prop^u}{b}\appl{b}{(u_1,u_2)} \Par 
	\binp{\prop^v}{b}\appl{b}{(v_1,v_2)} \Par  
	\\
	& \proprout{u}{\abs{(x_1,x_2)}{
	\proprout{v}{\abs{(y_1,y_2)}{\appl{\V{2}{\epsilon}{V}}{(x_1,x_2,y_1,y_2)}}}\inact}}\inact
 = R^1
\end{align*}
From $R^1$ we have a synchronization along  name $\prop^u$: 
\begin{align*}
	R^1 \red & \news{\tilde \prop}\news{\prop^u \prop^v} \big( 
	\appl{\abs{(x_1,x_2)}{
	\proprout{v}{\abs{(y_1,y_2)}{\appl{\V{2}{\epsilon}{V}}{(x_1,x_2,y_1,y_2)}}}\inact}}{(u_1,u_2)} \Par \\
	& \binp{\prop^v}{b}\appl{b}{(v_1,v_2)} \big)
 = R^2
\end{align*}
Upon receiving the value, the server applies it to $(u_1,u_2)$ 
obtaining the following process: 
\begin{align*}
	R^2 \red & \news{\tilde \prop}\news{\prop^u \prop^v} \big( 
		\proprout{v}{\abs{(y_1,y_2)}{\appl{\V{2}{\epsilon}{V}}{(u_1,u_2},y_1,y_2)}}\inact \Par 
	\binp{\prop^v}{b}\appl{b}{(v_1,v_2)}\big)  = R^3
\end{align*}
Up to here, we have partially  instantiated name variables of a value 
with the sequence $\tilde u$. Next, the first trio in $R^3$ can communicate with the server on name $\prop^v$:  
\begin{align*}
	R^3 \red & \news{\tilde \prop}\news{\prop^u \prop^v} \big( \appl{\abs{(y_1,y_2)}{\appl{\V{2}{\epsilon}{V}}{(u_1,u_2,y_1,y_2)}}}{(v_1,v_2)}\big)  
	\\
	\red & \news{\tilde \prop}\news{\prop^u \prop^v} \big( {\appl{\V{2}{\epsilon}{V}}{(u_1,u_2,v_1,v_2)}}\big) 
\end{align*}
This completes the instantiation of name variables with appropriate sequences of names with recursive types.
At this point, $\D{Q}$ can proceed to mimic the application in $Q$. 
\hfill $\lhd$
\end{example}

These two examples illustrate the main ideas of the decomposition of processes that involve names with recursive types.
\tabref{t:bdownrecur} presents a formal account of the extension of 
the definition of process decomposition given in \defref{def:decomp}.
Before explaining the table in detail, we require an auxiliary definition.


Given an unfolded recursive session type $S$,  
the  auxiliary function $f(S)$   returns the position of the top-most prefix of $S$ within its body.
(Whenever $S = \trec{t}{S'}$, we have $f(S)=1$.)

\begin{definition}[Index function]
\label{def:indexfunction}
Let $S$ be an (unfolded) recursive session type. The function $f(S)$ is defined as follows:  
\begin{align*}
	f(S) = \begin{cases}
 	f'_0(S'\subst{S}{\vart{t}}) & \text{if} \ S =\trec{t}{S'} \\
 	f'_0(S) & \text{otherwise}
 \end{cases}
\end{align*}
\noindent where:
$f'_l(\trec{t}{S}) = \len{\Rt{S}} - l + 1$,
		\qquad
		$f'_l(\btout{U}S)  = f'_{l+1}(S)$, 
		\qquad
	$f'_l(\btinp{U}S)  = f'_{l+1}(S)
	$.
\end{definition}


\begin{example}
\label{ex:fs}
		Let  
		$S' = \btinp{\mathsf{Bool}}\btout{\mathsf{Bool}}S$ where 
		$S$ is as in \exref{ex:rtype}.
		Then $f(S') = 2$ since the top-most prefix of $S'$ (`$\btinp{\mathsf{Bool}}$') is 
		the second prefix in the body of $S$. 
		\hspace*{\fill} $\lhd$
\end{example}

Given a typed process $P$, we write $\rfn{P}$ to denote the set of free names of $P$ whose types are recursive.
As mentioned above, for each 
$r \in \rfn{P}$ with $r : {S}$
we shall rely on a  control trio of the form  $\binp{\prop^r}{b}\appl{b}{\widetilde r}$, 
where $\widetilde{r} = r_1, \ldots,  r_{\len{\Gt{S}}}$.

\begin{definition}[Decomposition of a Process with Recursive Session Types]
	\label{def:decomp-rec}
	Let $P$ be a closed \HO process with $\widetilde u = \fn{P}$ and $\widetilde v = \rfn{P}$.
  The \emph{decomposition} of $P$, denoted $\D{P}$, is
defined as:
$$
  \D{P} = \news{\widetilde \prop}\news{\widetilde \prop_r}\Big(
  \prod_{r \in \tilde{v} } \binp{\prop^r}{b}\appl{b}{\widetilde r} \Par
  \propout{k}{} \inact \Par \B{k}{\epsilon}{P\sigma}\Big)
  $$
  \noindent where: $k >0$;
  $\widetilde \prop = (\prop_k,\ldots,\prop_{k+\len{P}-1})$;
  $\widetilde{\prop_r} = \bigcup_{r\in \tilde{v}}\prop^r$;
    $\sigma = \subst{\mathsf{init}(\widetilde u)}{\widetilde u}$.
\end{definition}

\noindent We now describe the required extensions for the function $\B{k}{\tilde x}{\cdot}$. We will use predicate $\mathsf{tr}(S)$ on types to indicate that  $S$ is a tail-recursive session type. 
\tabref{t:bdownrecur} describes the breakdown of prefixes  whose type is recursive; all other prefixes can be treated as in 
\tabref{t:bdowncore}.

  \begin{table}[!t]
\begin{tabular}{ |l|l|l|}
  \rowcolor{gray!25}
  \hline 
    \multicolumn{2}{|l|}{
    \begin{tabular}{l}
        \noalign{\smallskip}
        $\B{k}{\tilde x}{\bout{r}{V}{Q}}$
        \smallskip
    \end{tabular} } 
    \\
  \hline 
    \begin{tabular}{l}
        \noalign{\smallskip}
        $\propinp{k}{\widetilde x}
      		\abbout{\prop^r}{N_V}   \Par$ 
        $\B{k+\degree+1}{\tilde w}{Q}$
        \smallskip
        \\
        where:
        \\
        $N_V = \abs{\widetilde z}
      		{\bbout{z_{f(S)}}{\V{k+1}{\tilde y}{V}}}{}$ \\ 
      		\quad \qquad \qquad \qquad $\propout{k+\degree+1}
      		{\widetilde w} \binp{\prop^r}{b}
			(\appl{b}{\widetilde z}$)
			\smallskip
    \end{tabular} & 
    \begin{tabular}{l}
           \noalign{\smallskip}
        $r:S \wedge \mathsf{tr}(S)$ \\
        $\widetilde y = \fv{V}, \ \widetilde w = \fv{Q}$ \\
        $\degree = \len{V}$ \\
        $\widetilde z = (z_1,\ldots,z_{\len{\Rts{}{s}{S}}})$
        \smallskip
    \end{tabular} 
    \\
  \rowcolor{gray!25}
  \hline 
    \multicolumn{2}{|l|}{
    \begin{tabular}{l}
    \noalign{\smallskip}
        $\B{k}{\tilde x}{\binp{r}{y}Q}$
        \smallskip
    \end{tabular} 
  } 
  \\
  \hline 
    \begin{tabular}{l}
      \noalign{\smallskip}
      $\propinp{k}{\widetilde x}\abbout{\prop^r}
      {N_y} 
      \Par \B{k+1}{\tilde x'}{Q}$
      \\
      where:
      \\
      $N_y = \abs{\widetilde z}{\binp{z_{f(S)}}{y}\propout{k+1}{\widetilde x'}
      \binp{\prop^r}{b}(\appl{b}{\widetilde z})}$
      \smallskip
    \end{tabular} & 
    \begin{tabular}{l}
        \noalign{\smallskip}
        $r:S \wedge \mathsf{tr}(S)$ \\
        $\widetilde x' = \fv{Q}$ \\
	   $\widetilde z = (z_1,\ldots,z_{\len{\Rts{}{s}{S}}})$ 
	\smallskip
    \end{tabular} \\
    \rowcolor{gray!25}
  \hline 
    \multicolumn{2}{|l|}{
    \begin{tabular}{l}
    \noalign{\smallskip}
        $\B{k}{\tilde x}{\appl{V}{(\widetilde r, u_i)}}$
            \smallskip
    \end{tabular} 
} 
\\
  \hline 
    \begin{tabular}{l}
        \noalign{\smallskip}
        $\propinp{k}{\widetilde x}\overbracket{\prop^{r_1}!\big\langle 
        \lambda \widetilde z_1. \prop^{r_2}!\langle\lambda \widetilde z_2.\cdots. 
	 	\prop^{r_n}!\langle \lambda \widetilde z_n.}^{n = |\tilde r|} 
	 	Q \rangle \,\rangle \big\rangle$ \\
		where:
		\\
	 	$Q = \appl{\V{k+1}{\tilde x}{V}}{(\widetilde z_1,\ldots,
	 	\widetilde z_n, \widetilde m)}$
    \smallskip
    \end{tabular} & 
    \begin{tabular}{l}
          \noalign{\smallskip}
  		$\forall r_i \in \widetilde r.(r_i: S_i \wedge \mathsf{tr}(S_i) \wedge$\\
        \qquad $\widetilde{z_i} = (z^i_1,\ldots,z^i_{\len{\Rts{}{s}{S_i}}}))$\\
        $u_i : C$ \\ 
        $\widetilde m = (u_i, \ldots, u_{i+\len{\Gt{C}}-1})$
        \smallskip
    \end{tabular} \\
    \rowcolor{gray!25}
  \hline 
    \multicolumn{2}{|l|}{
    \begin{tabular}{l}
    \noalign{\smallskip}
        $\B{k}{\tilde x}{\news{s:\trec{t}{S}}{P'}}$
            \smallskip
    \end{tabular} 
    } \\
  \hline 
    \begin{tabular}{l}
        \noalign{\smallskip}
 	      $\news{\widetilde{s}:\mathcal{R}(S)}
 		\news{c^s}\binp{\prop^s}{b}(\appl{b}{\widetilde s}) \Par$ \\ 
 		\qquad \qquad \quad 
 		$\news{c^{\bar{s}}}\binp{\prop^{\bar{s}}}{b}(\appl{b}{\widetilde{\dual s}}) \Par$ 
 		$\B{k}{\tilde x}{P'}$
    \smallskip
    \end{tabular} & 
    \begin{tabular}{l}
        \noalign{\smallskip}
        $\mathsf{tr}(\trec{t}{S})$ \\
        $\widetilde{s} = (s_1,\ldots,s_{\len{\Rt{S}}})$ \\
        $\widetilde {\dual{s}} = (\dual{s_1},\ldots,\dual{s_{\len{\Rt{S}}}})$ \\
    \smallskip
    \end{tabular} \\
    \rowcolor{gray!25}
  \hline 
    \multicolumn{2}{|l|}{
    \begin{tabular}{l}
        \noalign{\smallskip}
        $\V{k}{\tilde x}{\abs{(\widetilde y, z):\slhotup{(\widetilde S, C)}}P}$
        \smallskip
    \end{tabular} 
    } \\
  \hline 
    \begin{tabular}{l}
        \noalign{\smallskip}
        $\abs{(\widetilde{y^1},\ldots,\widetilde{y^n}, \widetilde z):
  	 \slhotup{(\widetilde{T})}}{N}$ 
	 \\
	 \smallskip
	   where:
	   \\
	   $\widetilde{T} = (\Gt{S_1},\ldots,\Gt{S_n}, \Gt{C})$
	   \\
            $N = \news{\widetilde \prop}
        \prod_{i \in \len{\widetilde y}}(\binp{\prop^{y_i}}{b}(\appl{b}{\widetilde y^i})) \Par \apropout{k}{\widetilde x}
    \Par$ \\
  \quad \quad \qquad \qquad \qquad  \qquad  \ \  $\B{k}{\tilde x}{P \subst{z_1}{z}}$
  \\
  \smallskip
    \end{tabular} & 
    \begin{tabular}{l}
           \noalign{\smallskip}
        $\forall y_i \in \widetilde y.(y_i: S_i \wedge \mathsf{tr}(S_i) \wedge$\\
         \qquad $\widetilde{y^i} = (y^i_1,\ldots,y^i_{\len{\Gt{S_i}}}))$\\
         $\widetilde z = (z_1,\ldots,z_{\len{\Gt{C}}})$ \\
         $\widetilde{\prop} = \begin{cases} 
        \epsilon & \text{if} \ \leadsto = \multimap \\
     (\prop_k,\ldots,\prop_{k+\len{P}-1}) & \text{if} \ \leadsto = \rightarrow
 \end{cases}$ 
        \smallskip
    \end{tabular} \\
    \hline 
  \end{tabular}
  \caption{The breakdown function for processes and values (extension with
recursive types). \label{t:bdownrecur}}
  \end{table}

\begin{description}
	\item[Output] The breakdown of process $\bout{r}{V}Q$, when  
	$r$ has a recursive type $S$, is as follows: 
	\begin{align*}
		\B{k}{\tilde x}{\bout{r}{V}Q}=~~&\propinp{k}{\widetilde x}
      		\abbout{\prop^r}{N_V} \Par
     	 \B{k+\degree+1}{\tilde w}{Q} \\
    	 &\text{where} \ N_V = \abs{\widetilde z}
      		{\bbout{z_{f(S)}}{\V{k+1}{\tilde y}{V}}}\propout{k+\degree+1}
      		{\widetilde w} \binp{\prop^r}{b}
			(\appl{b}{\widetilde z}) 
	\end{align*}
	The decomposition consists of a leading trio that mimics the output action running in parallel with the breakdown of  $Q$. 
After receiving the context $\widetilde x$, the leading trio sends an abstraction $N_V$ along $\prop^r$. 
Value $N_V$ performs several tasks.
First, it collects the sequence $\tilde r$; 
then, it mimics the output action of $P$ along one of them ($r_{f(S)}$) and 
triggers the next trio, with context $\widetilde w$; finally, it 
reinstates the server on $\prop^r$ for the next trio that uses~$r$.  
 Notice that differently from what is done in \tabref{t:bdowncore}, indexing is not relevant when breaking down names with recursive types. 
 	Also, since by definition $\V{k}{\tilde y}{y}=y$, $y \sigma = y$, and 
	$\len{y}=0$, when the communicated value $V$ is a variable $y$ we obtain the following: 
	$$
	\B{k}{\tilde x}{\bout{r}{y}Q} = \propinp{k}{\widetilde x} \abbout{\prop^r}
	{\abs{\widetilde z}{\bout{z_{f(S)}}{y}}
	\propout{k+1}{\widetilde w}\binp{c^r}{b}(\appl{b}{\widetilde z})} 
	\Par \B{k+1}{\tilde w}{Q}
	$$
	
%

\item[Input] The breakdown of process $\binp{r}{y}Q$, when $r$ has recursive session type $S$, is as follows: 
$$
\B{k}{\tilde x}{\binp{r}{y}Q}=\propinp{k}{\widetilde x}\abbout{\prop^r}
      {\abs{\widetilde z}{\binp{z_{f(S)}}{y}\propout{k+1}{\widetilde x'}
      \binp{\prop^r}{b}(\appl{b}{\widetilde z})}} 
      \Par \B{k+1}{\tilde x'}{Q}
$$
The breakdown follows the lines of the output case, but also of the linear case  in \tabref{t:bdowncore},
with additional structure needed to implement the reception of   
$\tilde r$, using one of the received names ($r_{f(S)}$) as a subject for the input action and propagating 
those names further.  

\item[Application] 
For simplicity we consider applications $\appl{V}{(\widetilde r, u_i)}$, where names in $\widetilde r$ have 
recursive types and only name $u_i$ has a non-recursive type; the general case involving different orders in names and multiple names with non-recursive types is as expected. We have:  
\begin{align*}
 \B{k}{\tilde x}{\appl{V}{(\widetilde r, u_i)}} = &
 \propinp{k}{\widetilde x}\overbracket{\prop^{r_1}!\big\langle 
        \lambda \widetilde z_1. \prop^{r_2}!\langle\lambda \widetilde z_2.\cdots. 
	 	\prop^{r_n}!\langle \lambda \widetilde z_n.}^{n = |\tilde r|} 
	 	\appl{\V{k+1}{\tilde x}{V}}{(\widetilde z_1,\ldots,
	 	\widetilde z_n, \widetilde m)} \rangle \,\rangle \big\rangle
\end{align*}

We rely on types  to decompose 
every name in $(\widetilde r, u_i)$. 
Letting $|\tilde r| = n$ and $i \in \{1,\ldots,n\}$, 
for each $r_i \in \widetilde r$  (with $r_i:S_i$)
we generate a sequence 
$\widetilde z_i=(z^i_1,\ldots,z^i_{\len{\Rts{}{s}{S_i}}})$ as in the output case.
Since name $u_i$ has a non-recursive session type, we decompose it as in \tabref{t:bdowncore}. 
Subsequently, we define an output action on propagator 
$\prop^{r_1}$ that sends a value containing $n$ abstractions that occur nested within output prefixes:
for each $j \in \{1,\ldots,n-1\}$, each abstraction binds $\widetilde z_j$ and sends the next abstraction along $\prop^{r_{j+1}}$. 
The innermost abstraction abstracts over $\widetilde z_n$ and encapsulates process 
$\appl{\V{k+1}{\tilde x}{V}}{(\widetilde z_1,\ldots,
	 	\widetilde z_n, \widetilde m)}$, which mimics the application in the source process.
 By this abstraction nesting we bind all variables $\widetilde z_i$ in $Q$.
This structure can be seen as an encoding of partial application: by virtue of a single synchronization on $\prop^{r_i}$  
part of variables (i.e., $\widetilde z_i$)  will be instantiated.   

\smallskip
The breakdown of a value application of the form 
$\appl{y}{(\widetilde r, u_i)}$ results into a specific form of the breakdown:
\begin{align*}
	 \B{k}{\tilde x}{\appl{y}{(\widetilde r, u_i)}} = &
	 \propinp{k}{\widetilde x}\overbracket{\prop^{r_1}!\big\langle 
        \lambda \widetilde z_1. \prop^{r_2}!\langle\lambda \widetilde z_2.\cdots. 
	 	\prop^{r_n}!\langle \lambda \widetilde z_n.}^{n = |\tilde r|} 
	 	\appl{y}{(\widetilde z_1,\ldots,
	 	\widetilde z_n, \widetilde m)} \rangle \,\rangle \big\rangle
\end{align*}

\item[Restriction] The restriction process $\news{s:\trec{t}S}{P'}$ is 
translated as follows: 
\begin{align*}
\B{k}{\tilde x}{\news{s:\recp{t}S}{P'}} = \news{\widetilde{s}:\mathcal{R}(S)}
 		 \news{c^s}\binp{\prop^s}{b}(\appl{b}{\widetilde s}) \Par 
 	\news{c^{\bar{s}}}\binp{\prop^{\bar{s}}}{b}(\appl{b}{\widetilde{\dual s}}) \Par
 		\B{k}{\tilde x}{P'}
\end{align*}

We decompose $s$ into $\widetilde{s} = (s_1,\ldots,s_{\len{\Rt{S}}})$ and 
$\dual{s}$ 
 into $\widetilde{\dual{s}} = (\dual{s_1},\ldots,\dual{s_{\len{\Rt{S}}}})$ . 
 The breakdown introduces two servers in parallel with the breakdown of
  $P'$; these servers provide names for $s$ and $\dual{s}$ along $\prop^s$ and $\prop^{\dual{s}}$, respectively.
  The server on $\prop^s$ (resp. $\prop^{\dual{s}}$) receives a value   
and applies it to the sequence $\widetilde s$ (resp. $\widetilde{\dual{s}}$).
 We restrict over $\widetilde s$ and propagators $\prop^s$ and $\prop^{\dual{s}}$.

\item[Value] The polyadic value $\abs{(\widetilde y, z):\slhotup{(\widetilde S, C)}}P$, where $\leadsto \in \{\lollipop, \sharedop\}$,  is 
decomposed as follows: 
\begin{align*}
\V{k}{\tilde x}{\abs{(\widetilde y, z):\slhotup{(\widetilde S, C)}}P} = & \abs{(\widetilde{y^1},\ldots,\widetilde{y^n}, \widetilde z):
  	\slhotup{(\Gt{S_1},\ldots,\Gt{S_n}, \Gt{C})}}{N}  
	\\
  	\text{where:} &  ~~N = \news{\widetilde \prop}
   \prod_{i \in \len{\widetilde y}}(\binp{\prop^{y_i}}{b}(\appl{b}{\widetilde y^i})) \Par \apropout{k}{\widetilde x}
   \Par \B{k}{\tilde x}{P \subst{z_1}{z}}
\end{align*}  
We assume variables in $\widetilde y$
have recursive session types $\widetilde S$ and variable $z$ has some non-recursive 
session type $C$; 
the general case involving different orders in variables and multiple variables with non-recursive types is as expected.
Every variable $y_i$ (with $y_i : S_i$) is decomposed into 
$\widetilde y^i=(y_1,\ldots,y_{\len{\Gt{S_i}}})$. 
Variable $z$ is decomposed as in \tabref{t:bdowncore}.
The breakdown is similar to the (monadic) shared value given in \tabref{t:bdowncore}. 
In this case, for every $y_i \in \widetilde y$ there is a server $\binp{\prop^{y_i}}{b}(\appl{b}{\widetilde y^i})$ as a subprocess in the abstracted composition.  The rationale for these servers is as described in previous cases. 

\end{description}

To sum up, each trio using a name with a recursive session type first receives a sequence of names;  
then, it 
uses one of such names to mimic the appropriate action; 
finally, it propagates the entire sequence by reinstating a server defined as a control trio. 
Interestingly, this scheme for name propagation follows the implementation of the encoding of name-passing in \HO.

\begin{example}[Breakdown of Recursion Encoding]
	Consider the recursive  process
	$P = \recp{X}\binp{a}{m}\bout{a}{m}X$, which  is not an $\HO$ process.  $P$ can be encoded into $\HO$ as follows~\cite{DBLP:conf/esop/KouzapasPY16}:
	\begin{align*}
	\map{P} = \binp{a}{m}\bout{a}{m} \news{s}
	{(\appl{V}{(a,s)} \Par \bout{s}{V}\inact)}
	\end{align*}
	where the value $V$ is an abstraction that potentially reduces to $\map{P}$: 
	\begin{align*}
	V = \abs{(x_a,y_1)}\binp{y_1}{z_x}\binp{x_a}{m}
	\bout{x_a}{m}\news{s}
	{(\appl{z_x}{(x_a,s)} \Par \bout{\dual s}{z_x}\inact)} 
	\end{align*}
	We compose $\map P$ with an appropriate client process to illustrate the encoding of recursion:
	\begin{align*}
	& \map{P} \Par \bout{a}{W} \binp{a}{b}R \\
	& \red^2 \news{s} 
	{(\appl{V}{(a,s)} \Par \bout{s}{V}\inact)} \Par R \\
	&  \red  \news{s} {(\binp{s}{z_x}\binp{a}{m}
		\bout{a}{m}\news{s'}
		{(\appl{z_x}{(a,s')} \Par \bout{\dual {s'}}{z_x}\inact)} \Par \bout{s}{V}\inact)} \Par R \\
	& \red \binp{a}{m}
		\bout{a}{m}\news{s'}
		{(\appl{V}{(a,s')} \Par \bout{\dual {s'}}{V}\inact)} \Par R = \map{P} \Par R
	\end{align*}
	\noindent where $R$ is some unspecified process such that $a \in \rfn{R}$. 
	We now analyze $\D{\map{P}}$ and its reduction chain. By \defref{def:sizeproc}, we have $\len{\map P} = 7$, and $\len{V} = 0$. Then, we choose $k=1$ and observe that $\sigma = \subst{a_1\dual{a_1}}{a\dual a}$. Following \defref{def:decomp-rec}, we get:
	\begin{align*}
	&\D{\map P} = \news{\prop_1,\ldots,\prop_7}
	\news{\prop^a}(\binp{\prop^a}{b}\appl{b}{(a_1,a_2)} \Par 
	\propout{1}{} \inact \Par \B{1}{\epsilon}{\map P\sigma})\\
	&\B{1}{\epsilon}{\map P} =
	\propinp{1}{} \bout{\prop^a}{\abs{(z_1,z_2)}\binp{z_1}
		{m}\propout{2}{m}\binp{\prop^a}{b}\appl{b}{(z_1,z_2)}}\inact
	\\
	&\qquad \qquad 
	\Par
	\propinp{2}{m}
	\bout{\prop^a}{\abs{(z_1,z_2)}\bout{z_2}
		{m}\propout{3}{}\binp{\prop^a}{b}\appl{b}{(z_1,z_2)}}\inact
	 \\
	&\qquad \qquad \Par \news{s_1}\big(\propinp{3}{} \propout{4} {} 
		\propout{5}{} \inact
		\Par
		\propinp{4}{}\bout{\dual{\prop^a}}
		{\abs{(z_1,z_2)}\appl{\V{5}{\epsilon}{V}}{(z_1,z_2,s_1)}} \inact
	\\
	&\qquad \qquad \qquad \quad 
	\Par \propinp{5}{}\bout{\dual s_1}{\V{6}{\epsilon}{V}}\propout{7}{}\inact 
		\Par \propinp{7}{}\inact\big)
	\end{align*}
	The decomposition relies twice on the breakdown of value $V$, so we give $\V{k}{\epsilon}{V}$ here for arbitrary $k>0$. For this, we observe that $V$ is an abstraction of a process $Q$ with $\len Q = 7$. 
	\begin{align*}
	&\V{k}{\epsilon}{V} = \abs{(x_{a_1},x_{a_2},y_1)}
	{	\news{\prop_{k},\ldots,\prop_{k+6}} (
		\binp{\prop^{x_{a}}}{b}\appl{b}{(x_{a_1},x_{a_2})} 
		\Par \propout{k}{}\inact \Par \B{k}{\epsilon}{Q}}) \\
	&\B{k}{\epsilon}{Q} = \propinp{k}{}\binp{y_1}{z_x}
	\propout{k+1}{z_x}\inact \\
	& \qquad
	\Par
	\propinp{k+1}{z_x}\bout{\prop^a_1}
	{\abs{(z_1,z_2)}{\binp{z_1}{m}\propout{k+2}{z_x,m}
			\binp{\prop^a_2}{b}
			\appl{b}{(z_1,z_2)}}}\inact \\
	& \qquad \Par
	 \propinp{k+2}{z_x}\bout{\prop^a_2}
	{\abs{(z_1,z_2)}{\bout{z_2}{m}} \propout{k+3}{z_x}\binp{\prop^a_3}
		{b}\appl{b}{(z_1,z_2)}} \inact  \\
	& \qquad \Par \news{s_1}\big(\propinp{k+3}{x_z} \propout{k+4}{z_x}
		\propout{k+5}{z_x}\inact  \\
	&\quad \quad \Par \propinp{k+4}{z_x}
	\bout{\prop^a_3}{\abs{(z_1,z_2)}\appl{z_x}{(z_1,z_2,s_1)}}
	\inact \Par \propinp{k+5}{z_x} \bout{\dual{s_1}}{z_x}
	\propout{k+6}{}\inact \Par \propinp{k+6}{}\inact\big)
	\end{align*}
	We follow the reduction chain on $\D{\map P}$ until it is ready to mimic the first action with channel $a$, which is an input. First, $c_1$ will synchronize, after which $c^a$ sends the abstraction to which then $(a_1,a_2)$ is applied. We obtain $\D{\map P} \red^3 \news{c_2,\ldots, c_7,c^a}P'$, where
	\begin{align*}
		& P' = \binp{a_1}
		{m}\propout{2}{m}\binp{\prop^a}{b}\appl{b}{(a_1,a_2)}
		\\
		&\qquad \qquad \Par
		\propinp{2}{m}
		\bout{\prop^a}{\abs{(z_1,z_2)}\bout{z_2}
			{m}\propout{3}{}\binp{\prop^a}{b}\appl{b}{(z_1,z_2)}}\inact
		\\
		&\qquad \qquad \Par \news{s_1}\big(\propinp{3}{} \propout{4} {} 
		\propout{5}{} \inact
		\Par
		\propinp{4}{}\bout{\dual{\prop^a}}
		{\abs{(z_1,z_2)}\appl{\V{5}{\epsilon}{V}}{(z_1,z_2,s_1)}} \inact
		\\
		&\qquad \qquad \qquad \quad \Par \propinp{5}{}\bout{\dual s_1}{\V{6}{\epsilon}{V}}\propout{7}{}\inact 
		\Par \propinp{7}{}\inact\big)
	\end{align*} 
	Note that this process is awaiting an input on channel $a_1$, after which $c_2$ can synchronize with its dual. At that point, $c^a$ is ready to receive another abstraction that mimics an input on $a_1$. This strongly suggests a tight operational correspondence between a process $P$ and its decomposition in the case where $P$ performs higher-order recursion.
$\hfill \lhd$
\end{example}


Below we write $\envR$ to denote a session environment that concerns only recursive types.
We state our main results:

\begin{theorem}[Typability of Breakdown]
\label{t:typerec}
   Let $P$ be an initialized \HO process and $V$ be a value.
  \begin{enumerate}
  \item
  If $\Gamma;\Lambda;\Delta \cat \envR \proves P \hastype \Proc$ then
  $\Gt{\Gamma_1}\cat \envPropR;\es;\Gt{\Delta} \cat \Theta \proves \B{k}{\tilde x}{P} \hastype \Proc$
  where: $\widetilde r = \text{dom}(\envR)$;
  $\envPropR = \prod_{r \in \tilde r} c^r:\chtype{\lhot{\Rts{}{s}{\envR(r )}}}$;
  $\widetilde x = \fv{P}$; 
  $k > 0$;
  $\Gamma_1=\Gamma \setminus \widetilde x$; 
  and
  $\balan{\Theta}$ with
  $\text{dom}(\Theta) =
  \{\prop_k,\ldots,\prop_{k+\len{P}-1}\} \cup
  \{\dual{\prop_{k+1}},\ldots,\dual{\prop_{k+\len{P}-1}}\}$ such that $\Theta(\prop_k)=
  \btinp{\widetilde M} \tinact$, where $\widetilde M =
  (\Gt{\Gamma}\cat\Gt{\Lambda})(\widetilde x)$.

  \item If $\Gamma;\Lambda;\Delta \proves V \hastype \lhot{C}$ then
  $\Gt{\Gamma};\Gt{\Lambda};\Gt{\Delta} \cat \Theta \proves \V{k}{\tilde
  x}{V} \hastype \lhot{\Gt{C}}$, 
  where: 
  $\widetilde x = \fv{V}$; 
  $k > 0$; and
   $\balan{\Theta}$ with
   $\text{dom}(\Theta) =  \{\prop_k,\ldots,\prop_{k+\len{V}-1}\} \cup
  \{\dual{\prop_{k}},\ldots,\dual{\prop_{k+\len{V}-1}}\}$ such that $\Theta(\prop_k)=
  \btinp{\widetilde M} \tinact$ and $\Theta(\dual{\prop_k})= \btout{\widetilde
  M} \tinact$, where $\widetilde M = (\Gt{\Gamma}\cat\Gt{\Lambda})(\widetilde x)$.

  \item If $\Gamma;\es;\es \proves V \hastype \shot{C}$ then
  $\Gt{\Gamma};\es;\es \proves \V{k}{\tilde x}{V} \hastype \shot{\Gt{C}}$ where $\widetilde x = \fv{V}$
  and $k > 0$.

\end{enumerate}
\end{theorem}

\begin{proof}
  By mutual induction on the structure of $P$ and $V$.
\appendx{See Appendix~\ref{app:typerec} for details.}
\end{proof}

\begin{theorem}[Typability of the Decomposition with Recursive Types]
	\label{t:decomprec}
	Let $P$ be a closed \HO process with $\widetilde u = \fn{P}$ and $\widetilde v = \rfn{P}$.
	If $\Gamma;\es;\Delta \cat \envR \proves P \hastype \Proc$, where $\envR$ only involves recursive session types,  then
	$\Gt{\Gamma \sigma};\es;\Gt{\Delta\sigma}  \cat \Gt{\envR \sigma} \proves \D{P} \hastype \Proc$, where
	$\sigma = \subst{\mathsf{init}(\widetilde u)}{\widetilde u}$. 
\end{theorem}

\begin{proof}
Directly from the definitions, using \thmref{t:typerec}.
\appendx{See Appendix~\ref{app:decomprec} for details.}
\end{proof}

\begin{remark}[Non-Tail-Recursive Session Types]
\label{r:ntrsts}
Our definitions and results apply to tail-recursive session types.
We can accommodate the non-tail-recursive type $\trec{t}{\btinp{\shot{(\widetilde T,\vart{t})}}\tinact}$ into our approach: in 
 \defref{def:recurdecomptypes}, we need to have $\Gt{\trec{t}{S}} = \trec{t}{\Gt{S}}$ if $\trec{t}{S}$ is non-tail-recursive.
The decomposition functions for non-recursive session types suffice in this case.
\end{remark}


\section{Optimizations of the Decomposition}\label{s:opt}
Here we briefly discuss two optimizations of the decompositions. They simplify the structure of trios and the underlying communication discipline. Interestingly, they are both enabled by the higher-order nature of $\HO$. 
In fact, they  hinge on \emph{thunk processes}, i.e., inactive processes that can be activated upon reception.
We write $\thunk{P}$
to stand for 
the thunk process
$\abs{x:\chtype{\shot{\tinact}}}P$, with $x \not\in \fn{P}$. 
We write $\appthunk{\thunk{P}}$ 
to denote the application of a thunk to a (dummy) name of type $\shot{\tinact}$. 
This way, we have 
$\appthunk{\thunk{P}} \red P$.

	\subparagraph*{From Trios to Duos} We can simplify the breakdown functions by replacing trios with \emph{duos}, i.e., processes with exactly two sequential prefixes.
	The idea is to transform trios such as $\binp{c_k}{\widetilde x}\bout{u}{V}
			\about{c_{k+1}}{\widetilde y}{}$ into the composition of a 
			 	  duo with a control trio: 
				  			\begin{equation}
			\binp{c_k}{\widetilde x}\abbout{c_{k+1}}{\thunk{\bout{u}{V}\about{c_{k+2}}{\widetilde z}}} \Par \binp{c_{k+1}}{b}{(\appthunk{b})}
			\label{eq:duo}
			\end{equation}
		The first action is as before; the
	two remaining prefixes are encapsulated into a thunk. This thunk  is sent via a propagator to the control trio that  activates it upon reception. This transformation involves an additional propagator, denoted $c_{k+2}$ above. This requires minor modifications in the definition of the degree  function 
	  $\len{\cdot}$ (cf. \defref{def:sizeproc}).

	In some cases, the breakdown function in \secref{ss:core} already produces duos. Breaking down input and output prefixes and  parallel composition involves proper trios; following the scheme illustrated by \eqref{eq:duo}, we can 
	define a map $\dmap{\cdot}$ to 
	transform these trios into duos:
	\begin{align*} 
			\dmap{\propinp{k}{\widetilde x}\bout{u_i}{V}
			\apropout{k+1}{\widetilde z}{}} & =
			\propinp{k}{\widetilde x}\apropbout{k+1}{\thunk{\bout{u_i}{V}\apropout{k+2}{\widetilde z}}} \Par \propinp{k+1}{b}{(\appthunk{b})} 
			\\
			\dmap{\propinp{k}{\widetilde x}\binp{u_i}{y}
			\apropout{k+1}{\widetilde x'}{}} & =  
			\propinp{k}{\widetilde x}\apropbout{k+1}{\thunk{\binp{u_i}{y}\apropout{k+2}{\widetilde x'}}} \Par \propinp{k+1}{b}{(\appthunk{b})} 
			\\
			\dmap{\propinp{k}{\widetilde x}\propout{k+1}{\widetilde y}
			\apropout{k+\degree+1}{\widetilde z}} &= 
			\propinp{k}{\widetilde x}\apropbout{k+1}{
			\thunk{
			\propout{k+2}{\widetilde y}  
			\apropout{k+\degree+2}{\widetilde z}}}  \Par 
			\propinp{k+1}{b}{(\appthunk{b})} 
	\end{align*} 
	  
	In the breakdown given in \secref{ss:exti} there is a 
	proper trio, which can be transformed as follows:
	\begin{align*}
		\dmap{\bsel{u_i}{l_j}\binp{u_i}{z}\propbout{k}{\widetilde x}
			{(\appl{z}{\widetilde y})}} =  
			\bsel{u_i}{l_j}\apropbout{k}{\thunk{\binp{u_i}{z}
			\propout{k+1}{\widetilde x}{\appl{z}{\widetilde y}}}} \Par \propinp{k}{b}(\appthunk{b})
	\end{align*}
	Similarly, 
	in the breakdown function extended with recursion (cf. \secref{ss:extii}) there is only one 
	trio pattern, which can be transformed into a duo following the very same idea.  
	
	\subparagraph*{From Polyadic to Monadic Communication}
	Since we consider \emph{closed} \HO processes, 
	we can dispense with polyadic communication in the breakdown function. 
	We can define a \emph{monadic decomposition}, $\mD{P}$, that simplifies  \defref{def:decomp} as follows: 
	\begin{align*}
	   \mD{P} =  \news{\widetilde c}\big(\propinp{k}{b}{(\appthunk{b})} \Par \mB{k}{}{P\sigma}\big)
	\end{align*}
	\noindent where $k > 0$, $\widetilde \prop = (\prop_k,\dots,\prop_{k+\len{P}-1})$, and $\sigma$ is  as in \defref{def:decomp}. 
	Process $\propinp{k}{b}{(\appthunk{b})}$ activates a thunk received from $\mB{k}{}{\cdot}$, the \emph{monadic} breakdown function that simplifies the one in \tabref{t:bdowncore} by using only one parameter, namely $k$:  
	\begin{align*}
		\mB{k}{}{\binp{u_i}{x}Q} & = 
		\news{\prop_x}\big(\apropbout{k}{\thunk{\binp{u_i}{x}\propinp{k+1}{b}({\about{\prop_{x}}{x} \Par (\appthunk{b})})}} \Par  \mB{k+1}{}{Q\sigma}\big)
		\\
		\mB{k}{}{\bout{u_i}{x}Q} & = 
		\apropbout{k}{\thunk{\binp{\prop_x}{x} \bout{u_i}{x}\propinp{k+1}{b}(\appthunk{b})}} 	\Par 		\mB{k+1}{}{Q\sigma} 
		\\
		\mB{k}{}{\bout{u_i}{V}Q} & = 
		\apropbout{k}{\thunk{\bout{u_i}{\mV{k+1}{}{V\sigma}}\propinp{k+1}{b}(\appthunk{b})}}
		\Par 
		\mB{k+1}{}{Q\sigma} 
		\\
	\mB{k}{}{\appl{x}{u}} & = \apropbout{k}{\thunk{\propinp{x}{x}(\appl{x}{\widetilde u})}}
	\\ 
		\mB{k}{}{\appl{V}{u}} & = \apropbout{k}{\thunk{\appl{\mV{k+1}{}{V}}{\widetilde u})}}
		\\
	\mB{k}{}{\news{s}{P'}} & = \news{\widetilde s}\mB{k}{}{P'\sigma}	
	\\
	\mB{k}{}{Q\Par R} & = \apropout{k}{\thunk{\propinp{k+1}{b}\appthunk{b} 
	\Par \propinp{k+\len{Q}+1}{b}\appthunk{b}}} 
	\Par \mB{k+1}{}{Q} \Par \mB{k+\len{Q}+1}{}{R}  	
	\end{align*}
	\noindent 
	Above, $\sigma$ is as in \tabref{t:bdowncore}.
	 $\mB{k}{}{\cdot}$ propagates values using thunks and a dedicated propagator  $\prop_x$ for each variable $x$.
	 	We describe only the definition of $\mB{k}{}{\binp{u_i}{x}Q}$: it illustrates key ideas common to all other cases. It consists of 
	an output of a thunk on $\prop_k$ composed in parallel with $\mB{k+1}{}{Q\sigma}$. 
	The thunk will be activated by a process $\propinp{k}{b}{(\appthunk{b})}$ at the top-level; this activation triggers the input action on $u_i$, and prepares the activation for the next thunk (exchanged on name $\prop_{k+1}$).
	Upon reception, such a thunk is activated in parallel with $\about{\prop_x}{x}$, which 
	 propagates the value received on $u_i$. The scope of $\prop_x$ 
	 is restricted to include input actions on $\prop_x$  in
	  $\mB{k+1}{}{Q\sigma}$; such actions are the first in the thunks present in, e.g.,
	  $\mB{k}{}{\bout{u_i}{x}Q}$. 
	 We also need to revise the breakdown 
	 function for values $\V{k}{\tilde x}{\cdot}$. 
%
	The breakdown functions  given in \secref{ss:exti} and \secref{ss:extii} (cf. Tables~\ref{t:bdown-selbra} and \ref{t:bdownrecur})
	can be made monadic following similar lines.

\smallskip
	These two optimizations can be combined by transforming the 
	trios of the monadic breakdown into duos, following the key idea of 
	the first optimization (cf. \eqref{eq:duo}).

\section{Related Work}
\label{s:rw}
Our developments are related to results by Parrow~\cite{DBLP:conf/birthday/Parrow00}, who showed that every process 
in the untyped, summation-free $\pi$-calculus with replication is weakly bisimilar to its decomposition into trios processes (i.e., $P \approx \D{P}$). 
We draw inspiration from insights developed in~\cite{DBLP:conf/birthday/Parrow00}, but pursuing different goals in a different technical setting: our decomposition treats processes from a calculus without name-passing but with higher-order concurrency (abstraction-passing), supports labeled choices, and accommodates recursive types.
Our goals are different than those in~\cite{DBLP:conf/birthday/Parrow00} because trios processes are relevant to our work in that  they allow us to formally justify minimal session types; however, they are not an end in themselves. Still, we opted to retain the definitional style and terminology for trios from~\cite{DBLP:conf/birthday/Parrow00}, which are elegant and clear.

Our main  result connects the typability of a process with that of its decomposition; this is a \emph{static guarantee}.
Based on our examples,
we conjecture that the \emph{behavioral guarantee} given by $P \approx \D{P}$ in~\cite{DBLP:conf/birthday/Parrow00} holds in our setting too, under an appropriate  \emph{typed} weak bisimilarity. 
An obstacle here is that known notions of typed bisimilarity for session-typed processes, such as those given by Kouzapas et al.~\cite{KouzapasPY17}, are not adequate: they only relate processes typed under the \emph{same} typing environments. 
We need a relaxed equivalence that (i)~relates processes typable under different environments (e.g., $\Delta$ and $\Gt{\Delta}$) and (ii)~admits that actions along $s$ from $P$ can be  matched by $\D{P}$ using actions along $s_k$, for some $k$  (and viceversa).
Defining this notion precisely and studying its properties goes beyond the scope of this paper.

Our apporach is broadly related to works that relate session types with other type systems for the $\pi$-calculus (cf.~\cite{DBLP:conf/unu/Kobayashi02,DBLP:conf/ppdp/DardhaGS12,DBLP:journals/iandc/DardhaGS17,DBLP:conf/concur/DemangeonH11,DBLP:journals/corr/GayGR14}).
Kobayashi~\cite{DBLP:conf/unu/Kobayashi02} encoded a finite session $\pi$-calculus into 
a $\pi$-calculus with linear types with usages (without sequencing); this encoding uses a continuation-passing style to codify  
a session name using multiple linear channels. 
Dardha et al.~\cite{DBLP:conf/ppdp/DardhaGS12,DBLP:journals/iandc/DardhaGS17} formalize and extend Kobayashi's approach.
They use two separate encodings, one for processes and one for types. 
The former uses a freshly generated linear name to mimic each session action; this fresh name becomes an additional argument in communications. Polyadicity is thus an essential ingredient in~\cite{DBLP:conf/ppdp/DardhaGS12,DBLP:journals/iandc/DardhaGS17}, whereas in our work it is convenient but not indispensable (cf. \secref{s:opt}). 
The encoding of types in~\cite{DBLP:conf/ppdp/DardhaGS12,DBLP:journals/iandc/DardhaGS17} codifies sequencing in session types by  nesting  payload types. In contrast, we ``slice'' the $n$ actions occurring in a session $s$ along indexed names $s_1, \ldots, s_n$ with minimal session types, i.e., slices of the type for $s$. 
All in all, an approach based on minimal session types appears simpler than that in~\cite{DBLP:conf/ppdp/DardhaGS12,DBLP:journals/iandc/DardhaGS17}.
Works by Padovani~\cite{Padovani17A} 
and
Scalas et al.~\cite{DBLP:conf/ecoop/ScalasY16} is also related: they
rely on~\cite{DBLP:conf/ppdp/DardhaGS12,DBLP:journals/iandc/DardhaGS17} to develop  verification techniques based on session types for OCaml and Scala programs, respectively. 

Gay et al.~\cite{DBLP:journals/corr/GayGR14} formalize how to encode
a  monadic $\pi$-calculus, equipped with a finite variant of the binary session types of~\cite{DBLP:journals/acta/GayH05},
into a polyadic $\pi$-calculus with an instance of the generic process types of~\cite{DBLP:journals/tcs/IgarashiK04}.
The work of Demangeon and Honda~\cite{DBLP:conf/concur/DemangeonH11}  encodes a session $\pi$-calculus into 
a linear/affine $\pi$-calculus with subtyping based on choice and selection types.
Our developments differ from these previous works in an important respect: we relate two formulations of session types, namely standard session types and minimal session types. Indeed, while~\cite{DBLP:conf/unu/Kobayashi02,DBLP:conf/ppdp/DardhaGS12,DBLP:journals/iandc/DardhaGS17,DBLP:conf/concur/DemangeonH11,DBLP:journals/corr/GayGR14} target the \emph{relative expressiveness} of session-typed process languages, our work emerges as the first study of \emph{absolute expressiveness} in this context. 

Finally, we elaborate further on our choice of \HO as source language. 
\HO is a sub-calculus of \HOp, whose basic theory and expressivity were studied by Kouzapas et al.~\cite{DBLP:conf/esop/KouzapasPY16,KouzapasPY17} as a hierarchy of session-typed calculi based on relative expressiveness. 
Our developments enable us to include \HO with minimal session types within this hierarchy.
Still, our approach does not rely on having \HO as source language, and can be adapted 
to other typed frameworks based on session types, such as the type discipline for first-order $\pi$-calculus processes in~\cite{DBLP:journals/iandc/Vasconcelos12}.

\section{Concluding Remarks}
\label{s:concl}

Session types are a class of \emph{behavioral types} for message-passing  programs. We presented a \emph{decomposition} of session-typed processes in~\cite{DBLP:conf/esop/KouzapasPY16,KouzapasPY17} using \emph{minimal} session types, in which there is no sequencing.
The decomposition of a process $P$, denoted  $\D{P}$,  is a collection of \emph{trios processes} that trigger each other 
mimicking its sequencing.
We prove that typability of $P$ using standard session types implies the typability of $\D{P}$ with minimal session types. 
Our results hold for all session types constructs, including labeled choices and recursive types.

Our contributions can be interpreted in three ways.
\emph{First}, 
from a foundational standpoint, our study of minimal session types is a conceptual contribution to the theory of behavioral types, in that we precisely identify sequencing as a source of redundancy in all preceding session types theories. As remarked in \secref{s:intro}, there are many session types variants, and their expressivity often comes at the price of an involved underlying theory. Our work contributes in the opposite direction, as we identified a simple yet expressive fragment of an already minimal session-typed framework~\cite{DBLP:conf/esop/KouzapasPY16,KouzapasPY17}, which allows us to justify session types in terms of themselves. Understanding further the underlying theory of minimal session types (e.g., notions such as type-based compatibility) is an exciting direction for future work.

\emph{Second},  
our work can be seen as a new twist on Parrow's decomposition results in the \emph{untyped} setting~\cite{DBLP:conf/birthday/Parrow00}. While Parrow's 
work indeed does not consider types, in fairness we must observe that when Parrow's work appeared (1996) the study of types for the $\pi$-calculus was rather incipient (for instance, binary session types appeared in 1998~\cite{honda.vasconcelos.kubo:language-primitives}). 
That said, we should stress that our results are not merely an extension of Parrow's with session types, for
 types in our setting drastically narrow down the range of conceivable decompositions. Also, we exploit features not supported in~\cite{DBLP:conf/birthday/Parrow00}, most notably higher-order concurrency.

\emph{Last but not least}, from a practical standpoint, we believe that our approach paves a new avenue to the integration of session types in programming languages whose type systems lack sequencing, such as Go.
It is natural to envision program analysis tools which,
given a message-passing program that should conform to protocols specified as session types,
exploit our decomposition as an intermediate step in the verification of communication correctness. Remarkably, our decomposition lends itself naturally to an implementation---in fact, we generated our examples automatically using \misty, an associated artifact written in Haskell.



\bibliography{session}

\appendix

\section{Appendix to \secref{ss:core}}
\label{app:core}

\subsection{Auxiliary Results}
\begin{remark}
	We derive polyadic rules for typing \HOp as an expected extension of
	\HO typing rules:

  \begin{prooftree}
    \AxiomC{}
    \LeftLabel{\scriptsize PolyVar}
    \UnaryInfC{$\Gamma, \widetilde x : \widetilde U_x;\widetilde y :
                  \widetilde U_y; \es \proves
                  \widetilde x \widetilde y :\widetilde U_x \widetilde U_y$}
  \end{prooftree}

  \begin{prooftree}
  	\AxiomC{}
  	\LeftLabel{\scriptsize PolySess}
  	\UnaryInfC{$\Gamma;\es;\widetilde u:\widetilde S
  				\proves \widetilde u \hastype \widetilde S$}
  \end{prooftree}

	\begin{prooftree}
		\AxiomC{$\Gamma;\Lambda_1;\Delta \cat u:S \proves P \hastype \Proc$}
		\AxiomC{$\Gamma;\Lambda_2;\es \proves \widetilde x \hastype \widetilde U$}
		\LeftLabel{\scriptsize PolyRcv}		\BinaryInfC{$\Gamma \setminus \widetilde x;\Lambda_1 \setminus \Lambda_2; \Delta \cat
		u:\btinp{\widetilde U}S
		\proves
		\binp{u}{\widetilde x}P \hastype \Proc
		$}
	\end{prooftree}

	\begin{prooftree}
		\AxiomC{$\Gamma;\Lambda_1;\Delta \proves P$}
		\AxiomC{$\Gamma;\Lambda_2;\es \proves
						\widetilde x \hastype \widetilde U$}
		\LeftLabel{\scriptsize PolySend}
		\BinaryInfC{$\Gamma;\Lambda_1 \cat \Lambda_2;
							(\Delta \setminus u:S) \cat
							u:\btout{\widetilde U}S \proves
							\bout{u}{\widetilde x}P$}
	\end{prooftree}
	\begin{prooftree}
		\AxiomC{$\leadsto \in \{\multimap,\rightarrow \}$}
		\AxiomC{$\Gamma;\Lambda;\Delta_1 \proves V \hastype \widetilde C \leadsto \diamond$}
		\AxiomC{$\Gamma;\es;\Delta_2 \proves \widetilde u \hastype \widetilde C$}
		\LeftLabel{\scriptsize PolyApp}
		\TrinaryInfC{$\Gamma;\Lambda;\Delta_1 \cat \Delta_2 \proves
						\appl{V}{\widetilde u}$}
	\end{prooftree}

	\begin{prooftree}
		\AxiomC{$\Gamma;\Lambda;\Delta_1 \proves P \hastype \Proc$}
		\AxiomC{$\Gamma;\es;\Delta_2 \proves \widetilde x \hastype \widetilde C$}
		\LeftLabel{\scriptsize PolyAbs}
		\BinaryInfC{$\Gamma \setminus \widetilde x;\Lambda;\Delta_1 \setminus \Delta_2 \proves \abs{\widetilde x}{P} \hastype \lhot{\widetilde C}$}
	\end{prooftree}

	\begin{prooftree}
		\AxiomC{$\Gamma \cat \widetilde a : \widetilde {\chtype{U}};\Lambda;\Delta \proves
					P$}
		\LeftLabel{\scriptsize PolyRes}
		\UnaryInfC{$\Gamma;\Lambda;\Delta \proves
						\news{\widetilde a}{P}$}
	\end{prooftree}

	\begin{prooftree}
		\AxiomC{$\Gamma;\Lambda;\Delta \cat \widetilde s:\widetilde S_1
					\cat \dual{\widetilde s}:\widetilde S_2 \proves
					P$}
		\AxiomC{$\widetilde S_1 \ \dualof \ \widetilde S_2$}
		\LeftLabel{\scriptsize PolyResS}
		\BinaryInfC{$\Gamma;\Lambda;\Delta \proves
						\news{\widetilde s}{P}$}
	\end{prooftree}

\end{remark}

\begin{lemma}[Substitution Lemma~\cite{DBLP:conf/esop/KouzapasPY16}]
  \label{lem:subst}
  $\Gamma;\Lambda;\Delta \cat x : S \proves P \hastype \Proc$
  and $u \notin \text{dom}(\Gamma,\Lambda,\Delta)$ implies
  $\Gamma;\Lambda;\Delta \cat u : S \proves P \subst{u}{x} \hastype \Proc$.
\end{lemma}

\begin{lemma}[Shared environment weakening]
	\label{lem:weaken}
	If $\Gamma;\Lambda;\Delta \proves P \hastype \Proc$
	then
	$\Gamma \cat x:\shot{C};\Lambda;\Delta \proves P \hastype \Proc$ and
	$\Gamma \cat u:\chtype{U};\Lambda;\Delta \proves P \hastype \Proc$.
\end{lemma}

\begin{lemma}[Shared environment strengthening]
	\label{lem:strength}
	\begin{itemize}
	\item	If $\Gamma;\Lambda;\Delta \proves P \hastype \Proc$ and
	$x \notin \fv{P}$ then $\Gamma \setminus x;\Lambda;\Delta \proves P \hastype \Proc$.
	\item If $\Gamma;\Lambda;\Delta \proves P \hastype \Proc$ and
	$u \notin \fn{P}$ then $\Gamma \setminus u;\Lambda;\Delta \proves P \hastype \Proc$.
	\end{itemize}
\end{lemma}

\subsection{Proof of \thmref{t:typecore}}
\label{app:typecore}

\begin{proof}
  By mutual induction on the structure of $P$ and $V$. The proof of the output case in Part~(1) relies on Parts~(2) and (3), whereas the proof of Parts~(2) and (3) relies on Part~(1).
  We analyze each part of the theorem separately, following the definition of $\B{k}{\tilde x}{\cdot}$ in \tabref{t:bdowncore}:

  \begin{enumerate}
  \item By assumption, $\Gamma;\Lambda;\Delta \proves P$. We consider six cases, depending on the shape of $P$:
  \begin{enumerate}
    \item Case $P=\inact$. The only rule that can be applied here is \textsc{Nil}.
    By inversion of this rule, we have:
    $\Gamma;\es;\es \proves \inact$.
    We shall then prove the following judgment:
    \begin{align}
    	\Gt{\Gamma};\es;\Theta \proves \B{k}{\tilde x}{\inact} \hastype \Proc
    \end{align}

    \noindent where $\widetilde x = \fv{\inact}=\es$ and $\Theta =
    \{\prop_k:\btinp{\lhot{\tinact}} \tinact\}$. By \tabref{t:bdowncore}:
    $\B{k}{\epsilon}{\inact} =
    \propinp{k}{} \inact$. By convention we know that
    $\propinp{k}{} \inact$ stands for $\propinp{k}{y} \inact$ with
    $\prop_k : \btinp{\shot{\tinact}}\tinact$.
    The following tree proves this case:
    \def\proofSkipAmount{\vskip 1.2ex plus.8ex minus.4ex}
    \begin{prooftree}
    \AxiomC{}
    \LeftLabel{\scriptsize Nil}
      \UnaryInfC{$\Gamma';\es;\es \proves \inact \hastype \Proc$}
      \AxiomC{$\prop_k \notin \dom{\Gamma}$}
      \LeftLabel{\scriptsize End}
      \BinaryInfC{$\Gamma';\es;\prop_k : \tinact \proves \inact \hastype \Proc$}
      \AxiomC{}
      \LeftLabel{\scriptsize LVar}
      \UnaryInfC{$\Gt{\Gamma};y \hastype \lhot{\tinact};\es \proves y \hastype \lhot{\tinact}$}
      \LeftLabel{\scriptsize EProm}
      \UnaryInfC{$\Gamma';\es;\es \proves y \hastype \lhot{\tinact}$}
      \LeftLabel{\scriptsize Prom}
      \UnaryInfC{$\Gamma';\es;\es \proves y \hastype \shot{\tinact}$}
      \LeftLabel{\scriptsize Rcv}
      \BinaryInfC{$\Gt{\Gamma};\es;\Theta \proves \propinp{k}{} \inact \hastype \Proc$}
    \end{prooftree}
    \noindent where $\Gamma' = \Gt{\Gamma} \cat y : \shot{\tinact}$. 
	We know $\prop_k \notin \dom{\Gamma}$ since we use reserved names for 	propagators channels. 
  \item Case $P = \binp{u_i}{y}P'$. We distinguish two sub-cases: \rom{1}
  $u_i \in \dom{\Delta}$ and \rom{2} $u_i \in \dom{\Gamma}$.
  We consider
  sub-case \rom{1} first. For this case Rule \textsc{Rcv} can be applied:
    \begin{align}
      \label{pt:inputInv}
      \AxiomC{$\Gamma; \Lambda_1; \Delta' \cat u_i : S \proves P' \hastype \Proc$}
      \AxiomC{$\Gamma; \Lambda_2; \es \proves y \hastype U$}
      \LeftLabel{\scriptsize Rcv}
      \BinaryInfC{$\Gamma \setminus y;\Lambda_1 \setminus \Lambda_2;
      				\Delta' \cat u_i:
                    \btinp{U}S \proves \binp{u_i}{y}P' \hastype \Proc$}
      \DisplayProof
    \end{align}

Let $\widetilde x = \fv{P}$ and
$\widetilde x' = \fv{P'}$.  Also, let
   $\Gamma_1'=\Gamma \setminus \widetilde x'$ and $\Theta_1$ be a
   balanced environment such that
   $$
   \dom{\Theta_1}=\{\prop_{k+1},\ldots,\prop_{k+\len{P'}}\}
   \cup \{\dual{\prop_{k+2}},\ldots,\dual{\prop_{k+\len{P'}}}\}
   $$
   and ${\Theta_1(\prop_{k+1})}=\btinp{\widetilde M'}\tinact$
   where $\widetilde M' = (\Gt{\Gamma},\Gt{\Lambda_1})(\widetilde x')$.
  Then, by IH on the first assumption of \eqref{pt:inputInv} we know:
  \begin{align}
    \label{eq:input-ih}
    \Gt{\Gamma'_1};\es;\Gt{\Delta' \cat u_i:S} \cat \Theta_1 \proves
    \B{k+1}{\tilde x'}{P'} \hastype \Proc
  \end{align}

 By \defref{def:typesdecomp} and \defref{def:typesdecompenv} and the second
 assumption of \eqref{pt:inputInv} we have:
  \begin{align}
  	\label{eq:input-ih2}
  	\Gt{\Gamma};\Gt{\Lambda_2};\es \proves y \hastype \Gt{U}
  \end{align}

  We define $\Theta = \Theta_1 \cat \Theta'$, where
  \begin{align*}
  \Theta' = \prop_k:\btinp{\widetilde M} \tinact \cat \dual {\prop_{k+1}}:\btout{\widetilde M'} \tinact
  \end{align*}
  with $\widetilde M = (\Gt{\Gamma},\Gt{\Lambda_1 \setminus \Lambda_2})(\widetilde x)$.
  By \defref{def:sizeproc},
  $\len{P} = \len{P'} + 1$ so
  $$\dom{\Theta} = \{\prop_k,\ldots,\prop_{k+\len{P}-1}\}
  	\cup \{\dual{\prop_{k+1}},\ldots,\dual{\prop_{k+\len{P}-1}}\}$$
  	and $\Theta$ is balanced since $\Theta(\prop_{k+1}) \dualof \Theta(\dual{\prop_{k+1}})$ and
  	 $\Theta_1$ is balanced.
  By \tabref{t:bdowncore}:
   \begin{align*}
   \B{k}{\tilde x}{\binp{u_i}{y}P'} = \propinp{k}{\widetilde
   x}\binp{u_i}{y}\propout{k+1}{\widetilde x'} \inact \Par \B{k+1}{\tilde
   x'}{P'\incrname{u}{i}}
 \end{align*}

  Let $\Gamma_1 = \Gamma \setminus \widetilde x$.
  We shall prove the following judgment:
  \begin{align*}
  \Gt{\Gamma_1 \setminus y};\es;\Gt{\Delta' \cat u_i:\btinp{U}S}
    \cat \Theta
    \proves
    	\B{k}{\tilde x}{\binp{u_i}{y}P'}
  \end{align*}


  The left-hand side component of $\B{k}{\tilde x}{\binp{u_i}{y}P'}$ is typed using
   some auxiliary derivations:
  \begin{align}
    \label{st:input1}
    \AxiomC{}
    \LeftLabel{\scriptsize Nil}
    \UnaryInfC{$\Gt{\Gamma};\es;\es \proves \inact \hastype \Proc$}
    \LeftLabel{\scriptsize End}
    \UnaryInfC{$\Gt{\Gamma};\es;\dual{\prop_{k+1}}:\tinact \proves \inact \hastype \Proc$}
    \AxiomC{}
    \LeftLabel{\scriptsize PolyVar}
    \UnaryInfC{$\Gt{\Gamma};\Gt{\Lambda_1};\es
                \proves \widetilde x' \hastype \widetilde M'$}
    \LeftLabel{\scriptsize PolySend}
    \BinaryInfC{$\Gt{\Gamma};\Gt{\Lambda_1}\cat\Gt{\Lambda_2};
                  \dual {\prop_{k+1}}:\btout{\widetilde M'} \tinact
                  \proves \propout{k+1}{\widetilde x'} \inact \hastype \Proc$}
    \LeftLabel{\scriptsize End}
    \UnaryInfC{$\Gt{\Gamma};\Gt{\Lambda_1}\cat\Gt{\Lambda_2};
                  \dual {\prop_{k+1}}:\btout{\widetilde M'} \tinact
                  \cat u_i:\tinact
                  \proves \propout{k+1}{\widetilde x'} \inact \hastype \Proc$}
    \DisplayProof
  \end{align}
  \begin{align}
    \label{st:input-rcv}
    \AxiomC{\eqref{st:input1}}
	  \AxiomC{\eqref{eq:input-ih2}}
    \LeftLabel{\scriptsize Rcv}
    \BinaryInfC{
    \begin{tabular}{c}
        $\Gt{\Gamma \setminus y};\Gt{\Lambda_1 \setminus \Lambda_2}; u_i:\btinp{\Gt{U}} \tinact \cat
                \dual {\prop_{k+1}}:\btout{\widetilde M'} \tinact \proves$ \\
                $\binp{u_i}{y}\propout{k+1}{\widetilde x'} \inact \hastype \Proc$
    \end{tabular}
}
     \LeftLabel{\scriptsize End}
    \UnaryInfC{ 
    \begin{tabular}{c}
    $\Gt{\Gamma \setminus y};\Gt{\Lambda_1 \setminus \Lambda_2}; u_i:\btinp{\Gt{U}} \tinact \cat
                \dual {\prop_{k+1}}:\btout{\widetilde M'} \tinact
                \cat \prop_k: \tinact \proves$ \\ 
                $\binp{u_i}{y}\propout{k+1}{\widetilde x'} \inact \hastype \Proc$
     \end{tabular}}
    \DisplayProof
  \end{align}

  \begin{align}
    \label{st:input3}
    \AxiomC{\eqref{st:input-rcv}}
    \AxiomC{}
    \LeftLabel{\scriptsize PolyVar}
    \UnaryInfC{$\Gt{\Gamma \setminus y};\Gt{\Lambda_1};\es
    \proves \widetilde x \hastype \widetilde M$}
    \LeftLabel{\scriptsize PolyRcv}
    \BinaryInfC{$\Gt{\Gamma_1 \setminus y};\es;u_i:\btinp{\Gt{U}} \tinact
    \cat \Theta' \proves \propinp{k}{\widetilde
    x}\binp{u_i}{y}\propout{k+1}{\widetilde x'} \inact \hastype \Proc$}
  \DisplayProof
  \end{align}

  The following tree
  proves this case:
  \begin{align}
  	\label{pt:input}
    \AxiomC{\eqref{st:input3}}
    \AxiomC{\eqref{eq:input-ih}}
     \LeftLabel{\scriptsize (\lemref{lem:subst}) with $\subst{\tilde n}{\tilde u}$}
    \UnaryInfC{
    \begin{tabular}{c}
    $\Gt{\Gamma_1 \setminus y};\es; \Gt{\Delta' \cat u_{i+1}:S} \cat \Theta_1 \proves$ \\
    $\B{k+1}{\tilde x'}{P'\incrname{u}{i}} \hastype \Proc$
    \end{tabular}}
    \LeftLabel{\scriptsize Par}
    \BinaryInfC{\begin{tabular}{c}
    			$\Gt{\Gamma_1 \backslash y};\es;\Gt{\Delta' \cat u_i:\btinp{U}S}
                  \cat \Theta \proves$ \\
                  $\propinp{k}{\widetilde x}
                  \binp{u_i}{y}\propout{k+1}{\widetilde x'} \inact \Par \B{k+1}
                  {\tilde x'}{P'\incrname{u}{i}} \hastype \Proc$
                  \end{tabular}}
    \DisplayProof
  \end{align}
  	\noindent
  	where $\widetilde n =
  	(u_{i+1},\ldots,u_{i+\len{\Gt{S}}})$ and $\widetilde u =
  	(u_i,\ldots,u_{i+\len{\Gt{S}}-1})$.
  	We may notice that
  	if $y \in \fv{P'}$ then $\Gamma'_1 = \Gamma_1 \setminus y$. On the other hand,
  	when $y \notin \fv{P'}$ then $\Gamma'_1 = \Gamma_1$ so
  	we need to apply \lemref{lem:strength} with $y$
  	 after \lemref{lem:subst} to \eqref{eq:input-ih}
  	 in \eqref{pt:input}.
  	Note that we have used the following for the right assumption of \eqref{pt:input}:
  	\begin{align*}
  		\Gt{\Delta' \cat u_i:S}\subst{\tilde n}{\tilde u} &=
  		\Gt{\Delta' \cat u_{i+1}:S} \\
  		\B{k+1}{\tilde x'}{P'}\subst{\tilde n}{\tilde u} &=
  		\B{k+1}{\tilde x'}{P'\incrname{u}{i}}
  	\end{align*}

   This concludes sub-case \rom{1}. We now consider sub-case \rom{2}, i.e.,
   $u_i \in \dom{\Gamma}$.
   Here Rule~\textsc{Acc} can be applied:
   \begin{align}
   	\label{pt:inputptr-subcase2}
   	\AxiomC{$\Gamma;\es;\es \proves u_i \hastype \chtype{U}$}
   	\AxiomC{$\Gamma;\Lambda_1;\Delta \proves P' \hastype \Proc$}
  	\AxiomC{$\Gamma;\Lambda_2;\es \proves y \hastype U$}
  	\LeftLabel{\scriptsize Acc}
   	\TrinaryInfC{$\Gamma \setminus y;\Lambda_1 \setminus \Lambda_2;\Delta
   				    \proves \binp{u_i}{y}P' \hastype \Proc$}
   	\DisplayProof
   \end{align}

   Let $\widetilde x = \fv{P}$ and $\widetilde x'=\fv{P'}$. Furthermore, let
   $\Theta_1$, $\Theta$, $\Gamma_1$, and $\Gamma_1'$ be defined as in sub-case
   \rom{1}. By IH on the second assumption of \eqref{pt:inputptr-subcase2} we have:
   \begin{align}
   \label{eq:input-ih-2}
    \Gt{\Gamma'_1};\es;\Gt{\Delta} \cat \Theta_1 \proves \B{k+1}{\tilde x'}{P'} \hastype \Proc
   \end{align}

	By \defref{def:typesdecomp} and \defref{def:typesdecompenv} and the first
   	assumption of \eqref{pt:inputptr-subcase2} we have:
  	\begin{align}
  		\label{eq:input-ih2-2}
  		\Gt{\Gamma};\es;\es \proves u_i \hastype \chtype{\Gt{U}}
  	\end{align}

	By \defref{def:typesdecomp}, \defref{def:typesdecompenv}, and the third
   	assumption of \eqref{pt:inputptr-subcase2} we have:
  	\begin{align}
  		\label{eq:input-ih3-2}
  		\Gt{\Gamma};\Gt{\Lambda_2};\es \proves y \hastype \Gt{U}
  	\end{align}

 	By \tabref{t:bdowncore}, we have:
 	\begin{align}
  		\B{k}{\tilde x}{\binp{u_i}{y}P'} = \propinp{k}{\widetilde x}
  		\binp{u_i}{y}\propout{k+1}{\widetilde x'} \inact \Par \B{k+1}{\tilde x'}
  		{P'\incrname{u}{i}}
 	\end{align}

 	We shall prove the following judgment:
 	\begin{align}
 		\Gt{\Gamma_1 \setminus y};\es;\Gt{\Delta} \cat \Theta \proves
 		\B{k}{\tilde x}{\binp{u_i}{y}P'} \hastype \Proc
 	\end{align}

	To this end, we use some auxiliary derivations:
  \begin{align}
    \label{st:input1-2}
    \AxiomC{}
    \LeftLabel{\scriptsize Nil}
    \UnaryInfC{$\Gt{\Gamma};\es;\es \proves \inact \hastype \Proc$}
    \LeftLabel{\scriptsize End}
    \UnaryInfC{$\Gt{\Gamma};\es;\dual{\prop_{k+1}}:\tinact \proves \inact \hastype \Proc$}
    \AxiomC{}
    \LeftLabel{\scriptsize PolyVar}
    \UnaryInfC{$\Gt{\Gamma};\Gt{\Lambda_1};\es
                \proves \widetilde x' \hastype \widetilde M'$}
    \LeftLabel{\scriptsize PolySend}
    \BinaryInfC{$\Gt{\Gamma};\Gt{\Lambda_1};
                  \dual {\prop_{k+1}}:\btout{\widetilde M'} \tinact
                  \proves \propout{k+1}{\widetilde x'} \inact \hastype \Proc$}
    \DisplayProof
  \end{align}

  \begin{align}
    \label{st:input2-2}
    \AxiomC{\eqref{eq:input-ih2-2}}
    \AxiomC{\eqref{st:input1-2}}
	\AxiomC{\eqref{eq:input-ih3-2}}
    \LeftLabel{\scriptsize Acc}
    \TrinaryInfC{$\Gt{\Gamma \setminus y};\Gt{\Lambda_1};
    			\dual{\prop_{k+1}}:\btout{\widetilde M'}
                 \tinact \proves \binp{u_i}{y}\propout{k+1}{\widetilde x'} \inact \hastype \Proc$}
     \LeftLabel{\scriptsize End}
    \UnaryInfC{$\Gt{\Gamma \setminus y};\Gt{\Lambda_1 \setminus \Lambda_2};
    			\dual{\prop_{k+1}}:\btout{\widetilde M'}
                 \tinact \cat \prop_k:\tinact \proves \binp{u_i}{y}\propout{k+1}{\widetilde x'} \inact \hastype \Proc$}
    \DisplayProof
  \end{align}

  \begin{align}
    \label{st:input3-2}
    \AxiomC{\eqref{st:input2-2}}
    \AxiomC{}
    \LeftLabel{\scriptsize PolyVar}
    \UnaryInfC{$\Gt{\Gamma \setminus y};\Gt{\Lambda_1 \setminus \Lambda_2};\es \proves
    			\widetilde x \hastype \widetilde M$}
    \LeftLabel{\scriptsize PolyRcv}
    \BinaryInfC{$\Gt{\Gamma_1 \setminus y};\es;\Theta' \proves
    				\propinp{k}{\widetilde x}
                  	\binp{u_i}{y}\propout{k+1}{\widetilde x'} \inact
                  	\hastype \Proc$}
    \DisplayProof
  \end{align}

  The following tree
  proves this sub-case:
  \begin{align}
  	\label{pt:input-2}
    \AxiomC{\eqref{st:input3-2}}
    \AxiomC{\eqref{eq:input-ih-2}}
    \LeftLabel{\scriptsize Par}
    \BinaryInfC{$\Gt{\Gamma_1 \setminus y};\es;\Gt{\Delta} \cat \Theta \proves
                	\propinp{k}{\widetilde x}\binp{u_i}{y}\propout{k+1}{\widetilde x'}
                	\inact \Par \B{k+1}{\tilde x'}{P'} \hastype \Proc$}
    \DisplayProof
  \end{align}
  As in sub-case \rom{1}, we may notice that if $y \in \fv{P'}$ then $\Gamma'_1 = \Gamma_1 \setminus y$.
  On the other hand, if $y \notin \fv{P'}$ then $\Gamma'_1 = \Gamma_1$ so we need to apply
  \lemref{lem:strength} with $y$ to \eqref{eq:input-ih-2} in \eqref{pt:input-2}.
  This concludes the analysis for the input case $P = \binp{u_i}{y}P'$.

  \item Case $P = \bout{u_i}{V}P'$. 
  We distinguish two
  sub-cases: \rom{1} $u_i \in \dom{\Delta}$ and \rom{2} $u_i \in \dom{\Gamma}$.
  We consider sub-case \rom{1} first. For this case Rule~\textsc{Send} can be applied:
  \begin{align}
    \label{pt:outputitr}
    \AxiomC{$\Gamma;\Lambda_1;\Delta_1 \proves P' \hastype \Proc$}
    \AxiomC{$\Gamma;\Lambda_2;\Delta_2 \proves V \hastype U$}
    \AxiomC{$u_i:S \in \Delta_1 \cat \Delta_2$}
    \LeftLabel{\scriptsize Send}
    \TrinaryInfC{$\Gamma;\Lambda_1 \cat \Lambda_2;((\Delta_1\cat\Delta_2)
    			  \setminus \{ u_i:S \})
                  \cat u_i:\btout{U}S
                  \proves \bout{u_i}{V}P' \hastype \Proc$}
    \DisplayProof
  \end{align}

	Let $\widetilde z = \fv{P'}$ and $\degree = \len{V}$. Also, let $\Gamma'_1 =
	\Gamma \setminus \widetilde z$ and $\Theta_1$ be a balanced environment such
	that
	$$
	\dom{\Theta_1}=\{\prop_{k+\degree+1},\ldots,\prop_{k+\degree+\len{P'}}\}\cup
	\{\dual{\prop_{k+\degree+2}},\ldots,\dual{\prop_{k+\degree+\len{P'}}}\}
	$$
	and $\Theta_1(\prop_{k+\degree+1})=\btinp{\widetilde M_1}\tinact$ where
	$\widetilde M_1 =
	(\Gt{\Gamma},\Gt{\Lambda_1})(\widetilde z)$.

  	Then, by IH on the first assumption of \eqref{pt:outputitr} we have:
  	\begin{align}
  		\label{eq:output-ih-1}
    	\Gt{\Gamma'_1};\es;\Gt{\Delta_1} \cat \Theta_1 \proves
    	\B{k+\degree+1}{\tilde z}{P'} \hastype \Proc
  	\end{align}

	Let $\widetilde y = \fv{V}$ and $\Gamma'_2 = \Gamma \setminus \widetilde y$.
	There are two cases:
	$U=\lhot{C}$ or $U=\shot{C}$. We consider the first case.
	If $U=\lhot{C}$ then
	let $\Theta_2$ be a balanced environment such that
	$$\dom{\Theta_2}=\{\prop_{k+1},\ldots,\prop_{k+\len{V}}\}\cup
	\{\dual{\prop_{k+1}},\ldots,\dual{\prop_{k+\len{V}}}\}$$
	$\Theta_2(\prop_{k+1})=\btinp{\widetilde M_2}\tinact$, and
	$\Theta_2(\dual{\prop_{k+1}})=\btout{\widetilde M_2}\tinact$ where $\widetilde
	M_2 = (\Gt{\Gamma},\Gt{\Lambda_2})(\widetilde y)$.
	Then, by IH (Part 2)
	on the second assumption of \eqref{pt:outputitr} we have:
  	\begin{align}
  		\label{eq:output-ih-2}
    	\Gt{\Gamma};\Gt{\Lambda_2};\Gt{\Delta_2} \cat \Theta_2 \proves
    	\V{k+1}{\tilde y}{V} \hastype \Gt{U}
  	\end{align}

  	We comment the second case for $U$.
  	We may notice that if $U = \shot{C}$ we can
  	 define $\Theta_2 = \es$ and apply the IH (Part 3) to the second assumption to reach
  	  \eqref{eq:output-ih-2} with $\Lambda_2 = \es$ and $\Delta_2 = \es$.

  	Let $\widetilde x = \fv{P}$.
  	We define $\Theta = \Theta_1 \cat \Theta_2 \cat \Theta'$, where:
  	\begin{align*}
    	\Theta' = \prop_{k}: \btinp{\widetilde M} \tinact \cat
    	\dual{\prop_{k+r+1}}:\btout{\widetilde M_2} \tinact
  	\end{align*}

  \noindent with $\widetilde M = (\Gt{\Gamma}\cat \Gt{\Lambda_1 \cat
  \Lambda_2})(\widetilde x)$. By \defref{def:sizeproc}, we know $\len{P} =
  \len{V} + \len{P'} + 1$, so
  $$
  \dom{\Theta}=\{\prop_k,\ldots,\prop_{k+\len{P}-1}\} \cup
  \{\dual{\prop_{k+1}},\ldots,\dual{\prop_{k+\len{P}-1}}\}
  $$
   By construction $\Theta$ is balanced since $\Theta(\prop_{k+\degree+1}) \dualof \Theta(\dual{\prop_{k+\degree+1}})$ and
  	 $\Theta_1$ and $\Theta_2$ are balanced.
  By \tabref{t:bdowncore}, we have:
  \begin{align*}
    \B{k}{\tilde x}{\bout{u_i}{V}P'} =
    \propinp{k}{\widetilde x}
    \bbout{u_i}{\V{k+1}{\tilde y}{V\incrname{u}{i}}} \propout{k+\degree+1}{\widetilde z}
    \inact \Par
    \B{k+\degree+1}{\tilde z}{P'\incrname{u}{i}}
  \end{align*}

\noindent We know
$\fv{P'} \subseteq \fv{P}$ and $\fv{V} \subseteq \fv{P}$ that is
$\widetilde z \subseteq \widetilde x$ and $\widetilde y \subseteq \widetilde x$.
Let $\Gamma_1 = \Gamma \setminus \widetilde x$. We shall prove the following judgment:
 \begin{align}
    \AxiomC{$\Gt{\Gamma_1};\es; \Gt{((\Delta_1\cat\Delta_2)\setminus \{ u_i:S \})
    \cat u_i:\btout{U}S} \cat \Theta \proves
    \B{k}{\tilde x}{\bout{u_i}{V}P'} \hastype \Proc$}
    \DisplayProof
  \end{align}

  Let $\sigma = \incrname{u}{i}$. To type the left-hand side component of
  $\B{k}{\tilde x}{\bout{u_i}{V}P'}$
  we use some auxiliary derivations:
  \begin{align}
  	\label{pt:zsend}
  	\AxiomC{}
  	\LeftLabel{\scriptsize Nil}
  	\UnaryInfC{$\Gt{\Gamma};\es;\es \proves \inact \hastype \Proc$}
  	\LeftLabel{\scriptsize End}
  	\UnaryInfC{$\Gt{\Gamma};\es; \dual{\prop_{k+\degree+1}} : \tinact \proves \inact \hastype \Proc$}
  	\AxiomC{}
  	\LeftLabel{\scriptsize PolyVar}
  	\UnaryInfC{$\Gt{\Gamma};\Gt{\Lambda_1};\es
  				\proves \widetilde z \hastype \widetilde M_2$}
  	\LeftLabel{\scriptsize PolySend}
  	\BinaryInfC{$\Gt{\Gamma};\Gt{\Lambda_1};
  				\dual{\prop_{k+\degree+1}}:\btout{\widetilde M_2} \tinact
  				\proves \bout{\dual {\prop_{k+\degree+1}}}{\widetilde z} \inact
  				\hastype \Proc$}
  	\LeftLabel{\scriptsize End}
  	\UnaryInfC{$\Gt{\Gamma};\Gt{\Lambda_1};
  				\dual{\prop_{k+\degree+1}}:\btout{\widetilde M_2} \tinact
  				\cat u_i:\tinact
  				\proves \bout{\dual {\prop_{k+\degree+1}}}{\widetilde z} \inact
  				\hastype \Proc$}
  	\DisplayProof
  \end{align}

  	\begin{align}
  		\label{pt:output-weaken}
  	\AxiomC{\eqref{eq:output-ih-2}}
  	\LeftLabel{\scriptsize (\lemref{lem:subst}) with $\subst{\tilde n}{\tilde u}$}
  	\UnaryInfC{$\Gt{\Gamma};\Gt{\Lambda_2};\Gt{\Delta_2\sigma} \cat
  	             \Theta_2 \proves
  	            \V{k+1}{\tilde y}{V\sigma} \hastype \Gt{U}$}
     \DisplayProof
  	\end{align}
%

  \begin{align}
  	\label{pt:uisend}
  	\AxiomC{\eqref{pt:zsend}}
  	\AxiomC{\eqref{pt:output-weaken}}
  	\LeftLabel{\scriptsize Send}
  	\BinaryInfC{\begin{tabular}{c}
  				$\Gt{\Gamma};\Gt{\Lambda_1\cat\Lambda_2};
  				\Gt{\Delta_2\sigma}
  			      	\cat u_i:\btout{\Gt{U}}\tinact \cat \Theta_2 \cat
  			       	\dual{\prop_{k+\degree+1}}:\btout{\widetilde M_2}\tinact
  			       	\proves$ \\
  			       	$\bbout{u_i}{\V{k+1}{\tilde y}{V\sigma}}
                 	\propout{k+\degree+1}{\widetilde z} \inact \hastype \Proc$
                 	\end{tabular}}
    	\LeftLabel{\scriptsize End}
  	\UnaryInfC{\begin{tabular}{c}
  				$\Gt{\Gamma};\Gt{\Lambda_1\cat\Lambda_2};
  				\Gt{\Delta_2\sigma}
  			      	\cat u_i:\btout{\Gt{U}}\tinact \cat \Theta_2 \cat
  			       	\dual{\prop_{k+\degree+1}}:\btout{\widetilde M_2}\tinact \cat \prop_k:\tinact
  			       	\proves$ \\
  			       	$\bbout{u_i}{\V{k+1}{\tilde y}{V\sigma}}
                 	\propout{k+\degree+1}{\widetilde z} \inact \hastype \Proc$
                 	\end{tabular}}
  	\DisplayProof
  \end{align}

  \begin{align}
  \label{pt:stvalue}
  \AxiomC{\eqref{pt:uisend}}
  \AxiomC{}
  \LeftLabel{\scriptsize PolyVar}
  \UnaryInfC{$\Gt{\Gamma};\Gt{\Lambda_2};\es \proves
  				\widetilde x : \widetilde M$}
  \LeftLabel{\scriptsize PolyRcv}
  \BinaryInfC{\begin{tabular}{c}
  				$\Gt{\Gamma_1};\es;\Gt{\Delta_2\sigma} \cat
                u_i:\btout{\Gt{U}} \tinact \cat \Theta_2
                \cat \Theta'
     	        \proves$ \\
     	        $\propinp{k}{\widetilde x}
  		        \bbout{u_i}{\V{k+1}{\tilde y}{V\sigma}}
                \propout{k+\degree+1}{\widetilde z} \inact
                \hastype \Proc$
                \end{tabular}}
  \DisplayProof
  \end{align}
  The following tree proves this case:
  \begin{align}
    \label{pt:output1}
    \AxiomC{\eqref{pt:stvalue}}
    \AxiomC{\eqref{eq:output-ih-1}}
    \LeftLabel{\scriptsize (\lemref{lem:subst}) with $\subst{\tilde n}{\tilde u}$}
    \UnaryInfC{$\Gt{\Gamma'_1};\es;\Gt{\Delta_1\sigma} \cat \Theta_1
                \proves \B{k+r+1}{\tilde z}{P'\sigma} \hastype \Proc$}
    \LeftLabel{\scriptsize (\lemref{lem:strength}) with $\tilde x \setminus \tilde z$}
    \UnaryInfC{$\Gt{\Gamma_1};\es;\Gt{\Delta_1\sigma} \cat
    			\Theta_1 \proves
                \B{k+r+1}{\tilde z}{P'\sigma} \hastype \Proc$}
     \LeftLabel{\scriptsize Par}
    \BinaryInfC{$\Gt{\Gamma_1};\es; \Gt{((\Delta_1\cat\Delta_2)
    			\setminus \{ u_i:S \})
                  \cat u_i:\btout{U}S} \cat \Theta \proves
                  \B{k}{\tilde x}{\bout{u_i}{V}P'} \hastype \Proc$}
    \DisplayProof
  \end{align}
  \noindent where $\widetilde n =
  (u_{i+1},\ldots,u_{i+\len{\Gt{S}}})$ and $\widetilde u =
  (u_i,\ldots,u_{i+\len{\Gt{S}}-1})$.

  We consider sub-case \rom{2}. For this sub-case Rule~\textsc{Req}
		can be applied:

   \begin{align}
   	\label{pt:outputitr-2}
   	\AxiomC{$\Gamma;\es;\es \proves u \hastype \chtype{U}$}
   	\AxiomC{$\Gamma;\Lambda;\Delta_1 \hastype P' \hastype \Proc$}
  	\AxiomC{$\Gamma;\es;\Delta_2 \proves V \hastype U$}
  	\LeftLabel{\scriptsize Req}
   	\TrinaryInfC{$\Gamma;\Lambda;\Delta_1 \cat \Delta_2 \proves \bout{u}{V}P' \hastype \Proc$}
   	\DisplayProof
   \end{align}

   Let $\widetilde z = \fv{P'}$ and $\degree = \len{V}$. Further, let $\Gamma'_1 =
   \Gamma \setminus \widetilde z$ and let $\Theta_1$ and $\Theta_2$ be
     environments defined as in sub-case \rom{1}.

  By IH on the second assumption of \eqref{pt:outputitr-2} we have:
  \begin{align}
  	\label{eq:output-ih-1-2}
    \Gt{\Gamma_1'};\es;\Gt{\Delta_1} \cat \Theta_1 \proves
    \B{k+\degree+1}{\tilde z}{P'} \hastype \Proc
  \end{align}

	Let $\widetilde y = \fv{V}$. 
  By IH on the second assumption of \eqref{pt:outputitr} we have:
  \begin{align}
  	\label{eq:output-ih-2-2}
    \Gt{\Gamma};\es;\Gt{\Delta_2} \cat \Theta_2 \proves
    \V{k+1}{\tilde y}{V} \hastype \Gt{U}
  \end{align}

  Let $\widetilde x = \fv{P}$ and $\Gamma_1 = \Gamma \setminus \widetilde x$.
  We define $\Theta = \Theta_1 \cat \Theta_2 \cat \Theta'$, where:
  \begin{align*}
    \Theta' = \prop_{k}: \btinp{\widetilde M} \tinact \cat
    \dual{\prop_{k+r+1}}:\btout{\widetilde M_2} \tinact
  \end{align*}
    \noindent with
  $\widetilde M = (\Gt{\Gamma}\cat \Gt{\Lambda})(\widetilde x)$.
  By \defref{def:sizeproc}, we know $\len{P} =
  \len{V} + \len{P'} + 1$, so
  $$\dom{\Theta}=(\prop_k,\ldots,\prop_{k+\len{P}-1}) \cup
  (\dual{\prop_{k+1}},\ldots,\dual{\prop_{k+\len{P}-1}})$$
  By construction $\Theta$ is balanced since $\Theta(\prop_{k+\degree+1}) \dualof \Theta(\dual{\prop_{k+\degree+1}})$ and
  	 $\Theta_1$ and $\Theta_2$ are balanced.
  By \tabref{t:bdowncore}, we have:
\begin{align*}
  \B{k}{\tilde x}{\bout{u_i}{V}P'} =
  \propinp{k}{\widetilde x}
  \bbout{u_i}{\V{k+1}{\tilde y}{V}} \propout{k+\degree+1}{\widetilde z} \inact \Par
  \B{k+\degree+1}{\tilde z}{P'}
\end{align*}
\noindent We know
$\fv{P'} \subseteq \fv{P}$ and $\fv{V} \subseteq \fv{P}$ that is
$\widetilde z \subseteq \widetilde x$ and $\widetilde y \subseteq \widetilde x$. 

  To prove that
  $\Gt{\Gamma_1};\es; \Gt{\Delta_1\cat\Delta_2} \cat \Theta
                  \proves \B{k}{\tilde x}{\bout{u_i}{V}P'}$,
  we use some auxiliary derivations:
  \begin{align}
  	\label{pt:zsend-2}
  	\AxiomC{}
  	\LeftLabel{\scriptsize Nil}
  	\UnaryInfC{$\Gt{\Gamma};\es;\es \proves \inact \hastype \Proc$}
  	\LeftLabel{\scriptsize End}
  	\UnaryInfC{$\Gt{\Gamma};\es;\dual{\prop_{k+\degree+1}}:\tinact \proves \inact \hastype \Proc$}
  	\AxiomC{}
  	\LeftLabel{\scriptsize PolyVar}
  	\UnaryInfC{$\Gt{\Gamma};\Gt{\Lambda};\es
  				\proves \widetilde z : \widetilde M_2$}
  	\LeftLabel{\scriptsize PolySend}
  	\BinaryInfC{$\Gt{\Gamma};\Gt{\Lambda};
  				\dual{\prop_{k+\degree+1}}:\btout{\widetilde M_2} \tinact
  				\proves \bout{\dual {\prop_{k+\degree+1}}}{\widetilde z} \inact \hastype \Proc$}
  	\DisplayProof
  \end{align}

  \begin{align}
  	\label{pt:uisend-2}
  	\AxiomC{\eqref{pt:zsend-2}}
  	\AxiomC{\eqref{eq:output-ih-2-2}}
  	\LeftLabel{\scriptsize Req}
  	\BinaryInfC{$\Gt{\Gamma};\Gt{\Lambda};\Gt{\Delta_2}
  			\cat \Theta_2 \cat
  			\dual{\prop_{k+\degree+1}}:\btout{\widetilde M_2} \tinact
  			\proves \bbout{u_i}{\V{k+1}{\tilde y}{V}}
  			\propout{k+\degree+1}{\widetilde z} \inact \hastype \Proc$}
  	\DisplayProof
  \end{align}

  \begin{align}
  \label{pt:stvalue-2}
  \AxiomC{\eqref{pt:uisend-2}}
  \AxiomC{}
  \LeftLabel{\scriptsize PolyVar}
  \UnaryInfC{$\Gt{\Gamma};\Gt{\Lambda};\es \proves
  			\widetilde x \hastype \widetilde M$}
  \LeftLabel{\scriptsize PolyRcv}
  \BinaryInfC{$\Gt{\Gamma_1};\es;\Gt{\Delta_2} \cat \Theta_2
  		\cat \Theta'
     	\proves \propinp{k}{\widetilde x}
  		\bbout{u_i}{\V{k+1}{\tilde y}{V}} \propout{k+\degree+1}{\widetilde z} \inact \hastype \Proc$}
  \DisplayProof
  \end{align}

  The following tree proves this case:
  \begin{align}
    \label{pt:output1-2}
    \AxiomC{\eqref{pt:stvalue-2}}
    \AxiomC{\eqref{eq:output-ih-1-2}}
    \LeftLabel{\scriptsize (\lemref{lem:strength}) with $\tilde x \setminus \tilde z$}
    \UnaryInfC{$\Gt{\Gamma_1};\es;\Gt{\Delta_1} \cat \Theta_1 \proves
    \B{k+\degree+1}{\tilde z}{P'} \hastype \Proc$}
     \LeftLabel{\scriptsize Par}
    \BinaryInfC{$\Gt{\Gamma_1};\es; \Gt{\Delta_1\cat\Delta_2} \cat \Theta
                  \proves \B{k}{\tilde x}{\bout{u_i}{V}P'} \hastype \Proc$}
    \DisplayProof
  \end{align}
 This concludes the analysis for the output case $P = \bout{u_i}{V}P'$. 
 We remark that the proof for the case when $V = y$ is specialization 
 of above the proof where $\degree = \len{y} = 0$, $\tilde y = \fv{y} = y$, 
 $\V{k+1}{\tilde y}{y} = y$ and 
 it holds that $y \sigma = y$.

  \item Case $P = \appl{V}{u_i}$. For this case only
  Rule~\textsc{App} can be applied:
  \begin{align}
  \label{pt:applitr}
  	\AxiomC{$\leadsto \in \{\multimap,\rightarrow \}$}
  	\AxiomC{$\Gamma;\Lambda;\Delta_1 \proves V \hastype C \leadsto \diamond$}
  	\AxiomC{$\Gamma;\es;\Delta_2 \proves u_i \hastype C$}
  	\LeftLabel{\scriptsize App}
  	\TrinaryInfC{$\Gamma;\Lambda;\Delta_1 \cat \Delta_2 \proves
                  \appl{V}{u_i} \hastype \Proc$}
  	\DisplayProof
  \end{align}

We distinguish two sub-cases based on the higher-order type: \rom{1} $\leadsto = \multimap$
or \rom{2} $\leadsto = \rightarrow$. We consider sub-case \rom{1} first. Let $\widetilde x =
\fv{V}$ and let $\Theta_1$ be a balanced environment such that
$$
\dom{\Theta_1}=\{\prop_{k+1},\ldots,\prop_{k+\len{V}}\}
\cup \{\dual{\prop_{k+1}},\ldots,\dual{\prop_{k+\len{V}}}\}
$$
and $\Theta_1(\prop_{k+1}) = \btinp{\widetilde M} \tinact$
and
$\Theta_1(\dual{\prop_{k+1}}) = \btout{\widetilde M} \tinact$ where $\widetilde
M = (\Gt{\Gamma}\cat\Gt{\Lambda})(\widetilde x)$. Then, by IH (Part 2) on the second assumption
of \eqref{pt:applitr} we have:
\begin{align}
	\label{eq:appl-ih-1}
	\Gt{\Gamma};\Gt{\Lambda};\Gt{\Delta_1} \cat \Theta_1 \proves
  	\V{k+1}{\tilde x}{V} \hastype \Gt{U}
\end{align}

By \defref{def:typesdecomp} and \defref{def:typesdecompenv} and
\eqref{pt:applitr} we have:
\begin{align}
	\label{eq:appl-ih-2}
	\Gt{\Gamma};\es;\Gt{\Delta_2} \proves \widetilde m \hastype \Gt{C}
\end{align}
\noindent where $\widetilde m = (u_i,\ldots,u_{i+\len{\Gt{C}}-1})$.
We define $\Theta = \Theta_1 \cat \prop_k : \btinp{\widetilde M} \tinact$.
By \defref{def:sizeproc}, $\len{P} = \len{V}+1$ so
$\dom{\Theta}=\{\prop_k,\ldots,\prop_{k+\len{P}-1}\}
\cup \{\dual{\prop_{k+1}},\ldots,\dual{\prop_{k+\len{P}-1}}\}$.
 By construction $\Theta$ is balanced since $\Theta_1$ is balanced.

By \tabref{t:bdowncore}, we have:
\begin{align}
  \B{k}{\tilde x}{\appl{V}{u_i}} = \propinp{k}{\widetilde x} \appl{\V{k+1}{\tilde x}
  {V}}{\widetilde m}
\end{align}

We shall prove the following judgment:
\begin{align}
  \Gt{\Gamma_1};\es;\Gt{\Delta_1 \cat \Delta_2} \cat \Theta \proves
  \propinp{k}{\widetilde x} \appl{\V{k+1}{\tilde x}{V}}{\widetilde m} \hastype
  \Proc
\end{align}

We use auxiliary derivation: 

\begin{align}
	\label{pt:var-app}
	\AxiomC{\eqref{eq:appl-ih-1}}
	\AxiomC{\eqref{eq:appl-ih-2}}
	\LeftLabel{\scriptsize PolyApp}
	\BinaryInfC{\begin{tabular}{c}
	$\Gt{\Gamma};\Gt{\Lambda};\Gt{\Delta_1 \cat \Delta_2} \cat \Theta
				\proves$  $\appl{\V{k+1}{\tilde x}{V}}{\widetilde m} \hastype \Proc$
				\end{tabular}}
	\DisplayProof
\end{align}

The following tree proves this case:

\begin{align}
	\AxiomC{\eqref{pt:var-app}}
	\AxiomC{}
	\LeftLabel{\scriptsize PolyVar}
	\UnaryInfC{$\Gt{\Gamma};\Gt{\Lambda};\es
				\proves \widetilde x \hastype \widetilde M$}
	\LeftLabel{\scriptsize End}
	\UnaryInfC{$\Gt{\Gamma};\Gt{\Lambda};\prop_k:\tinact
				\proves \widetilde x \hastype \widetilde M$}
	\LeftLabel{\scriptsize PolyRcv}
	\BinaryInfC{$\Gt{\Gamma_1};\es;\Gt{\Delta_1 \cat \Delta_2} \cat \Theta
              \proves \binp{\prop_k}{\widetilde x}
              \appl{\V{k+1}{\tilde x}{V}}{\widetilde m} \hastype \Proc$}
	\DisplayProof
\end{align}
\noindent In applying \eqref{eq:appl-ih-1} and \eqref{eq:appl-ih-2} we may notice
that by \defref{def:typesdecomp}: $\Gt{U} = \lhot{\Gt{C}}$.
This concludes sub-case \rom{1}.

We now consider sub-case \rom{2}. For this case we know that $\Lambda = \es$ and
$\Delta_1 = \es$ in \eqref{pt:applitr}. Let $\widetilde x = \fv{V}$.
Then, by  IH (Part 3) on the second assumption of \eqref{pt:applitr} we have:
\begin{align}
	\label{eq:appl-ih-1-2}
	\Gt{\Gamma};\es;\es \proves \V{k+1}{\tilde x}{V} \hastype \Gt{U}
\end{align}

By \defref{def:typesdecomp}, \defref{def:typesdecompenv}, and
\eqref{pt:applitr} we have:
\begin{align}
	\label{eq:appl-ih-2-2}
	\Gt{\Gamma};\es;\Gt{\Delta_2} \proves \widetilde m \hastype \Gt{C}
\end{align}
\noindent where $\widetilde m = (u_i,\ldots,u_{i+\len{\Gt{C}}-1})$. By \tabref{t:bdowncore}, we have:
\begin{align}
  \B{k}{\tilde x}{\appl{V}{u_i}} = \propinp{k}{\widetilde x}
  \appl{\V{k+1}{\tilde x}{V}}{\widetilde m}
\end{align}

We define $\Theta = \prop_k : \btinp{\widetilde M} \tinact$
with $\widetilde M = \Gt{\Gamma}(\widetilde x)$.

We shall prove the following judgment:
\begin{align}
  \Gt{\Gamma_1};\es;\Gt{\Delta_2} \cat \Theta \proves \propinp{k}{\widetilde x}
  \appl{\V{k+1}{\tilde x}{V}}{\widetilde m} \hastype \Proc
\end{align}

The following tree proves this case:

\begin{align}
	\AxiomC{\eqref{eq:appl-ih-1-2}}
	\AxiomC{\eqref{eq:appl-ih-2}}
	\LeftLabel{\scriptsize PolyApp}
	\BinaryInfC{$\Gt{\Gamma};\Gt{\Lambda};\Gt{\Delta_2} \cat \Theta
				 \proves \appl{\V{k+1}{\tilde x}{V}}{\widetilde m}
				 \hastype \Proc$}
	\AxiomC{}
	\LeftLabel{\scriptsize PolyVar}
	\UnaryInfC{$\Gt{\Gamma};\es;\es
				      \proves \widetilde x \hastype \widetilde M$}
\LeftLabel{\scriptsize End}
	\UnaryInfC{$\Gt{\Gamma};\es;\prop_k:\tinact
				      \proves \widetilde x \hastype \widetilde M$}
	\LeftLabel{\scriptsize PolyRcv}
	\BinaryInfC{$\Gt{\Gamma_1};\es;\Gt{\Delta_2} \cat \Theta \proves
                \propinp{k}{\widetilde x} \appl{\V{k+1}{\tilde x}{V}}
                {\widetilde m} \hastype \Proc$}
	\DisplayProof
\end{align}

\noindent In applying \eqref{eq:appl-ih-1-2} and \eqref{eq:appl-ih-2} we may
notice that by \defref{def:typesdecomp}: $\Gt{U} = \shot{\Gt{C}}$.
 We remark that the proof for the case when $V = y$ is specialization 
 of above the proof where $r = \len{y} = 0$, $\widetilde x = \fv{y} = y$, 
 $\V{k+1}{\tilde x}{y} = y$ and 
 it holds that $y \sigma = y$. 

\item Case $P = \news{s:C}P'$. We distinguish two sub-cases:
\rom{1} $C = S$ and \rom{2} $C = \chtype{U}$. Firstly, we $\alpha$-convert $P$ as follows:
\begin{align*}
	P \equiv_{\alpha} \news{s_1:C}P'\subst{s_1 \dual{s_1}}{s \dual{s}}
\end{align*}

We consider sub-case \rom{1} first. For this case Rule~\textsc{ResS} can be applied:
\begin{align}
	\label{pt:restritr}
	\AxiomC{$\Gamma;\Lambda;\Delta \cat s_1:S \cat \dual {s_1}:\dual S
				\proves P'\subst{s_1 \dual{s_1}}{s \dual{s}} \hastype \Proc$}
	\LeftLabel{\scriptsize ResS}
	\UnaryInfC{$\Gamma;\Lambda;\Delta \proves
              \news{s_1:S}P'\subst{s_1 \dual{s_1}}{s \dual{s}} \hastype \Proc$}
	\DisplayProof
\end{align}
Let $\widetilde x = \fv{P'}$ and let $\Theta_1$ be
a balanced environment such that
$$
\dom{\Theta_1}=\{\prop_k,\ldots,\prop_{k+\len{P'}-1}\}
\cup \{\dual{\prop_{k+1}},\ldots,\dual{\prop_{k+\len{P'}-1}}\}
$$
and $\Theta_1(\prop_k) =
\btinp{\widetilde M} \tinact$ with $\widetilde M =
(\Gt{\Gamma}\cat\Gt{\Lambda})(\widetilde x)$.
Then, by IH on the assumption of \eqref{pt:restritr} we have:
\begin{align}
	\label{eq:restr-ih}
	\Gt{\Gamma};\es;\Gt{\Delta \cat s_1:S \cat \dual {s_1}:\dual S}
				\cat \Theta_1 \proves
				\B{k}{\tilde x}{P'\subst{s_1 \dual{s_1}}{s \dual{s}}} \hastype \Proc
\end{align}

Note that we take $\Theta = \Theta_1$ since $\fv{P}=\fv{P'}$ and $\len{P}=\len{P'}$. By \defref{def:typesdecomp} and \defref{def:typesdecompenv} and~\eqref{eq:restr-ih}, we know that:
\begin{align}
	\label{eq:restr-ih1}
	\Gt{\Gamma};\es;\Gt{\Delta} \cat \widetilde s:\Gt{S} \cat	\dual{\widetilde s}:
  \Gt{\dual S} \proves \B{k}{\tilde x}{P'\subst{s_1 \dual{s_1}}{s \dual{s}}}
	\hastype \Proc
\end{align}
\noindent where $\widetilde s = (s_1,\ldots,s_{\len{\Gt{S}}})$ and
$\dual{\widetilde s} = (\dual{s_1},\ldots,\dual{s_{\len{\Gt{S}}}})$.
By \tabref{t:bdowncore}, we have:
\begin{align*}
	\B{k}{\tilde x}{\news{s}P'} = \news{\widetilde s : \Gt{S}}
								\B{k}{\tilde x}{P'\subst{s_1 \dual{s_1}}{s \dual{s}}}
\end{align*}

The following tree proves this sub-case:
\begin{align}
	\label{pt:restr1}
	\AxiomC{\eqref{eq:restr-ih1}}
	\LeftLabel{\scriptsize PolyResS}
	\UnaryInfC{$\Gt{\Gamma};\es;\Gt{\Delta} \proves
				\news{\widetilde s : \Gt{S}}
				\B{k}{\tilde x}{P'\subst{s_1 \dual{s_1}}{s \dual{s}}} \hastype \Proc$}
	\DisplayProof
\end{align}

We now consider sub-case \rom{2}. Similarly to sub-case \rom{1} we first $\alpha$-convert $P$ as
follows:
\begin{align*}
	P \equiv_\alpha \news{s_1}{P'\subst{s_1}{s}}
\end{align*}

For this sub-case Rule~\textsc{Res} can be applied:
\begin{align}
	\label{eq:restritr2}
	\AxiomC{$\Gamma \cat s_1:\chtype{U};\Lambda;\Delta \proves
            P'\subst{s_1}{s} \hastype \Proc$}
	\LeftLabel{\scriptsize Res}
	\UnaryInfC{$\Gamma;\Lambda;\Delta \proves \news{s_1}P'\subst{s_1}{s}
					\hastype \Proc$}
	\DisplayProof
\end{align}
Let $\widetilde x = \fv{P'}$ and let $\Theta_1$ be defined as in sub-case
\rom{1}.

By IH on the
first assumption of \eqref{eq:restritr2} we have:
\begin{align}
	\label{eq:restr-ih2}
	\Gt{\Gamma \cat s_1:\chtype{U}};\es;\Gt{\Delta}
		\cat \Theta_1
		\proves \B{k}{\tilde x}{P'\subst{s_1}{s}} \hastype \Proc
\end{align}

Here we also take $\Theta = \Theta_1$ since $\fv{P} = \fv{P'}$ and $\len{P}=\len{P'}$.
We notice that by \defref{def:typesdecomp} and \defref{def:typesdecompenv} and
\eqref{eq:restr-ih2}:
\begin{align}
	\label{eq:restr-ih3}
	\Gt{\Gamma} \cat s_1:\Gt{\chtype{U}};\es;\Gt{\Delta}
		\cat \Theta_1
		\proves \B{k}{\tilde x}{P'\subst{s_1}{s}} \hastype \Proc
\end{align}

By \tabref{t:bdowncore}, we have:
\begin{align*}
	\B{k}{\tilde x}{\news{s}P'} = \news{s_1 : \Gt{\chtype{U}}}
								\B{k}{\tilde x}{P'\subst{s_1}{s}}
\end{align*}
\noindent where $\widetilde s = s_1$ since $\len{\Gt{\chtype{U}}}=1$.
The following tree proves this sub-case:
\begin{align}
	\AxiomC{\eqref{eq:restr-ih3}}
	\LeftLabel{\scriptsize Res}
	\UnaryInfC{$\Gt{\Gamma};\es;\Gt{\Delta}\cat \Theta
					\proves \news{s_1 : \Gt{\chtype{U}}}
								\B{k}{\tilde x}{P'\subst{s_1}{s}} \hastype \Proc$}
	\DisplayProof
\end{align}

\item Case $P = Q \Par R$. For this case only Rule~\textsc{Par} can be applied:
\begin{align}
	\label{pt:compitr}
	\AxiomC{$\Gamma;\Lambda_1;\Delta_1 \proves Q \hastype \Proc$}
	\AxiomC{$\Gamma;\Lambda_2;\Delta_2 \proves R \hastype \Proc$}
	\LeftLabel{\scriptsize Par}
	\BinaryInfC{$\Gamma;\Lambda_1 \cat \Lambda_2;\Delta_1 \cat \Delta_2
					     \proves Q \Par R \hastype \Proc$}
	\DisplayProof
\end{align}
Let $\widetilde y = \fv{Q}$ and let $\Theta_1$ be a balanced
environment such that
$$\dom{\Theta_1}=\{\prop_{k+1},\ldots,\prop_{k+\len{Q}}\}
\cup \{\dual{\prop_{k+2}},\ldots,\dual{\prop_{k+\len{Q}}}\}
$$
and
$\Theta_1(\prop_{k+1}) = \btinp{\widetilde M_1} \tinact$ with $\widetilde M_1 =
(\Gt{\Gamma}\cat\Gt{\Lambda_1})(\widetilde y)$.
Let $\Gamma'_1 = \Gamma \setminus \widetilde y$.
Then, by IH on the first assumption of \eqref{pt:compitr} we have:
\begin{align}
	\label{eq:comp-ih-1}
	\Gt{\Gamma'_1};\es;\Gt{\Delta_1} \cat \Theta_1 \proves
  \B{k+1}{\tilde y}{Q} \hastype \Proc
\end{align}

Let $\widetilde z = \fv{R}$ and $\degree = |Q|$. Also, let $\Theta_2$ be a balanced
environment such that
$$\dom{\Theta_2}=\{\prop_{k+\degree+1},\ldots,\prop_{k+\degree+\len{R}}\}
\cup \{\dual{\prop_{k+\degree+2}},\ldots,\dual{\prop_{k+\degree+\len{R}}}\}$$ and
$\Theta_2(\prop_{k+\degree+1}) = \btinp{\widetilde M_2} \tinact$ with $\widetilde M_2 =
(\Gt{\Gamma}\cat\Gt{\Lambda_2})(\widetilde z)$.

Let $\Gamma'_2 = \Gamma \setminus \widetilde z$.
Then, by IH on the second assumption of \eqref{pt:compitr} we have:
\begin{align}
	\label{eq:comp-ih-2}
	\Gt{\Gamma'_2};\es;\Gt{\Delta_2} \cat \Theta_2 \proves
  \B{k+\degree+1}{\tilde z}{R} \hastype \Proc
\end{align}

Let $\widetilde x = \fv{P}$ and $\widetilde M = (\Gt{\Gamma},\Gt{\Lambda_1\cat\Lambda_2})(\widetilde x)$.
We may notice that $\widetilde M_i \subseteq \widetilde M$ for $i \in \{1,2\}$.
We define $\Theta = \Theta_1 \cat \Theta_2 \cat \Theta'$ where:

 \begin{align*}
	\Theta' = \prop_k:\btinp{\widetilde M} \tinact \cat
				\dual{\prop_{k+1}}:\btout{\widetilde M_1} \tinact
				\cat
				\dual{\prop_{k+\degree+1}}:\btout{\widetilde M_2} \tinact
\end{align*}

 By construction $\Theta$ is balanced since $\Theta(\prop_{k+1}) \dualof \Theta(\dual{\prop_{k+1}})$, $\Theta(\prop_{k+\degree+1}) \dualof \Theta(\dual{\prop_{k+\degree+1}})$, and
  	 $\Theta_1$ and $\Theta_2$ are balanced.

By \tabref{t:bdowncore} we have:
\begin{align*}
  	\B{k}{\tilde x}{Q \Par R} =\propinp{k}{\widetilde x} \propout{k+1}{\widetilde y}
    \propout{k+\degree+1}{\widetilde z} \inact \Par
    \B{k+1}{\tilde y}{Q} \Par \B{k+\degree+1}{\tilde z}{R}
\end{align*}

Let $\Gamma_1 = \Gamma \setminus \widetilde x$. We shall prove the following judgment:
\begin{align*}
\Gt{\Gamma_1};\es;\Gt{\Lambda_1 \cat \Lambda_2}
                \cat \Theta \proves
                \propinp{k}{\widetilde x} \bout{\dual {\prop_{k+1}}}{\widetilde y}
   						  \bout{\dual {\prop_{k+\degree+1}}}{\widetilde z} \inact \Par
    					  \B{k+1}{\tilde y}{Q} \Par \B{k+\degree+1}{\tilde z}{R} \hastype \Proc
\end{align*}

To type the trio, which is the left-hand side component, we use some auxiliary derivations:
\begin{align}
	\label{pt:compsend2}
	\AxiomC{}
\LeftLabel{\scriptsize Nil}
	\UnaryInfC{$\Gt{\Gamma};\es;\es \proves \inact$}
	\LeftLabel{\scriptsize End}
	\UnaryInfC{$\Gt{\Gamma};\es;\dual{\prop_{k+\degree+1}}:\tinact \proves \inact$}
	\AxiomC{}
	\LeftLabel{\scriptsize PolyVar}
	\UnaryInfC{$\Gt{\Gamma};\Gt{\Lambda_2};\es
				\proves \widetilde z \hastype \widetilde M_2$}
	\LeftLabel{\scriptsize PolySend}
	\BinaryInfC{$\Gt{\Gamma};\Gt{\Lambda_2};
				\dual{\prop_{k+\degree+1}}:\btout{\widetilde M_2}\tinact
				\proves \propout{k+\degree+1}{\widetilde z} \inact \hastype \Proc$}
	\LeftLabel{\scriptsize End}
	\UnaryInfC{$\Gt{\Gamma};\Gt{\Lambda_2};
				\dual{\prop_{k+\degree+1}}:\btout{\widetilde M_2}\tinact
				\cat \dual{\prop_{k+1}}:\tinact
				\proves \propout{k+\degree+1}{\widetilde z} \inact \hastype \Proc$}
	\DisplayProof
\end{align}

\begin{align}
	\label{pt:compsend1}
	\AxiomC{\eqref{pt:compsend2}}
	\AxiomC{}
	\LeftLabel{\scriptsize PolyVar}
	\UnaryInfC{$\Gt{\Gamma};\Gt{\Lambda_1};\es
				\proves \widetilde y \hastype \widetilde M_1$}
	\LeftLabel{\scriptsize PolySend}
	\BinaryInfC{\begin{tabular}{c}
	$\Gt{\Gamma};\Gt{\Lambda_1\cat\Lambda_2};
                \dual{\prop_{k+1}}:\btout{\widetilde M_1}\tinact \cat
				 \dual{\prop_{k+\degree+1}}:\btout{\widetilde M_2}\tinact
			     \proves$ \\
			     $\propout{k+1}{\widetilde y}
   			    \propout{k+\degree+1}{\widetilde z} \inact \hastype \Proc$
   			    \end{tabular}}
   	\LeftLabel{\scriptsize End}
   \UnaryInfC{\begin{tabular}{c}
	$\Gt{\Gamma};\Gt{\Lambda_1\cat\Lambda_2};
                \dual{\prop_{k+1}}:\btout{\widetilde M_1}\tinact \cat
				 \dual{\prop_{k+\degree+1}}:\btout{\widetilde M_2}\tinact
				 \cat \prop_k:\tinact
			     \proves$ \\
			     $\propout{k+1}{\widetilde y}
   			    \propout{k+\degree+1}{\widetilde z} \inact \hastype \Proc$
   			    \end{tabular}}
	\DisplayProof
\end{align}

\begin{align}
	\label{pt:compinput1}
	\AxiomC{\eqref{pt:compsend1}}
	\AxiomC{}
	\LeftLabel{\scriptsize PolyVar}
	\UnaryInfC{$\Gt{\Gamma};\Gt{\Lambda_1\cat\Lambda_2};\es
				      \proves \widetilde x \hastype \widetilde M$}
	\LeftLabel{\scriptsize PolyRcv}
	\BinaryInfC{$\Gt{\Gamma_1};\emptyset;\Theta' \proves
					     \propinp{k}{\widetilde x}
					     \propout{k+1}{\widetilde y}
   						 \propout{k+\degree+1}{\widetilde z} \inact \hastype \Proc$}
	\DisplayProof
\end{align}

\begin{align}
	\label{pt:comp-par-1}
	\AxiomC{\eqref{eq:comp-ih-1}}
	\LeftLabel{\scriptsize (\lemref{lem:strength}) with $\tilde x \setminus \tilde y$}
	\UnaryInfC{$\Gt{\Gamma_1};\es;\Gt{\Delta_1} \cat \Theta_1 \proves
  					\B{k+1}{\tilde y}{Q} \hastype \Proc$}
  	\DisplayProof
\end{align}

\begin{align}
	\label{pt:comp-par-2}
	\AxiomC{\eqref{eq:comp-ih-2}}
	\LeftLabel{\scriptsize (\lemref{lem:strength}) with $\tilde x \setminus \tilde z$}
	\UnaryInfC{$\Gt{\Gamma_1};\es;\Gt{\Delta_2} \cat \Theta_2 \proves
  					\B{k+\degree+1}{\tilde z}{R} \hastype \Proc$}
  	\DisplayProof
\end{align}

\begin{align}
	\label{pt:par}
	\AxiomC{\eqref{pt:comp-par-1}}
	\AxiomC{\eqref{pt:comp-par-2}}
	\LeftLabel{\scriptsize Par}
	\BinaryInfC{$\Gt{\Gamma_1};\es;\Gt{\Delta_1 \cat \Delta_2} \cat
                \Theta_1 \cat \Theta_2 \proves
                \B{k+1}{\tilde y}{Q} \Par \B{k+\degree+1}{\tilde z}{R} \hastype \Proc$}
     \DisplayProof
\end{align}

The following tree proves this case:

\begin{align*}
	\AxiomC{\eqref{pt:compinput1}}
	\AxiomC{\eqref{pt:par}}
	\LeftLabel{\scriptsize Par}
	\BinaryInfC{
	\begin{tabular}{c}
	$\Gt{\Gamma_1};\es;\Gt{\Delta_1 \cat \Delta_2}
                \cat \Theta \proves$ \\
                $\propinp{k}{\widetilde x}
                \bout{\dual {\prop_{k+1}}}{\widetilde y}
   						  \bout{\dual {\prop_{k+\degree+1}}}{\widetilde z} \inact \Par
    					  \B{k+1}{\tilde y}{Q} \Par \B{k+\degree+1}{\tilde z}{R} \hastype \Proc$
    					  \end{tabular}
    					  }
	\DisplayProof
\end{align*}
\end{enumerate}

\item This is the first case concerning values when $V = y$. 
        By assumption $\Gamma;\Lambda;\Delta \proves y \hastype \lhot{C}$. 
        For this case only Rule~\textsc{LVar} can be applied and by inversion  
        $\Delta = \{y \hastype \lhot{C}\}$ and $\Delta = \es$. 
        By \tabref{t:bdowncore} we have $\V{k}{\tilde x}{y} =y$ and by \defref{def:typesdecomp} and \defref{def:typesdecompenv} we have $\Gt{\Delta} = \{ \Gt{\lhot{C}}\}$. Hence, we prove the following judgment by applying Rule~\textsc{LVar}:
        \begin{align*} 
         \Gt{\Gamma};\Gt{\Delta};\es \proves \V{k}{\tilde x}{y} \hastype 
         \Gt{\lhot{C}}
         \LeftLabel{\scriptsize LVar}
        \end{align*}

\item
This is the second case concerning values when 
$V = \abs{y}P$. By assumption we have  
$\Gamma;\Lambda;\Delta \proves V \hastype \slhot{C}$. Here 
we distinguish two sub-cases (1) $\leadsto = \multimap$ and (2) 
$\leadsto = \rightarrow$: 
\begin{itemize} 
    \item $\leadsto = \multimap$. 
By assumption, $\Gamma;\Lambda;\Delta \proves V \hastype \lhot{C}$.
In this case Rule~\textsc{Abs}
can be applied. Firstly, we $\alpha$-convert value $V$ as follows:
\begin{align}
	V \equiv_{\alpha} \abs{y_1}P\subst{y_1}{y}
\end{align}
	For this case only Rule~\textsc{Abs} can be applied:
	\begin{align}
		\label{pt:absitr}
		\AxiomC{$\Gamma;\Lambda;\Delta_1 \proves P\subst{y_1}{y}
		         \hastype \Proc$}
		\AxiomC{$\Gamma;\es;\Delta_2 \proves y_1 \hastype C$}
		\LeftLabel{\scriptsize Abs}
		\BinaryInfC{$\Gamma \setminus y_1;\Lambda;\Delta_1 \setminus \Delta_2 \proves
					       \abs{y_1}{P\subst{y_1}{y}} \hastype \lhot{C}$}
		\DisplayProof
	\end{align}
Let $\widetilde x = \fv{P}$ and $\Gamma_1 = \Gamma \setminus \widetilde x$. Also, let $\Theta_1$ be a balanced environment such that
$$\dom{\Theta_1}=\{\prop_k,\ldots,\prop_{k+\len{P}-1}\} \cup
\{\dual{\prop_{k+1}},\ldots,\dual{\prop_{k+\len{P}-1}}\}$$
and $\Theta_1(\prop_k) = \widetilde M$ with
$\widetilde M = (\Gt{\Gamma \setminus y_1}, \Gt{\Lambda})(\widetilde x)$.
 Then, by IH (Part 1) on the first assumption of \eqref{pt:absitr} we have:
 \begin{align}
 		\label{eq:abs-ih-1}
 		\Gt{\Gamma_1};\es;\Gt{\Delta_1} \cat \Theta_1
 		\proves \B{k}{\tilde x}{P\subst{y_1}{y}} \hastype \Proc
 \end{align}


By \defref{def:typesdecomp} and \defref{def:typesdecompenv} and the
second assumption of
\eqref{pt:absitr} we have:
\begin{align}
	\label{eq:abs-ih-2}
	\Gt{\Gamma};\es;\Gt{\Delta_2}\proves \widetilde y \hastype \Gt{C}
\end{align}
\noindent where $\widetilde y = (y_1,\ldots,y_{\len{\Gt{C}}})$.

We define $\Theta = \Theta_1 \cat \dual {\prop_k}:
	\btout{\widetilde M} \tinact$. By \tabref{t:bdowncore}, we have:
\begin{align*}
	\V{k}{\tilde x}{\abs{y}{P}} = \abs{\widetilde{y}:\Gt{C}}{\propout{k}{\widetilde x} \inact
  	\Par \B{k}{\tilde x}{P \subst{y_1}{y}}}
\end{align*}

 By construction $\Theta$ is balanced since $\Theta(\prop_{k}) \dualof \Theta(\dual{\prop_{k}})$ and $\Theta_1$ is balanced.
 We use an auxiliary derivation:
 \begin{align}
 	\label{pt:abs-send}
 	\AxiomC{}
	\LeftLabel{\scriptsize Nil}
	\UnaryInfC{$\Gt{\Gamma};\es;\es
				       \proves \inact \hastype \Proc$}
	\LeftLabel{\scriptsize End}
	\UnaryInfC{$\Gt{\Gamma};\es;\prop_k:\tinact
				       \proves \inact \hastype \Proc$}
	\AxiomC{}
	\LeftLabel{\scriptsize PolyVar}
	\UnaryInfC{$\Gt{\Gamma};\Gt{\Lambda};\es
				        \proves \widetilde x \hastype \widetilde M$}
	\LeftLabel{\scriptsize Send}
	\BinaryInfC{$\Gt{\Gamma};\Gt{\Lambda}; \dual {\prop_k}:
				        \btout{\widetilde M} \tinact
				        \proves \propout{k}{\widetilde x} \inact \hastype \Proc$}
	\DisplayProof
 \end{align}

\begin{align}
	\label{pt:abspar}
	\AxiomC{\eqref{pt:abs-send}}
	\AxiomC{\eqref{eq:abs-ih-1}}
	\LeftLabel{\scriptsize (\lemref{lem:weaken}) with $\tilde x$}
	\UnaryInfC{$\Gt{\Gamma};\es;\Gt{\Delta_1} \cat \Theta_1
 		\proves \B{k}{\tilde x}{P\subst{y_1}{y}} \hastype \Proc$}
	\LeftLabel{\scriptsize Par}
	\BinaryInfC{$\Gt{\Gamma};\Gt{\Lambda};\Gt{\Delta_1} \cat \Theta
					     \proves \propout{k}{\widetilde x} \inact
               \Par \B{k}{\tilde x}{P \subst{y_1}{y}} \hastype \Proc$}
  \DisplayProof
\end{align}

The following tree proves this part:
\begin{align}
	\AxiomC{\eqref{pt:abspar}}
	\AxiomC{\eqref{eq:abs-ih-2}}
	\LeftLabel{\scriptsize Abs}
	\BinaryInfC{$\Gt{\Gamma \setminus y_1};\Gt{\Lambda};\Gt{\Delta_1 \setminus \Delta_2} \cat
	             \Theta \proves
	            \abs{\widetilde{y}:\Gt{C}}{\propout{k}{\widetilde x} \inact
              	 \Par \B{k}{\tilde x}{P \subst{y_1}{y}}} \hastype \Proc$}
	\DisplayProof
\end{align}
\item $\leadsto = \rightarrow$.
By assumption, $\Gamma;\Lambda;\Delta \proves V \hastype \shot{C}$.
In this case Rule~\textsc{Prom} can be applied:
\begin{align}
	\label{eq:abs-sub2-ih-1}
	\AxiomC{$\Gamma;\es;\Delta \proves P\subst{y_1}{y} \hastype \Proc$}
		\AxiomC{$\Gamma;\es;\Delta \proves y_1 \hastype C$}
		\LeftLabel{\scriptsize Abs}
	\BinaryInfC{$\Gamma \setminus y_1;\es;\es
					     \proves \abs{y_1}P \subst{y_1}{y} \hastype \lhot{C}$}
	\LeftLabel{\scriptsize Prom}
	\UnaryInfC{$\Gamma \setminus y_1;\es;\es
					     \proves \abs{y_1}P \subst{y_1}{y}
					     \hastype \shot{C}$}
	\DisplayProof
\end{align}

Let $\widetilde x = \fv{P}$, $\Gamma_1 = \Gamma \setminus \widetilde x$,
and let $\Theta_1$ be defined as in Part~2. Then, by IH (Part 1) on the first assumption of Rule~\textsc{Abs} in
\eqref{eq:abs-sub2-ih-1} we have:
\begin{align}
 		\label{eq:abs-ih-1-2}
 		\Gt{\Gamma_1};\es; \Gt{\Delta} \cat \Theta_1 \proves \B{k}{\tilde x}{P\subst{y_1}{y}}
    \hastype \Proc
 \end{align}
 By \defref{def:typesdecomp} and \defref{def:typesdecompenv} and
\eqref{eq:abs-sub2-ih-1} we have:
\begin{align}
	\label{eq:abs-ih-2-2}
	\Gt{\Gamma};\es;\Gt{\Delta}\proves \widetilde y \hastype \Gt{C}
\end{align}
\noindent where $\widetilde y = (y_1,\ldots,y_{\len{\Gt{C}}})$.

Let $\Theta$ be defined as in Part 2.
By \tabref{t:bdowncore} we have:
\begin{align*}
	\V{k}{\tilde x}{\abs{y}{P}} = \abs{\widetilde{y}:\Gt{C}}
				{\news{\widetilde \prop}\propout{k}{\widetilde x} \inact
  				\Par \B{k}{\tilde x}{P \subst{y_1}{y}}}
\end{align*}
\noindent where $\widetilde \prop = (\prop_k,\ldots,\prop_{k+\len{V}-1})$.

We use some auxiliary derivations:
\begin{align}
	\label{pt:abssend1}
	\AxiomC{}
	\LeftLabel{\scriptsize Nil}
	\UnaryInfC{$\Gt{\Gamma};\es;\es
				\proves \inact \hastype \Proc$}
	\LeftLabel{\scriptsize End}
	\UnaryInfC{$\Gt{\Gamma};\es;\prop_k:\tinact
				\proves \inact \hastype \Proc$}
	\AxiomC{}
	\LeftLabel{\scriptsize PolyVar}
	\UnaryInfC{$\Gt{\Gamma};\Gt{\Lambda};\es
				\proves \widetilde x \hastype \widetilde M$}
	\LeftLabel{\scriptsize Send}
	\BinaryInfC{$\Gt{\Gamma};\es; \dual {\prop_k}:
				\btout{\widetilde M} \tinact
				\proves \bout{\dual {\prop_{k}}}{\widetilde x} \inact \hastype \Proc$}
	\DisplayProof
\end{align}

\begin{align}
	\label{pt:abspar2}
	\AxiomC{\eqref{pt:abssend1}}
	\AxiomC{\eqref{eq:abs-ih-1}}
	\LeftLabel{\scriptsize (\lemref{lem:weaken}) with $\tilde x$}
	\UnaryInfC{$\Gt{\Gamma};\es;\Gt{\Delta} \cat \Theta_1
 		\proves \B{k}{\tilde x}{P\subst{y_1}{y}} \hastype \Proc$}
	\LeftLabel{\scriptsize Par}
	\BinaryInfC{$\Gt{\Gamma};\es;\Gt{\Delta} \cat
				\Theta
				\proves \bout{\dual {\prop_{k}}}{\widetilde x} \inact
  					\Par \B{k}{\tilde x}{P \subst{y_1}{y}} \hastype \Proc$}
	\LeftLabel{\scriptsize ResS}
	\UnaryInfC{$\Gt{\Gamma};\es;
					\Gt{\Delta}
					\proves \news{\widetilde \prop}\bout{\dual{\prop_{k}}}{\widetilde x} \inact
  \Par \B{k}{\tilde x}{P \subst{y_1}{y}} \hastype \Proc$}
  \DisplayProof
\end{align}

The following tree concludes this part (and the proof):
\begin{align}
	\AxiomC{\eqref{pt:abspar2}}
	\AxiomC{\eqref{eq:abs-ih-2-2}}
	\LeftLabel{\scriptsize Abs}
	\BinaryInfC{$\Gt{\Gamma \setminus y_1};\es; \es
	\proves \abs{\widetilde{y}:\Gt{C}}
		{\news{\widetilde \prop} \propout{k}{\widetilde x} \inact
  \Par \B{k}{\tilde x}{P \subst{y_1}{y}}} \hastype \Proc$}
	\DisplayProof
\end{align} 
\end{itemize}
\end{enumerate}
\end{proof}

\subsection{Proof of \thmref{t:decompcore}}
\label{app:decompcore}
\

\begin{proof}
By assumption, $\Gamma;\es;\Delta \proves P \hastype \Proc$. Then,
	by applying \lemref{lem:subst} we have:
	\begin{align}
		\label{eq:typedec-subst}
		\Gamma\sigma;\es;\Delta \sigma \proves P \sigma \hastype \Proc
	\end{align}
		By \defref{def:decomp}, we
	shall prove the following judgment:
	\begin{align}
		\Gt{\Gamma \sigma};\es;\Gt{\Delta \sigma} \proves \news{\widetilde c}{(\propout{k}{}\inact \Par \B{k}{\epsilon}{P\sigma})} \hastype \Proc
	\end{align}
		where $\widetilde \prop = (\prop_k,\ldots,\prop_{k+\len{P}-1})$
	and $k > 0$. Since $P\sigma$ is an initialized process, 
we apply \thmref{t:typecore} to \eqref{eq:typedec-subst} to get:
	\begin{align}
		\label{eq:typedec-bdown}
		\Gt{\Gamma\sigma};\es;\Gt{\Delta \sigma}\cat\Theta
		\proves \B{k}{\epsilon}{P\sigma} \hastype \Proc
	\end{align}
	\noindent where
	$\Theta$ is balanced with
	$\dom{\Theta} = \{\prop_k,\prop_{k+1},\ldots,\prop_{k+\len{P}-1} \}
	\cup \{\dual{\prop_{k+1}},\ldots,\dual{\prop_{k+\len{P}-1}} \}$ and
	$\Theta(\prop_k)=\btinp{\cdot}\tinact$.
	By assumption,  $\fv{P} = \es$.
%
%

	The following tree concludes the proof:
	\begin{align}
	\AxiomC{}
	\LeftLabel{\scriptsize Nil}
	\UnaryInfC{$\Gt{\Gamma\sigma};\es;\es \proves \inact \proves \Proc$}
	\AxiomC{}
	\LeftLabel{\scriptsize Send}
	\BinaryInfC{$\Gt{\Gamma\sigma};\es;\dual{\prop_k}:\btout{\cdot}\tinact
				\proves \propout{k}{}\inact \hastype \Proc$}
	\AxiomC{\eqref{eq:typedec-bdown}}
	\LeftLabel{\scriptsize Par}
		\BinaryInfC{$\Gt{\Gamma\sigma};\es;\Gt{\Delta\sigma}\cat
					\dual{\prop_k}:\btout{\cdot}\tinact
					\cat \Theta \proves
					\propout{k}{}\inact \Par \B{k}{\epsilon}{P\sigma}
					\hastype \Proc$}
		\LeftLabel{\scriptsize PolyResS}
		\UnaryInfC{$\Gt{\Gamma\sigma};\es;\Gt{\Delta \sigma}
						\proves
						\news{\widetilde c}{(\propout{k}{}\inact \Par \B{k}{\epsilon}{P\sigma})} \hastype \Proc$}
		\DisplayProof
	\end{align}
	 Rule PolyResS can be used: the environment obtained by extending $\Theta$ with $\dual{\prop_k}:\btout{\cdot}\tinact$ is balanced, and for each $\prop_j \in \widetilde c$ such that $\prop_j: S$ there is a $\dual{\prop_j}:T$, with $S \dualof T$.
\end{proof}

\section{Appendix to \secref{ss:exti}}
\label{app:exti}

\thmref{t:typecore} holds also for the extension with selection and branching.
We now present the additional cases of the proof.

\begin{proof}[Proof (Extension of \thmref{t:typecore} with Selection/Branching)]
The proof is as before.
We need to consider two additional cases for Part~(1) of the 
theorem:
\begin{enumerate}
	\item Case $P = \bbra{u_i}{l_j:P_j}_{j \in I}$. For this case
	Rule~\textsc{Bra} can be applied:
	\begin{align}
		\label{pt:braitr}
		\AxiomC{$\forall j \in I$}
		\AxiomC{$\Gamma;\Lambda;\Delta \cat u_i:S_j \proves P_j \hastype \Proc$}
		\LeftLabel{\scriptsize Bra}
		\BinaryInfC{$\Gamma;\Lambda;\Delta \cat
						u_i:\btbra{l_j:S_j}_{j \in I} \proves
						\bbra{u_i}{l_j:P_j}_{j \in I} \hastype \Proc$}
		\DisplayProof
	\end{align}

	Let $\widetilde x = \fv{P}$ and $\widetilde x_j = \fv{P_j}$ for all $j \in I$. 
	Notice that $\widetilde x_j \subseteq \widetilde x$, for all $j \in I$.
	Also,
	$\Gamma_{j,1}  = \Gamma \setminus \widetilde x_j$ and let $\Theta_j$ be a balanced environment
	such that $\dom{\Theta_j}=\{\prop_{k+1},\ldots,\prop_{k+\len{P_j}}\} \cup \{\dual{\prop_{k+2}},\ldots,\dual{\prop_{k+\len{P_j}}}\}$, $\Theta_j(\prop_{k+1})=\btinp{\widetilde M_j} \tinact$ with
	$\widetilde M_j = (\Gt{\Gamma},\Gt{\Lambda})(\widetilde x_j)$. Then,
	by IH on the assumption of \eqref{pt:braitr} we have for all $j \in I$:
	\begin{align}
	\label{eq:bra-ih}
			\Gt{\Gamma \setminus \widetilde x_j};\es;\Gt{\Delta \cat u_i:S_j} \cat
			\Theta_j \proves \B{k+1}{\tilde x_j}{P_j} \hastype \Proc
	\end{align}

	We define $\Theta = \prop_k:\btinp{\widetilde M}\tinact$ with
	$(\Gt{\Gamma},\Gt{\Lambda})(\widetilde x) = \widetilde M$.
	By construction, $\Theta$ is balanced.

By \defref{def:sizeproc-bra},
	$\len{P}=1$ and so $\text{dom}(\Theta)=\{c_k\}$.
		Let $\widetilde y^u_j = (y^u_1,\ldots,y^u_{\len{\Gt{S_j}}})$.
	By the function defined in \tabref{t:bdown-selbra} we have:
	\begin{align*}
	\B{k}{\tilde x}{\bbra{u_i}{l_j:P_j}_{j \in I}} =
		\propinp{k}{\widetilde x} \bbra{u_i}{l_j :
  		 \bbout{u_i}{N_{u,j}}\inact}_{j \in I}
	\end{align*}
	\noindent where $N_{u,j} = \abs{\widetilde y^u_j : \Gt{S_j}}{\news{\widetilde \prop_j}
 		 (\propout{k+1}{\widetilde x} \inact
  		\Par
  		\B{k+1}{\tilde x}{P_j \subst{y^u_1}{u_i}})}$
  		and
	$\widetilde \prop_j = (\prop_{k+1},\ldots,\prop_{k+\len{P_j}})$. Let $\Gamma_1 = \Gamma \setminus \widetilde x$. 
	We shall prove the following judgment:
	\begin{align*}
		\Gt{\Gamma_1};\es;\Gt{\Delta \cat
						u_i:\btbra{l_j:S_j}_{j \in I}} \cat \Theta \proves
						\B{k}{\tilde x}{\bbra{u_i}{l_j:P_j}_{j \in I}}
						\hastype \Proc
	\end{align*}
	We define $\Theta_{j,1} = \Theta_j \cat
	\dual{\prop_{k+1}}:\btout{\widetilde M}\tinact$.
	We use some auxiliary derivations for all $j \in I$:

	\begin{align}
		\label{pt:bra-psend1}
		\AxiomC{}
		\LeftLabel{\scriptsize Nil}
		\UnaryInfC{$\Gt{\Gamma};\es;\es \proves \inact \hastype \Proc$}
		\LeftLabel{\scriptsize End}
		\UnaryInfC{$\Gt{\Gamma};\es;\dual{\prop_{k+1}}:\tinact \proves \inact \hastype \Proc$}
		\AxiomC{}
		\LeftLabel{\scriptsize PolyVar}
		\UnaryInfC{$\Gt{\Gamma};\Gt{\Lambda};\es
						\proves \widetilde x \hastype \widetilde M$}
		\LeftLabel{\scriptsize PolySend}
		\BinaryInfC{$\Gt{\Gamma};\Gt{\Lambda};\dual{\prop_{k+1}}:
							\btout{\widetilde M} \tinact \proves
						\propout{k+1}{\widetilde x} \inact \hastype \Proc$}
		\DisplayProof
	\end{align}

	\begin{align}
		\label{pt:bra-par2}
		\AxiomC{\eqref{pt:bra-psend1}}
		\AxiomC{\eqref{eq:bra-ih}}
		\LeftLabel{\scriptsize (\lemref{lem:subst}) with $\sigma$}
		\UnaryInfC{$\Gt{\Gamma \setminus \widetilde x_j};\es;\Gt{\Delta} \cat \widetilde y^u_j : \Gt{S_j} \proves
						\B{k+1}{\tilde x}{P_j \subst{y^u_1}{u_i}} \hastype \Proc$}
		\LeftLabel{\scriptsize (\lemref{lem:weaken}) with $\tilde x_j$}
		\UnaryInfC{$\Gt{\Gamma};\es;\Gt{\Delta} \cat \widetilde y^u_j : \Gt{S_j} \cat \Theta_j \proves
						\B{k+1}{\tilde x}{P_j \subst{y^u_1}{u_i}} \hastype \Proc$}
		\LeftLabel{\scriptsize Par}
		\BinaryInfC{$\Gt{\Gamma};\Gt{\Lambda};
							\Gt{\Delta} \cat \widetilde y^u_j : \Gt{S_j}
							\cat \Theta_{j,1}
							\proves
							\propout{k+1}{\widetilde x} \inact
  			\Par
  			\B{k+1}{\tilde x}{P_j \subst{y^u_1}{u_i}} \hastype \Proc$}
  		\LeftLabel{\scriptsize PolyResS}
  		\UnaryInfC{$\Gt{\Gamma};\Gt{\Lambda};
							\Gt{\Delta} \cat \widetilde y^u_j : \Gt{S_j}
							\proves
							\news{\widetilde \prop_j}
							(\propout{k+1}{\widetilde x} \inact
  							\Par
  			\B{k+1}{\tilde x}{P_j \subst{y^u_1}{u_i}}) \hastype \Proc$}
  		\DisplayProof
	\end{align}
	\noindent where $\sigma = \subst{\widetilde y^u_j}{\widetilde u}$ with
	$\widetilde u = (u_i,\ldots,u_{i+\len{\Gt{S_j}}-1})$.

	\begin{align}
		\label{pt:braabs}
		\AxiomC{\eqref{pt:bra-par2}}
		\AxiomC{}
		\LeftLabel{\scriptsize PolySess}
		\UnaryInfC{$\Gt{\Gamma};\es;\widetilde y_j : \Gt{S_j}
						\proves \widetilde y_j \hastype \Gt{S_j}$}
		\LeftLabel{\scriptsize PolyAbs}
		\BinaryInfC{$\Gt{\Gamma};\Gt{\Lambda};\Gt{\Delta} 		\proves\abs{\widetilde y^u_j : \Gt{S_j}}
			\news{\widetilde \prop_j}(
 		 	{\propout{k+1}{\widetilde x} \inact
  			\Par
  			\B{k+1}{\tilde x}{P_j \subst{y^u_1}{u_i}})} \hastype \Proc$}
  			\DisplayProof
	\end{align}

	\begin{align}
	\label{pt:brares}
		\AxiomC{}
		\LeftLabel{\scriptsize Nil}
		\UnaryInfC{$\Gt{\Gamma};\es;\es \proves \inact$}
		\LeftLabel{\scriptsize End}
		\UnaryInfC{$\Gt{\Gamma};\es;u_i:\tinact \proves \inact$}
		\AxiomC{\eqref{pt:braabs}}
		\LeftLabel{\scriptsize Send}
		\BinaryInfC{\begin{tabular}{c}
					$\Gt{\Gamma};\Gt{\Lambda};\Gt{\Delta}
					\cat u_i:\btout{\lhot{\Gt{S_j}}} \tinact
					\proves$
					$\bbout{u_i}
					{N_{u,j}}\inact$
					\end{tabular}}
		\DisplayProof
	\end{align}

		Here we may notice that by \defref{def:typesdecompenv} and
	\defref{def:typesdecomp-bra} we have:
	$$\Gt{u_i:\btbra{l_j : S_j}_{j \in I}} = u_i:\btbra{l_i:\btout{\lhot{\Gt{S_j}}} \tinact}_{j \in I}$$

	\begin{align}
		\label{pt:bra-bra}
		\AxiomC{$\forall j \in I$}
		\AxiomC{\eqref{pt:brares}}
		\LeftLabel{\scriptsize Bra}
		\BinaryInfC{\begin{tabular}{c}
						$\Gt{\Gamma};\Gt{\Lambda};\Gt{\Delta \cat
						u_i:\btbra{l_j:S_j}_{j \in I}} \proves$
						$\bbra{u_i}{l_j :
  						\bbout{u_i}
						{N_{u,j}}\inact}_{j \in I} 						\hastype \Proc$
  						\end{tabular}}
		\LeftLabel{\scriptsize End}
		\UnaryInfC{\begin{tabular}{c}
	$\Gt{\Gamma};\Gt{\Lambda};\Gt{\Delta \cat
						u_i:\btbra{l_j:S_j}_{j \in I}} \cat \prop_k : \tinact \proves$ \\
				$\bbra{u_i}{l_j :
  						\bbout{u_i}
						{N_{u,j}}\inact}_{j \in I} 						\hastype \Proc$
						\end{tabular}}
  		\DisplayProof
	\end{align}

	The following tree proves this case:
		\begin{align*}
	\label{pt:selitr}
		\AxiomC{\eqref{pt:bra-bra}}
		\AxiomC{}
		\LeftLabel{\scriptsize PolyVar}
		\UnaryInfC{$\Gt{\Gamma_1};\Gt{\Lambda};\es \proves \widetilde x
					\hastype \widetilde M$}
		\LeftLabel{\scriptsize PolyRcv}
		\BinaryInfC{
$\Gt{\Gamma_1};\es;\Gt{\Delta \cat
						u_i:\btbra{l_j:S_j}_{j \in I}}
				\cat \Theta
				\proves
						\B{k}{\tilde x}{\bbra{u_i}{l_j:P_j}_{j \in I}} \hastype \Proc$
					}
		\DisplayProof
	\end{align*}


	\item Case $P=\bsel{u_i}{l_j}P'$. For this case Rule~\textsc{Sel} can
	be applied:
	\begin{align}
		\AxiomC{$\Gamma;\Lambda;\Delta \cat u_i: S_j \proves P' \hastype \Proc$}
		\AxiomC{$j \in I$}
		\LeftLabel{\scriptsize Sel}
		\BinaryInfC{$\Gamma;\Lambda;\Delta\cat u_i:\btsel{l_j:S_j}_{j \in I}
						\proves \bsel{u_i}{l_j}P' \hastype \Proc$}
		\DisplayProof
	\end{align}

	Let $\widetilde x = \fv{P'}$ and $\Gamma_1 = \Gamma \setminus
	\widetilde x$. Also, let $\Theta_1$ be a balanced environment such that
	$\dom{\Theta_1}=\{\prop_{k+2},\ldots,\prop_{k+1+\len{P'}}\}
	\cup \{\dual{\prop_{k+3}},\ldots,\dual{\prop_{k+1+\len{P'}}}\}$ and
	$\Theta_1(\prop_{k+2})=\btinp{\widetilde M} \tinact$ where
	$\widetilde M = (\Gt{\Gamma},\Gt{\Lambda})(\widetilde x)$. Then,
	by IH on the first assumption of \eqref{pt:selitr} we have:
	\begin{align}
		\label{eq:sel-ih}
		\Gt{\Gamma_1};\es;\Gt{\Delta \cat u_i:S_j} \cat \Theta_1 \proves
		\B{k+2}{\tilde x}{P'} \hastype \Proc
	\end{align}

	We define $\Theta = \Theta_1 \cat \Theta'$ where:
	\begin{align*}
		\Theta' = \prop_k:\btinp{\widetilde M} \tinact \cat
					\prop_{k+1}:\btinp{\lhot{\Gt{\dual{S_j}}}}\tinact \cat
					\dual{\prop_{k+1}}:\btout{\lhot{\Gt{\dual{S_j}}}}\tinact \cat
					\dual{\prop_{k+2}}:\btout{\widetilde M} \tinact
	\end{align*}
	By \defref{def:sizeproc-bra}, $\len{P}=\len{P'}+2$, so
	$$\dom{\Theta}=\{\prop_k,\prop_{k+1},\ldots,\prop_{k+\len{P}-1}\}
	\cup \{\dual{\prop_{k+1}},\ldots,\dual{\prop_{k+\len{P}-1}}\}$$
	By construction $\Theta$ is balanced since $\Theta(\prop_{k+2}) \dualof \Theta(\dual{\prop_{k+2}})$ and
  	 $\Theta_1$ and $\Theta'$ are balanced.
	Let $\widetilde y = (y_1,\ldots,y_{\len{\Gt{S_j}}})$.
	By \tabref{t:bdown-selbra}:
	\begin{align*}
	\B{k+2}{\tilde x}{\bsel{u_i}{l_j}P'}=
		\propinp{k}{\widetilde x}\propout{k+1}{\abs{\widetilde y}{\bsel{u_i}
		{l_j}{\binp{u_i}{z}\propout{k+2}{\widetilde x}\appl{z}{\widetilde y}}}}
		\inact
  		\Par \\
  		\news{\widetilde u:\Gt{S_j}}{(\propinp{k+1}{y}{\appl{y}
  		{\widetilde{\dual u}}} \Par
  		\B{k+2}{\tilde x}{P'\incrname{u}{i}})}
	\end{align*}
	\noindent where $\widetilde u=(u_{i+1},\ldots,u_{i+\len{\Gt{S_j}}})$ and
  $\widetilde {\dual u}=(\dual {u_{i+1}},\ldots,\dual {u_{i+\len{\Gt{S_j}}}})$.
	We shall prove the following judgment:
	\begin{align*}
		\Gt{\Gamma};\es;\Gt{\Delta\cat u_i:\btsel{l_j:S_j}_{j \in I}}
		\cat \Theta
		\proves
		\B{k+2}{\tilde x}{\bsel{u_i}{l_j}P'} \hastype \Proc
	\end{align*}

	For this, we will use the following auxiliary derivations for the left-hand side component:
	\begin{align}
	\label{pt:sel-polyapp}
	\AxiomC{}
		\LeftLabel{\scriptsize LVar}
		\UnaryInfC{$\Gt{\Gamma};z:\lhot{\Gt{\dual {S_j}}};\es \proves
						z \hastype \lhot{\Gt{\dual {S_j}}}$}
		\AxiomC{}
		\LeftLabel{\scriptsize PolySess}
		\UnaryInfC{
		\begin{tabular}{c}
		$\Gt{\Gamma};\es;\widetilde y : \Gt{\dual {S_j}}
					\proves$ \\ $\widetilde y \hastype \Gt{\dual {S_j}}$
					\end{tabular}}
		\LeftLabel{\scriptsize PolyApp}
		\BinaryInfC{$\Gt{\Gamma};z:\lhot{\Gt{\dual {S_j}}};
					\widetilde y : \Gt{\dual {S_j}}
					\proves
					\appl{z}{\widetilde y} \hastype \Proc$}
	\LeftLabel{\scriptsize End}
	\UnaryInfC{$\Gt{\Gamma};z:\lhot{\Gt{\dual {S_j}}};
					\widetilde y : \Gt{\dual {S_j}} \cat
					\dual{\prop_{k+2}}:\tinact
					\proves
					\appl{z}{\widetilde y} \hastype \Proc$}
		\DisplayProof
	\end{align}
	\begin{align}
		\label{pt:sel-send2}
		\AxiomC{\eqref{pt:sel-polyapp}}
		\AxiomC{$\Gt{\Gamma};\Gt{\Lambda};\es \proves
					\widetilde x \hastype \widetilde M$}
		\LeftLabel{\scriptsize PolySend}
		\BinaryInfC{$\Gt{\Gamma};\Gt{\Lambda} \cat z:\lhot{\Gt{\dual {S_j}}};
					\widetilde y : \Gt{\dual {S_j}} \cat
					\dual{\prop_{k+2}}:\btout{\widetilde M}\tinact \proves
					\propout{k+2}{\widetilde x}\appl{z}{\widetilde y} \hastype
					\Proc$}
		\LeftLabel{\scriptsize End}
		\UnaryInfC{$\Gt{\Gamma};\Gt{\Lambda} \cat z:\lhot{\Gt{\dual {S_j}}};
					\widetilde y : \Gt{\dual{S_j}} \cat
					\dual{\prop_{k+2}}:\btout{\widetilde M}\tinact \cat u_i: \tinact \proves
					\propout{k+2}{\widetilde x}\appl{z}{\widetilde y} \hastype
					\Proc$}
		\DisplayProof
	\end{align}
	\begin{align}
		\label{pt:sel-sel}
		\AxiomC{\eqref{pt:sel-send2}}
		\AxiomC{$\Gt{\Gamma};z:\lhot{\Gt{\dual {S_j}}};\es \proves
					z \hastype \lhot{\Gt{\dual {S_j}}}$}
		\LeftLabel{\scriptsize Rcv}
		\BinaryInfC{\begin{tabular}{c}
	$\Gt{\Gamma};\Gt{\Lambda};
					u_i:\btinp{\lhot{\Gt{\dual {S_j}}}}\tinact \cat
					\widetilde y : \Gt{\dual {S_j}} \cat
					\dual{\prop_{k+2}}:\btout{\widetilde M} \tinact
					\proves$ \\
					$\binp{u_i}{z}\propout{k+2}
						{\widetilde x}\appl{z}{\widetilde y} \hastype
					\Proc$
					\end{tabular}}
		\AxiomC{$j \in I$}
		\LeftLabel{\scriptsize Sel}
		\BinaryInfC{\begin{tabular}{c}
		$\Gt{\Gamma};\Gt{\Lambda};
					u_i:\btsel{\btinp{\lhot{\Gt{\dual {S_j}}}}\tinact} \cat
					\widetilde y : \Gt{\dual {S_j}} \cat
					\dual{\prop_{k+2}}:\btout{\widetilde M} \tinact
					\proves$ \\
					$\bsel{u_i}{l_j}{\binp{u_i}{z}\propout{k+2}
						{\widetilde x}\appl{z}{\widetilde y}} \hastype
					\Proc$
					\end{tabular}}
		\DisplayProof
	\end{align}

	\begin{align}
		\label{pt:sel-abs}
		\AxiomC{\eqref{pt:sel-sel}}
		\AxiomC{}
		\LeftLabel{\scriptsize PolyVar}
		\UnaryInfC{$\Gt{\Gamma};\es;\widetilde y : \Gt{\dual {S_j}} \proves
					\widetilde y \hastype \Gt{\dual {S_j}}$}
		\LeftLabel{\scriptsize Abs}
		\BinaryInfC{$\Gt{\Gamma};\Gt{\Lambda};
					\dual{\prop_{k+2}}:\btout{\widetilde M}\tinact
					\proves
					\abs{\widetilde y}															{\bsel{u_i}	{l_j}{\binp{u_i}{z}\propout{k+2}{\widetilde 					x}\appl{z}{\widetilde y}}} \hastype
					\lhot{\Gt{\dual {S_j}}}
					$}
		\DisplayProof
	\end{align}

	\begin{align}
		\label{pt:sel-psend1}
		\AxiomC{}
		\LeftLabel{\scriptsize Nil}
		\UnaryInfC{$\Gt{\Gamma};\es;\es \proves \inact$}
		\LeftLabel{\scriptsize End}
		\UnaryInfC{$\Gt{\Gamma};\es;u_i:\tinact \proves \inact$}
		\AxiomC{\eqref{pt:sel-abs}}
		\LeftLabel{\scriptsize Send}
		\BinaryInfC{\begin{tabular}{c}
					$\Gt{\Gamma};
					\Gt{\Lambda};
					u_i:\btsel{\btinp{\lhot{\Gt{\dual {S_j}}}}\tinact}\cat
					\dual{\prop_{k+1}}:\btout{\lhot{\Gt{\dual {S_j}}}}\tinact \cat$ \\
					$\dual{\prop_{k+2}}:\btout{\widetilde M}\tinact \proves$
					$\propout{k+1}{\abs{\widetilde y}{\bsel{u_i}
					{l_j}{\binp{u_i}{z}\propout{k+2}
					{\widetilde x}\appl{z}{\widetilde y}}}}\inact \hastype \Proc$
					\end{tabular}}
	\LeftLabel{\scriptsize End}
	\UnaryInfC{\begin{tabular}{c}
					$\Gt{\Gamma};
					\Gt{\Lambda};
					u_i:\btsel{\btinp{\lhot{\Gt{\dual {S_j}}}}\tinact}\cat
					\dual{\prop_{k+1}}:\btout{\lhot{\Gt{\dual {S_j}}}}\tinact \cat$ \\
					$\dual{\prop_{k+2}}:\btout{\widetilde M}\tinact
					\cat
					\prop_k:\tinact \proves$
					$\propout{k+1}{\abs{\widetilde y}{\bsel{u_i}
					{l_j}{\binp{u_i}{z}\propout{k+2}
					{\widetilde x}\appl{z}{\widetilde y}}}}\inact \hastype \Proc$
					\end{tabular}}
		\DisplayProof
	\end{align}

	\begin{align}
		\label{pt:sel-rcv}
		\AxiomC{\eqref{pt:sel-psend1}}
		\AxiomC{}
		\LeftLabel{\scriptsize PolyVar}
		\UnaryInfC{$\Gt{\Gamma};\Gt{\Lambda};\es \proves
					\widetilde x \hastype \widetilde M$}
		\LeftLabel{\scriptsize PolyRcv}
		\BinaryInfC{\begin{tabular}{c}$\Gt{\Gamma_1};\es;
					u_i:\btsel{\btinp{\lhot{\Gt{\dual {S_j}}}}\tinact} \cat
					(\Theta' \setminus \prop_{k+1})
					\proves$ \\
					$\propinp{k}{\widetilde x}\propout{k+1}{\abs{\widetilde y}
					{\bsel{u_i}{l_j}{\binp{u_i}{z}
					\propout{k+2}{\widetilde x}
					\appl{z}{\widetilde y}}}}\inact \hastype \Proc$
					\end{tabular}}
		\DisplayProof
	\end{align}
	
	The following auxiliary derivations are used to type the right-hand side component:

	\begin{align}
		\label{pt:sel-varapp2}
		\AxiomC{}
		\LeftLabel{\scriptsize LVar}
		\UnaryInfC{
		\begin{tabular}{c}
$\Gt{\Gamma_1};y:\lhot{\Gt{\dual {S_j}}};\es
					\proves$ \\ 
					$y \hastype \lhot{\Gt{\dual {S_j}}}$
					\end{tabular}}
		\AxiomC{}
		\LeftLabel{\scriptsize PolySess}
		\UnaryInfC{$\Gt{\Gamma_1};\es;\widetilde{\dual u}:\Gt{\dual {S_j}}
					\proves
						\widetilde{\dual u} \hastype \Gt{\dual {S_j}}$}
		\LeftLabel{\scriptsize App}
		\BinaryInfC{$\Gt{\Gamma_1};y:\lhot{\Gt{\dual {S_j}}};
						\widetilde{\dual u}:\Gt{\dual {S_j}} \proves
						\appl{y}{\widetilde{\dual u}} \hastype \Proc
						$}
		\DisplayProof
	\end{align}

	\begin{align}
		\label{pt:sel-rcv2}
		\AxiomC{\eqref{pt:sel-varapp2}}
		\AxiomC{}
		\LeftLabel{\scriptsize LVar}
		\UnaryInfC{$\Gt{\Gamma};y:\lhot{\Gt{\dual {S_j}}};\es \proves
						y \hastype \lhot{\Gt{\dual {S_j}}}$}
		\LeftLabel{\scriptsize Rcv}
		\BinaryInfC{$\Gt{\Gamma_1};\es;
					\prop_{k+1}:\btinp{\lhot{\Gt{\dual {S_j}}}}\tinact \cat
					\widetilde{\dual u}:\Gt{\dual {S_j}}
					\proves
					\propinp{k+1}{y}{\appl{y}{\widetilde{\dual u}}} \hastype \Proc$}
		\DisplayProof
	\end{align}

	\begin{align}
		\label{pt:sel-ress2}
		\AxiomC{\eqref{pt:sel-rcv2}}
		\AxiomC{\eqref{eq:sel-ih}}
		\LeftLabel{\scriptsize (\lemref{lem:subst}) with $\sigma$}
		\UnaryInfC{$\Gt{\Gamma};\es;\Gt{\Delta \cat u_{i+1}:S_j} \proves
					\B{k+2}{\tilde x}{P'\incrname{u}{i}} \hastype \Proc$}
		\LeftLabel{\scriptsize (\lemref{lem:strength}) with $\tilde x$}
		\UnaryInfC{$\Gt{\Gamma_1};\es;\Gt{\Delta \cat u_{i+1}:S_j}
					 \proves
					\B{k+2}{\tilde x}{P'\incrname{u}{i}} \hastype \Proc
					$}
		\LeftLabel{\scriptsize Par}
		\BinaryInfC{\begin{tabular}{c}
					$\Gt{\Gamma_1};\es;\Gt{\Delta \cat u_{i+1}:S_j}  \cat
					\prop_{k+1}:\btinp{\lhot{\Gt{\dual {S_j}}}}\tinact
					\proves$ \\
					$\propinp{k+1}{y}{\appl{y}
  					{\widetilde{\dual u}}} \Par
  					\B{k+2}{\tilde x}{P'\incrname{u}{i}} \hastype \Proc
					$
					\end{tabular}}
		\LeftLabel{\scriptsize PolyResS}
		\UnaryInfC{\begin{tabular}{c}
					$\Gt{\Gamma_1};\es;\Gt{\Delta} \cat
					\Theta_1 \cat
					\prop_{k+1}:\btinp{\lhot{\Gt{\dual {S_j}}}}\tinact
					\proves$ \\
					$\news{\widetilde u:\Gt{S_j}}{(\propinp{k+1}{y}{\appl{y}
  					{\widetilde{\dual u}}} \Par
  					\B{k+2}{\tilde x}{P'\incrname{u}{i}})} \hastype \Proc
  					$
  					\end{tabular}}
		\DisplayProof
	\end{align}
	\noindent where $\sigma = \subst{\widetilde u}{\widetilde n}$ with
	$\widetilde n = (u_i,\ldots,u_{i+\len{\Gt{S_j}}-1})$.
The following tree proves this case:
	\begin{align*}
		\AxiomC{\eqref{pt:sel-rcv}}
		\AxiomC{\eqref{pt:sel-ress2}}
		\LeftLabel{\scriptsize Par}
		\BinaryInfC{$\Gt{\Gamma_1};\es;\Gt{\Delta\cat
						u_i:\btsel{l_j:S_j}_{j \in I}} \cat \Theta
						\proves
					\B{k+2}{\tilde x}{\bsel{u_i}{l_j}P'} \hastype \Proc$}
		\DisplayProof
	\end{align*}
	\end{enumerate}
\end{proof}

\section{Appendix to \secref{ss:extii}}
\label{app:extii}

We use the following auxiliary lemma:

\begin{lemma}
	\label{lem:indexcor}
	Let $\widetilde z$ be tuple of channel names, $U$ a higher-order type, and $S$ a recursive 	session type.
	If $\widetilde z : \Rts{}{s}{\btout{U}S}$ and $k = f(\btout{U}S)$
	then $z_k:\trec{t}\btout{\Gt{U}}\vart{t}$.
\end{lemma}

\subsection{Proof of \thmref{t:typerec}}
\label{app:typerec}

\begin{proof}
  By mutual induction on the structure of $P$ and $V$. 
   Here, we analyze only Part~(1) of the theorem, as
   Part~(2) and Part~(3) are proven similarly: 

 \begin{enumerate}
 \item By assumption $\Gamma;\Lambda;\Delta \cat \envR \proves P \hastype \Proc$. There are four cases, depending on the shape of $P$. We consider two representative cases. We omit other cases as they are similar. 
 		\begin{enumerate}
 		\item Case $P = \bout{r}{V}P'$, when $V$ is not a variable.
 We consider case: $r:S \in \envR$.  For this case Rule~\textsc{Send} can be applied:
 \begin{align}
    \label{pt:r-outputitr}
    \AxiomC{$\Gamma;\Lambda_1;\Delta_1 \cat {\envR}_1 \proves P' \hastype \Proc$}
    \AxiomC{$\Gamma;\Lambda_2;\Delta_2 \cat {\envR}_2 \proves V \hastype U$}
    \AxiomC{$r:S \in {\envR}_1 \cat {\envR}_2$}
    \LeftLabel{\scriptsize Send}
    \TrinaryInfC{$\Gamma;\Lambda_1 \cat \Lambda_2;
    				\Delta_1 \cat \Delta_2 \cat
    				(({\envR}_1\cat{\envR}_2)
    			  \setminus \{ r:S \})
                  \cat r:\btout{U}S
                  \proves \bout{r}{V}P' \hastype \Proc$}
    \DisplayProof
  \end{align}
  Let $\widetilde w = \fv{P'}$ and $\degree = \len{V}$. Let $\Theta_1$,
  $\Theta_2$, $\Theta'$ and $\Theta$ be defined as in the corresponding case of
  \thmref{t:typecore} proof. Also, let
  $\Gamma'_1 = \Gamma \setminus \widetilde w$.
  We define:
  \begin{align}
  	\label{eq:prod_env}
  	{\envPropR}_i = \prod_{r \in {\dom{{\envR}_i}}} \prop^r :
  	\chtype{\lhot{\Rts{}{s}{{\envR}_i(r)}}} \ \text{for} \ i \in \{1,2\}
  \end{align}
  Then, by IH on the
  first assumption of \eqref{pt:r-outputitr} we have:
 	\begin{align}
 	\label{eq:r-case1-ih}
 		\Gt{\Gamma'_1}\cat {\envPropR}_1;\es;\Gt{\Delta_1}\cat\Theta_1
 		\proves \B{k+\degree+1}{\tilde w}{P'} \hastype \Proc
 	\end{align}
 	Let $\Gamma'_2 = \Gamma \setminus \widetilde y$.
 	Then, by IH (Part 2) on the
  second assumption of \eqref{pt:r-outputitr} we have:
  \begin{align}
  	\Gt{\Gamma'_2}\cat{\envPropR}_2;\Gt{\Lambda_2};\Gt{\Delta_2}\cat \Theta_2
  	\proves \V{k+1}{\tilde y}{V} \hastype \Gt{U}
  \end{align}
 		
	By \tabref{t:bdownrecur} we have:
	\begin{align*}
		\B{k}{\tilde x}{P} = \propinp{k}{\widetilde x}
      		\bbout{\prop^r}{N_V} & \inact \Par
      \B{k+\degree+1}{\tilde w}{P'} \\
      & \text{where} \ N_V = \abs{\widetilde z}
      		{\bout{z_{f(S)}}{\V{k+1}{\tilde y}{V}}}\propout{k+\degree+1}
      		{\widetilde w} \binp{\prop^r}{b}
			(\appl{b}{\widetilde z})
	\end{align*}

	We notice that $r \in \rfn{V}\cat\rfn{P}$ since $r$ has recursive type $S$.
  	Hence, by \eqref{eq:prod_env} we know $(\envPropR_1,\envPropR_2)(c^r)=\chtype{\lhot{\Rts{}{s}{S}}}$.
  	Further, we know that $S = \btout{U}S'$ and by \defref{def:decomptyp-rec},
  	$\Rts{}{s}{S}=\Rts{}{s}{S'}$.
  	So we define
  	$\envPropR = {\envPropR}_1 \cat {\envPropR}_2$. Let $\Gamma_1 = \Gamma \setminus \widetilde x$ where $\widetilde x = \fv{P}$. 
	We shall prove the following judgment:
	\begin{align*}
		\Gt{\Gamma_1} \cat \envPropR;\es; \Gt{\Delta_1 \cat \Delta_2} 
		\cat \Theta
                  \proves \B{k}{\tilde x}{\bout{r}{V}P'} \hastype \Proc
	\end{align*}

	We use auxiliary derivations:
	\begin{align}
		\label{pt:r-send3}
		\AxiomC{}
		\LeftLabel{\scriptsize LVar}
		\UnaryInfC{\begin{tabular}{c}$\Gt{\Gamma_1} \cat \envPropR;b:\lhot{\Rts{}{s}{S}};\es \proves$ \\
						$b \hastype \lhot{\Rts{}{s}{S}}$
						\end{tabular}}
		\AxiomC{}
		\LeftLabel{\scriptsize PolySess}
		\UnaryInfC{
		\begin{tabular}{c}
		$\Gt{\Gamma_1} \cat \envPropR;\es;\widetilde z:\Rts{}{s}{S}
					\proves$ \\ $\widetilde z \hastype \Rts{}{s}{S}$
					\end{tabular}}
		\LeftLabel{\scriptsize PolyApp}
		\BinaryInfC{$\Gt{\Gamma_1} \cat \envPropR;b:\lhot{\Rts{}{s}{S}};
					\widetilde z: \Rts{}{s}{S} \proves
						\appl{b}{\widetilde z} \hastype \Proc$}
		\DisplayProof
	\end{align}
	
	\begin{align}
		\label{pt:r-send2}
	\AxiomC{\eqref{pt:r-send3}}
		\AxiomC{}
	\LeftLabel{\scriptsize Sh}
		\UnaryInfC{\begin{tabular}{c}$\Gt{\Gamma} \cat \envPropR;\es;\es \proves$  \\
					$\prop^r \hastype \chtype{\lhot{\Rts{}{s}{S}}}$
					\end{tabular}}
		\AxiomC{}
		\LeftLabel{\scriptsize LVar}
		\UnaryInfC{\begin{tabular}{c}$\Gt{\Gamma} \cat \envPropR;b:\lhot{\Rts{}{s}{S}};\es \proves$ \\
					$b \hastype \lhot{\Rts{}{s}{S}}$
					\end{tabular}
					}
		\LeftLabel{\scriptsize Acc}
		\TrinaryInfC{$\Gt{\Gamma} \cat \envPropR;\es;\Theta' \cat
					\widetilde z: \Rts{}{s}{S} \proves
					\binp{\prop^r}{b}
					(\appl{b}{\widetilde z}) \hastype \Proc$}
		\DisplayProof
	\end{align}

	\begin{align}
		\label{pt:r-send4}
		\AxiomC{\eqref{pt:r-send2}}
		\AxiomC{}
		\LeftLabel{\scriptsize PolyVar}
		\UnaryInfC{$\Gt{\Gamma} \cat \envPropR;\Gt{\Lambda_2};\es
					\proves \widetilde w \hastype \widetilde M_1$}
		\LeftLabel{\scriptsize PolySend}
		\BinaryInfC{$\Gt{\Gamma} \cat \envPropR;\Gt{\Lambda_1};\Theta' \cat
					\widetilde z : \Rts{}{s}{S}
					\proves \propout{k+\degree+1}
      				{\widetilde w} \binp{\prop^r}{b}
					(\appl{b}{\widetilde z}) \hastype \Proc$}
		\DisplayProof
	\end{align}

	\begin{align}
		\label{pt:r-rcv1}
		\AxiomC{\eqref{pt:r-send4}}
		\AxiomC{\eqref{eq:r-case1-ih}} 
		\LeftLabel{\scriptsize (\lemref{lem:weaken}) with ${\envPropR}_1$}
		\UnaryInfC{\begin{tabular}{c}$\Gt{\Gamma'_2} \cat \envPropR;\Gt{\Lambda_2};\Theta_2
					\proves$$\V{k+1}{\tilde y}{V} \hastype \Gt{U}$
					\end{tabular}}
	\LeftLabel{\scriptsize (\lemref{lem:weaken}) with $\tilde z$}
		\UnaryInfC{$\Gt{\Gamma} \cat \envPropR;\Gt{\Lambda_2};\Theta_2
					\proves \V{k+1}{\tilde y}{V} \hastype \Gt{U}$}
		\LeftLabel{\scriptsize Send}
			\BinaryInfC{
			\begin{tabular}{c}
$\Gt{\Gamma} \cat \envPropR;\Gt{\Lambda_1 \cat \Lambda_2};
					\Gt{\Delta_2} \cat \widetilde z:\Rts{}{s}{S} \cat
					\Theta_2 \cat \Theta' \proves$ \\
					${\bout{z_{f(S)}}{\V{k+1}{\tilde y}{V}}}\propout{k+\degree+1}
      				{\widetilde w} \binp{\prop^r}{b}
					(\appl{b}{\widetilde z})$
					\end{tabular}}
		\DisplayProof
	\end{align}
	By \lemref{lem:indexcor} we know
  	that if $\widetilde z:\Rts{}{s}{S}$ then
  	$z_{f(S)}:\trec{t}\btout{\Gt{U}}\vart{t}$.

	\begin{align}
		\label{pt:r-abs}
		\AxiomC{\eqref{pt:r-rcv1}}
		\AxiomC{}
		\LeftLabel{\scriptsize PolySess}
		\UnaryInfC{$\Gt{\Gamma} \cat \envPropR;\es;\widetilde z : \Rts{}{s}{S}
					\proves \widetilde z \hastype \Rts{}{s}{S}$}
		\LeftLabel{\scriptsize PolyAbs}
		\BinaryInfC{$\Gt{\Gamma} \cat \envPropR; \Gt{\Lambda_1 \cat \Lambda_2};
					\Gt{\Delta_2} \cat \Theta_2 \cat \Theta'
			\proves N_V \hastype \lhot{\Rts{}{s}{S}}$}
		\DisplayProof
	\end{align}

	\begin{align}
		\label{pt:r-rec}
		\AxiomC{}
		\LeftLabel{\scriptsize LVar}
		\UnaryInfC{$\Gt{\Gamma} \cat \envPropR;\es;\es
					\proves \prop^r \hastype \chtype{\lhot{\Rts{}{s}{S}}})$}
	\AxiomC{}
	\LeftLabel{\scriptsize Nil}
		\UnaryInfC{$\Gt{\Gamma} \cat \envPropR;\es;\es
					\proves \inact \hastype \Proc$}
		\AxiomC{\eqref{pt:r-abs}}
		\LeftLabel{\scriptsize Req}
		\TrinaryInfC{$\Gt{\Gamma} \cat \envPropR;\Gt{\Lambda_1 \cat \Lambda_2};
					\Gt{\Delta_2} \cat \Theta_2 \cat \Theta'
					\proves \bbout{\prop^r}{N_V} \inact \hastype \Proc$}
		\DisplayProof
	\end{align}

	\begin{align}
		\label{pt:r-rcv}
		\AxiomC{\eqref{pt:r-rec}}
		\AxiomC{}
		\LeftLabel{\scriptsize PolyVar}
		\UnaryInfC{$\Gt{\Gamma} \cat \envPropR; \Gt{\Lambda_1 \cat \Lambda_2};\es
					\proves \widetilde x \hastype \widetilde M$}
		\LeftLabel{\scriptsize PolyRcv}
		\BinaryInfC{$\Gt{\Gamma_1} \cat \envPropR;\es;
					\Gt{\Delta_2} \cat \Theta_2 \cat \Theta'
					\proves \propinp{k}{\widetilde x}
      				\bbout{\prop^r}{N_V} \inact \hastype \Proc$}
      	\DisplayProof
	\end{align}

	The following tree proves this case:
	\begin{align}
		\AxiomC{\eqref{pt:r-rcv}}
      	\AxiomC{\eqref{eq:r-case1-ih}}
      	\LeftLabel{\scriptsize (\lemref{lem:weaken}) with ${\envPropR}_2$}
      	\UnaryInfC{$\Gt{\Gamma'_1}\cat {\envPropR};\es;
				\Gt{\Delta_1}\cat\Theta_1
 				\proves \B{k+\degree+1}{\tilde w}{P'} \hastype \Proc$}
      	\LeftLabel{\scriptsize (\lemref{lem:strength}) with $\tilde y$}
		\UnaryInfC{$\Gt{\Gamma_1}\cat {\envPropR};\es;
				\Gt{\Delta_1}\cat\Theta_1
 				\proves \B{k+\degree+1}{\tilde w}{P'} \hastype \Proc$}
		\LeftLabel{\scriptsize Par}
		\BinaryInfC{$\Gt{\Gamma_1} \cat \envPropR;\es; \Gt{\Delta_1 \cat \Delta_2}
		 			\cat \Theta
                  \proves \B{k}{\tilde x}{\bout{r}{V}P'} \hastype \Proc$}
		\DisplayProof
	\end{align}
	
	\item Case $P = \appl{V}{(\widetilde r, u_i)}$. We assume a certain order in the tuple $(\widetilde r, u_i)$: names in $\widetilde r$ have 
	recursive session types $\widetilde r = (r_1,\ldots,r_n) : (S_1,\ldots,S_n)$, and $u_i$ has non-recursive session type $u_i : C$. 
	 We distinguish two sub-cases: \rom{1} $V : \lhot{\widetilde S C}$ 
	 and \rom{2} $V : \shot{\widetilde S C}$. We will consider only sub-case 
	 \rom{1} since the other is similar. 
	 For this case Rule~\textsc{PolyApp} can be applied: 
	 \begin{align}
	 	\label{pt:r-app}
	 	\AxiomC{$\Gamma;\Lambda;\Delta_1 \cat {\envR}_1 \proves 
	 				V \hastype \lhot{\widetilde S C}$}
	 	\AxiomC{$\Gamma;\es;\Delta_2 \cat {\envR}_2 \proves 
	 				(\widetilde r, u_i) \hastype \widetilde S C$}
	 	\LeftLabel{\scriptsize PolyApp}
	 	\BinaryInfC{$\Gamma;\Lambda;\Delta_1 \cat \Delta_2 \cat {\envR}_1 \cat 
	 	{\envR}_2 \proves \appl{V}{(\widetilde r, u_i)}$}
	 	\DisplayProof
	 \end{align}
	 
	 Let $\widetilde x = \fv{V}$ and $\Gamma_1 \setminus \widetilde x$. 
	 Let $\widetilde x =
\fv{V}$ and let $\Theta_1$ be a balanced environment such that
$$
\dom{\Theta_1}=\{\prop_{k+1},\ldots,\prop_{k+\len{V}}\}
\cup \{\dual{\prop_{k+1}},\ldots,\dual{\prop_{k+\len{V}}}\}
$$
and $\Theta_1(\prop_{k+1}) = \btinp{\widetilde M} \tinact$
and
$\Theta_1(\dual{\prop_{k+1}}) = \btout{\widetilde M} \tinact$ where $\widetilde
M = (\Gt{\Gamma}\cat\Gt{\Lambda})(\widetilde x)$.

	We define: 
	\begin{align}
		{\envPropR}_1 = \prod_{r \in {\dom{{\envR}_1}}} c^r:\chtype{\lhot{\Rts{}{s}{{\envR}_1(r)}}}
	\end{align}
	Then, by IH (Part 2) on the first assumption of \eqref{pt:r-app}
	we have: 
	\begin{align}
		\label{eq:r-app-ih}
		\Gt{\Gamma_1}\cat {\envPropR}_1;\es; \Gt{\Delta_1}  \cat \Theta_1
		\proves \V{k+1}{\tilde x}{V}  \hastype \lhot{\Gt{\widetilde S C}}
	\end{align}
	By \defref{def:typesdecompenv}  and \defref{def:decomptyp-rec} and 
	the second assumption of \eqref{pt:r-app} we have: 
	\begin{align}
		\label{eq:r-app-ih2}
		\Gt{\Gamma};\es;\Gt{\Delta_2} \cat \Gt{{\envR}_2} \proves 
		(\widetilde r_1,\ldots,\widetilde r_n, \widetilde m):\Gt{\widetilde S C}
	\end{align}
	\noindent where 
	$\widetilde r_i = (r^i_{i},\ldots,r^i_{i+\len{\Gt{S_i}}-1})$
	for $i \in \{1,\ldots,n\}$ and $\widetilde m = (u_i,\ldots,u_{i+\len{G(C)}-1})$. 
	
	We define $\envPropR = {\envPropR}_1 \cat {\envPropR}_2$ where:
	\begin{align*}
		 {\envPropR}_2 = \prod_{r \in {\dom{{\envR}_2}}} c^r:\chtype{\lhot{\Rts{}{s}{{\envR}_2(r)}}}
	\end{align*}
	
	We define $\Theta = \Theta_1 \cat \prop_k : \btinp{\widetilde M} \tinact$.
	
	We will first consider the case where $n=3$; the proof is then generalized for any $n \geq 1$: 

\begin{itemize} 
    \item 
	If $n=3$ then $P = \appl{V}{(r_1,r_2,r_3,u_i)}$. 
	By \tabref{t:bdownrecur} we have:
	\begin{align*}
		\B{k}{\tilde x}{\appl{V}{(r_1,r_2,r_3,u_i)}} = 
		\propinp{k}{\tilde x}\bout{\prop^{r_1}}{\abs{\widetilde z_1}
		{\bout{\prop^{r_2}}{\abs{\widetilde z_2}{
		\bout{\prop^{r_3}}{\abs{\widetilde z_3}{Q}}\inact
		}}}\inact}\inact
	\end{align*}
	\noindent where 
	$Q = \appl{\V{k+1}{\tilde x}{V}}{(\widetilde z_1,\ldots,\widetilde z_n, \widetilde m)}$;
	$\widetilde{z_i} = (z^i_1,\ldots,z^i_{\len{\Gt{S_i}}})$
	for $i = \{1,2,3\}$; $\widetilde m =(u_{i},\ldots,u_{i+\len{\Gt{C}}-1})$.

	We shall prove the following judgment: 
	\begin{align}
			\Gt{\Gamma}\cat \envPropR;\es;\Gt{\Delta_1 \Delta_2} \cat \Theta 
			\proves \B{k}{\tilde x}{\appl{V}{(\widetilde r, u_i)}} \hastype \Proc
	\end{align}
	
	We use auxiliary derivations: 
	
	\begin{align}
		\label{pt:r-app-ih1}
		\AxiomC{\eqref{eq:r-app-ih}}
		\LeftLabel{\scriptsize (\lemref{lem:weaken}) with $\envPropR_2$}
		\UnaryInfC{\begin{tabular}{c}
			$\Gt{\Gamma}\cat \envPropR;\Gt{\Lambda};
						\Gt{\Delta_1} 
						\cat \Theta_1  
						\proves$  $\V{k+1}{\tilde x}{V} 
						\hastype \lhot{\Gt{\widetilde S C}}$
						\end{tabular}
						}
		\DisplayProof
	\end{align}

	\begin{align}
	\label{pt:r-app-ih2}
		\AxiomC{\eqref{eq:r-app-ih2}} 
	\LeftLabel{\scriptsize (\lemref{lem:subst}) with $\sigma$}
		\UnaryInfC{\begin{tabular}{c}$\Gt{\Gamma}\cat \envPropR; \es; 
					\Gt{\Delta_2} \cat \Gt{{\envR}_2} 
					\proves$  $(\widetilde z_1,\widetilde z_2, \widetilde z_3, \widetilde m) 
					\hastype \Gt{\widetilde S C}$
					\end{tabular}}
		\DisplayProof
	\end{align}
	
	\begin{align}
	\label{pt:r-app-3-6}
		\AxiomC{\eqref{pt:r-app-ih1}}
	\AxiomC{\eqref{pt:r-app-ih2}} 
		\LeftLabel{\scriptsize PolyApp}
		\BinaryInfC{$\Gt{\Gamma}\cat \envPropR;\Gt{\Lambda};
						\Gt{\Delta_1 \cat \Delta_2} 
						\cat \Theta_1 \cat \Gt{{\envR}_2} 
						\proves 
						\appl{\V{k+1}{\tilde x}{V}}
						{(\widetilde z_1,\widetilde z_2, \widetilde z_3, \widetilde m)}
						$}	
		\DisplayProof
	\end{align}
	\noindent where $\sigma = \subst{\widetilde n_1}{\widetilde z_1} 
	 \cdot \subst{\widetilde n_2}{\widetilde z_2} \cdot 
	 \subst{\widetilde n_3}{\widetilde z_3}$
	with $\widetilde n_i = (r^i_{i},\ldots,r^i_{i+\len{\Gt{S_i}}-1})$
	for $i = \{1,2,3\}$. 
	
	\begin{align}
		\label{pt:r-app-3-5}
		\AxiomC{\eqref{pt:r-app-3-6}} 
		\AxiomC{}
		\LeftLabel{\scriptsize PolySess}
		\UnaryInfC{$\Gt{\Gamma}\cat \envPropR;\es;
						\widetilde z_3 : \Gt{S_3} \proves 
		 				\widetilde z_3 \hastype \Gt{S_3}$}
		\LeftLabel{\scriptsize PolyAbs}
		\BinaryInfC{$\Gt{\Gamma}\cat \envPropR;\Gt{\Lambda};
						\Gt{\Delta_1 \Delta_2} 
						\cat \Theta_1 \proves
						\abs{\widetilde z_3}{Q} \hastype 
						\lhot{\Gt{S_3}}$}
		\DisplayProof
	\end{align}
	
	\begin{align}
		\label{pt:r-app-3-4}
		\AxiomC{\eqref{pt:r-app-3-5}}
		\AxiomC{} 
		\LeftLabel{\scriptsize Nil}
		\UnaryInfC{$\Gt{\Gamma}\cat \envPropR;\es;\es \proves \inact$}
		\AxiomC{}
		\LeftLabel{\scriptsize LVar}
		\UnaryInfC{$\Gt{\Gamma}\cat \envPropR;\es;\es \proves
					\prop^{r_3} \hastype \chtype{\lhot{\Gt{S_3}}}$} 
		\LeftLabel{\scriptsize Req}
		\TrinaryInfC{$\Gt{\Gamma}\cat \envPropR;\Gt{\Lambda};
						\Gt{\Delta_1 \Delta_2} \cat 
						\Theta_1 \cat 
						\widetilde z_1 : \Gt{S_1}
						\cat 
						\widetilde z_2 : \Gt{S_2} 
						 \proves 
						\bout{\prop^{r_3}}{\abs{\widetilde z_3}{Q}}\inact
						 \hastype 
						\Proc$}
		\DisplayProof
	\end{align}

	
	\begin{align}
		\label{pt:r-app-3-3}
		\AxiomC{\eqref{pt:r-app-3-4}} 
		\AxiomC{}
		\LeftLabel{\scriptsize PolySess}
		\UnaryInfC{$\Gt{\Gamma}\cat \envPropR;\es;
						\widetilde z_2 : \Gt{S_2} \proves 
		 				\widetilde z_2 \hastype \Gt{S_2}$}
		\LeftLabel{\scriptsize PolyAbs}
		\BinaryInfC{$\Gt{\Gamma}\cat \envPropR;\Gt{\Lambda};
						\Gt{\Delta_1 \Delta_2} 
						\cat \Theta_1 \cat 
						\widetilde z_1 : \Gt{S_1}\proves
						\abs{\widetilde z_2}{
						\bout{\prop^{r_3}}{\abs{\widetilde z_3}{Q}}\inact
						} \hastype 
						\lhot{\Gt{S_2}}$}
		\DisplayProof
	\end{align}
	
	\begin{align}
		\label{pt:r-app-3-2}
		\AxiomC{\eqref{pt:r-app-3-3}} 
		\AxiomC{}
		\LeftLabel{\scriptsize Nil}
		\UnaryInfC{$\Gt{\Gamma}\cat \envPropR;\es;\es \proves \inact$}
		\AxiomC{}
		\LeftLabel{\scriptsize LVar}
		\UnaryInfC{$\Gt{\Gamma}\cat \envPropR;\es;\es \proves
					\prop^{r_2} \hastype \chtype{\lhot{\Gt{S_2}}}$} 
		\LeftLabel{\scriptsize Req}
		\TrinaryInfC{$\Gt{\Gamma}\cat \envPropR;\Gt{\Lambda};
						\Gt{\Delta_1 \Delta_2} \cat \Theta_1 \cat 
						\widetilde z_1 : \Gt{S_1}
						  \proves 
						{\bout{\prop^{r_2}}{\abs{\widetilde z_2}{
						\bout{\prop^{r_3}}{\abs{\widetilde z_3}{Q}}\inact
						}}}\inact \hastype 
						\Proc$}
		\DisplayProof
	\end{align}	

	\begin{align}
		\label{pt:r-app-3-1}
		\AxiomC{\eqref{pt:r-app-3-2}} 
		\AxiomC{}
		\LeftLabel{\scriptsize PolySess}
		\UnaryInfC{$\Gt{\Gamma}\cat \envPropR;\es;
						\widetilde z_1 : \Gt{S_1} \proves 
		 				\widetilde z_1 \hastype \Gt{S_1}$}
		\LeftLabel{\scriptsize PolyAbs}
		\BinaryInfC{$\Gt{\Gamma}\cat \envPropR;\Gt{\Lambda};
						\Gt{\Delta_1 \Delta_2} 
						\cat \Theta_1 \proves
						\abs{\widetilde z_1}
						{\bout{\prop^{r_2}}{\abs{\widetilde z_2}{
						\bout{\prop^{r_3}}{\abs{\widetilde z_3}{Q}}\inact
						}}}\inact \hastype 
						\lhot{\Gt{S_1}}$}
		\DisplayProof
	\end{align}
	
	\begin{align}
		\label{pt:r-app-3}
		\AxiomC{\eqref{pt:r-app-3-1}} 
		\AxiomC{}
		\LeftLabel{\scriptsize Nil}
		\UnaryInfC{$\Gt{\Gamma}\cat \envPropR;\es;\es \proves \inact$}
		\AxiomC{} 
		\LeftLabel{\scriptsize LVar}
		\UnaryInfC{$\Gt{\Gamma}\cat \envPropR;\es;\es \proves
					\prop^{r_1} \hastype \chtype{\lhot{\Gt{S_1}}}$} 
		\LeftLabel{\scriptsize Req}
		\TrinaryInfC{$\Gt{\Gamma}\cat \envPropR;\Gt{\Lambda};
						\Gt{\Delta_1 \Delta_2} 
						\cat \Theta_1 \proves 
						\bout{\prop^{r_1}}{\abs{\widetilde z_1}
						{\bout{\prop^{r_2}}{\abs{\widetilde z_2}{
						\bout{\prop^{r_3}}{\abs{\widetilde z_3}{Q}}\inact
						}}}\inact}\inact$}
		\DisplayProof
	\end{align}
	The following tree proves this case: 
	\begin{align*}
		\AxiomC{\eqref{pt:r-app-3}} 
		\AxiomC{} 
		\LeftLabel{\scriptsize PolyVar}
		\UnaryInfC{$\Gt{\Gamma_1}\cat \envPropR;\Gt{\Lambda};\es \proves 
					\widetilde x \hastype \widetilde M$}
		\LeftLabel{\scriptsize PolyRcv}
		\BinaryInfC{$\Gt{\Gamma_1}\cat \envPropR;\es;\Gt{\Delta_1 \cat \Delta_2} 
						\cat \Theta \proves 
						\B{k}{\tilde x}{\appl{V}{(\widetilde r, u_i)}} 
						\hastype \Proc$}
		\DisplayProof	
	\end{align*}
\item 
	
	Now we consider the general case for any $n \geq 1$: 
	
	By \tabref{t:bdownrecur} we have: 
	\begin{align*}
	 	\B{k}{\tilde x}{\appl{V}{\widetilde r}} = \propinp{k}{\widetilde x}\overbracket{\prop^{r_1}!\langle \lambda \widetilde z_1. \ldots 
	 	\prop^{r_n}!\langle \lambda \widetilde z_n.}^{n = \len{\widetilde r}}
	 	Q
	 \rangle \ldots \rangle
	\end{align*}
	\noindent where 
	$Q = \appl{\V{k+1}{\tilde x}{V}}{(\widetilde r_1,\ldots,\widetilde r_n, \widetilde m)}$ with:
	$\widetilde{z_i} = (z^i_1,\ldots,z^i_{\len{\Gt{S_i}}})$
	for $i = \{1,\ldots,n\}$; and 
	$\widetilde m =(u_{i},\ldots,u_{i+\len{\Gt{C}}-1})$.

	We shall prove the following judgment: 
	\begin{align}
			\Gt{\Gamma}\cat \envPropR;\es;\Gt{\Delta_1 \Delta_2} \cat \Theta 
			\proves \B{k}{\tilde x}{\appl{V}{(\widetilde r, u_i)}} \hastype \Proc
	\end{align}

	We construct auxiliary derivations parametrized by $k$ and denoted by $d(k)$. 
	 If $k=n$, derivation $d(n)$ is defined as: 
	
	\begin{align} 
		\label{pt:r-app-b2}
		\AxiomC{\eqref{eq:r-app-ih}}
			\LeftLabel{\scriptsize (\lemref{lem:weaken}) with ${\envPropR}_2$} 
			\UnaryInfC{$\Gt{\Gamma}\cat \envPropR;\Gt{\Lambda};
						\Gt{\Delta_1} \cat \Theta_1 \proves \V{k+1}{\tilde x}{V}
						\hastype \shot{\Gt{\widetilde S C}} 
						$}
			\DisplayProof
	\end{align}

	\begin{align}
			\label{pt:r-app-b}
			\AxiomC{\eqref{pt:r-app-b2}}
			\AxiomC{\eqref{eq:r-app-ih2}}
			\LeftLabel{\scriptsize (\lemref{lem:subst}) with $\sigma$}
			\UnaryInfC{$\Gt{\Gamma}\cat \envPropR;\es;\Gt{{\envR}_2} 
							\proves (\widetilde r_1,\ldots,\widetilde r_n, \widetilde m)
							\hastype \Gt{\widetilde S C}$} 
			\LeftLabel{\scriptsize PolyAbs}
			\BinaryInfC{$\Gt{\Gamma}\cat \envPropR;\Gt{\Lambda};
						\Gt{\Delta_1 \cat \Delta_2} 
						\cat$ 
						$\Theta_1 \cat \Gt{{\envR}_2} \proves
						\abs{\widetilde z_n}{Q} \hastype 
						\lhot{\Gt{S_n}}$}
			\DisplayProof
	\end{align}
	
	\noindent where $\sigma = \prod_{i \in \{1,\ldots,n\}}\subst{\widetilde n_i}{\widetilde z_i}$
	with $\widetilde n_i = (r^i_{i},\ldots,r^i_{i+\len{\Gt{S_i}}-1})$
	for $i = \{1,\ldots,n\}$. 
	
	\begin{align}
	\AxiomC{\eqref{pt:r-app-b}} 
	\AxiomC{}
		\LeftLabel{\scriptsize Nil}
		\UnaryInfC{$\Gt{\Gamma}\cat \envPropR;\es;\es \proves \inact$}
	\AxiomC{}
	\LeftLabel{\scriptsize LVar}
	\UnaryInfC{$\Gt{\Gamma}\cat \envPropR;\es;\es \proves 
				\prop^{r_n} \hastype \chtype{\lhot{\Gt{S_n}}}$} 
		\LeftLabel{\scriptsize Req}
			\TrinaryInfC{$\Gt{\Gamma}\cat \envPropR;\Gt{\Lambda};
						\Gt{\Delta_1 \cat \Delta_2} 
						\cat$ 
						$\Theta_1 \cat \Gt{{\envR}_2} \proves \bout{\prop^{r_n}}{\abs{\widetilde z_n}{Q}}\inact$}
			\DisplayProof
	\end{align}

	Otherwise, if $k \in \{1,...,n-1\}$, derivation $d(k)$ is as follows: 
	 
	\begin{align}
		\label{pt:r-abs-k}
		\AxiomC{$d(k+1)$} 
	\AxiomC{}
	\LeftLabel{\scriptsize PolyVar}
	\UnaryInfC{$\Gt{\Gamma}\cat \envPropR;\es;\widetilde z_k : \Gt{S_k} 
				\proves \widetilde z_k \hastype \Gt{S_k}$} 
	\LeftLabel{\scriptsize PolyAbs}
	\BinaryInfC{\begin{tabular}{c}$\Gt{\Gamma}\cat \envPropR;\Gt{\Lambda};
						\Gt{\Delta_1 \cat \Delta_2} 
						\cat \Theta_1 
						\cat 
						(\widetilde z_1,\ldots,\widetilde z_{k-1}) : 
						(\Gt{S_1},\ldots,\Gt{S_{k-1}})
						\proves$ \\
						$\abs{\widetilde z_k}\overbracket{\prop^{r_{k+1}}!\langle \lambda 
						\widetilde z_{k+1}. \ldots 
	 	\prop^{r_n}!\langle \lambda \widetilde z_n.}^{n-k} 
	 	Q
	 	\overbracket{\rangle.\inact \ldots \rangle.\inact}^{n-k}
						\hastype \lhot{\Gt{S_1}}$
			\end{tabular}} 
			\DisplayProof 
	\end{align}
	
	\begin{align}
	\label{pt:r-acc-k}
	\AxiomC{\eqref{pt:r-abs-k}}
		\AxiomC{$\Gt{\Gamma}\cat \envPropR;\es; \es \proves \prop^{r_1} \hastype 
						\chtype{\lhot{\Gt{S_1}}}$}
	\LeftLabel{\scriptsize Acc}
		\BinaryInfC{
		\begin{tabular}{c}
		$\Gt{\Gamma}\cat \envPropR;\Gt{\Lambda};
						\Gt{\Delta_1 \cat \Delta_2} 
						\cat \Theta_1 
						\cat (\widetilde z_1,\ldots,\widetilde z_{k-1}) : 
						(\Gt{S_1},\ldots,\Gt{S_{k-1}})
						\proves$  \\
						$\overbracket{\prop^{r_k}!\langle \lambda 
						\widetilde z_k. \ldots 
	 	\prop^{r_n}!\langle \lambda \widetilde z_r.}^{n-k+1} 
	 	Q
	 	\overbracket{\rangle.\inact \ldots \rangle.\inact}^{n-k+1}
						$
				\end{tabular}}
		\DisplayProof	
	\end{align}


	The following tree proves this case: 
	\begin{align*}
		\AxiomC{$d(1)$} 
		\AxiomC{} 
		\LeftLabel{\scriptsize PolyVar}
		\UnaryInfC{$\Gt{\Gamma_1}\cat \envPropR;\Gt{\Lambda};\es \proves 
					\widetilde x \hastype \widetilde M$}
		\LeftLabel{\scriptsize PolyRcv}
		\BinaryInfC{$\Gt{\Gamma_1}\cat \envPropR;\es;\Gt{\Delta_1 \cat \Delta_2} 
						\cat \Theta \proves 
						\B{k}{\tilde x}{\appl{V}{(\widetilde r, u_i)}} 
						\hastype \Proc$}
		\DisplayProof	
	\end{align*}
\end{itemize}

\end{enumerate}
 \end{enumerate}
\end{proof}

\subsection{Proof of \thmref{t:decomprec}}
\label{app:decomprec}

\begin{proof} 
By assumption $\Gamma;\es;\Delta \cat \envR \proves P \hastype \Proc$.  Then,
	by applying \lemref{lem:subst} we have:
	\begin{align}
		\label{eq:typerec-subst-ih}
		\Gamma\sigma;\es;\Delta \sigma \cat \envR \sigma \proves P \sigma \hastype \Proc
	\end{align}

  By \thmref{t:typerec} on \eqref{eq:typerec-subst-ih} we have: 
  \begin{align}
  \label{eq:typerec-ih}
  	\Gt{\Gamma_1 \sigma}\cat \envPropR;\es;\Gt{\Delta \sigma} \cat \Theta	\proves 
  	\B{k}{\epsilon}{P\sigma} \hastype \Proc 
  \end{align}
	\noindent where $\Theta$ is balanced with
	$\dom{\Theta} = \{\prop_k,\prop_{k+1},\ldots,\prop_{k+\len{P}-1} \}
	\cup \{\dual{\prop_{k+1}},\ldots,\dual{\prop_{k+\len{P}-1}} \}$, and
	$\Theta(\prop_k)=\btinp{\cdot}\tinact$, and 
	  $\envPropR = \prod_{r \in \dom{\envR}} c^r:\chtype{\lhot{\Rts{}{s}{\envR(r)}}}$.
	 By assumption,  $\fv{P} = \es$.
  
  By \defref{def:decomp-rec}, we
	shall prove the following judgment:
	\begin{align*}
		\Gt{\Gamma \sigma};\es;\Gt{\Delta \sigma} \cat \Gt{\envR \sigma} \proves \news{\widetilde \prop}
		\news{\widetilde c_r}\Big(
  		\prod_{r \in \tilde v} \binp{\prop^r}{b}\appl{b}{\widetilde r} \Par
  		\propout{k}{} \inact \Par \B{k}{\epsilon}{P\sigma}\Big)
	\end{align*}
  \noindent where: $k >0$;  
  $\widetilde v = \rfn{P}$; $\widetilde r = 
  (r_1,\ldots,r_{\len{\Gt{S}}})$ for each $r \in \widetilde v$.
  
We know $\dom{\envR} = \widetilde v$.  We assume that
recursive session types are unfolded. 
  By \defref{def:typesdecompenv} and \defref{def:decomptyp-rec}, for
  $r \in \dom{\envR}$ we have:
   $$\Gt{\envR}(r) = \Rt{\envR(r)} = \Rts{}{s}{\envR(r)}$$
 We use a family of auxiliary derivations parametrized by $r \in \widetilde v$. 
  
  \begin{align}
  \label{pt:r-d1}
  \AxiomC{}
  \LeftLabel{\scriptsize LVar}
  \UnaryInfC{\begin{tabular}{c}$\Gt{\Gamma \sigma}\cat \envPropR;b:\lhot{\Rts{}{s}{\envR(r)}};\es \proves$\\
  		$b \hastype \lhot{\Rts{}{s}{\envR(r)}}$\end{tabular}}
  	\AxiomC{}
  	\LeftLabel{\scriptsize PolySess}
  \UnaryInfC{\begin{tabular}{c}$\Gt{\Gamma \sigma}\cat \envPropR;\es;\widetilde r:\Rts{}{s}{\envR(r)} \proves$ \\ $\widetilde r \hastype \Rts{}{s}{\envR(r)}$
  \end{tabular}}
  \LeftLabel{\scriptsize PolyApp}
  	\BinaryInfC{$\Gt{\Gamma \sigma}\cat \envPropR;b:\lhot{\Rts{}{s}{\envR(r)}};
  	\widetilde r:\Rts{}{s}{\envR(r)} \proves
  				(\appl{b}{\widetilde r})$}
  			\DisplayProof
  \end{align}

  \begin{align}
  	\label{pt:r-d}
  	\AxiomC{\eqref{pt:r-d1}} 
  	\AxiomC{}
  	\LeftLabel{\scriptsize Sh} 
  	\UnaryInfC{\begin{tabular}{c}$\Gt{\Gamma \sigma}\cat \envPropR;\es;\es \proves \prop^r \hastype$
 \\  				$\chtype{\lhot{\Rts{}{s}{\envR(r)}}}$\end{tabular}}
  	\AxiomC{}
  	\LeftLabel{\scriptsize LVar}
  	\UnaryInfC{\begin{tabular}{c}$\Gt{\Gamma \sigma}\cat \envPropR;b:\lhot{\Rts{}{s}{\envR(r)}};\es$ \\ 
  				$\proves b \hastype \lhot{\Rts{}{s}{\envR(r)}}$
  				\end{tabular}} 
  	\LeftLabel{\scriptsize Acc}
  	\TrinaryInfC{$\Gt{\Gamma \sigma}\cat \envPropR;\es;\widetilde r:\Rts{}{s}{\envR(r)} \proves \binp{\prop^r}{b}(\appl{b}{\widetilde r})$}
  	\DisplayProof
  \end{align}
  
  We will then use:

  
  \begin{align}
  \label{pt:prod}
  \AxiomC{for $r \in \widetilde v$} 
  \AxiomC{\eqref{pt:r-d}}
  \LeftLabel{\scriptsize Par ($\len{\widetilde v}-1$ times)}
  \BinaryInfC{$\Gt{\Gamma \sigma}\cat \envPropR;\es;\Gt{\envR \sigma} \proves \prod_{r \in \tilde v} \binp{\prop^r}{b}(\appl{b}{\widetilde r})$}
  \DisplayProof
  \end{align}
  
  where we apply Rule~\textsc{Par} $\len{\widetilde v}-1$ times and for every $r \in \widetilde v$ we apply derivation \eqref{pt:r-d}. 
  Notice that by \defref{def:typesdecompenv} and \defref{def:decomptyp-rec} we have $\Gt{\envR \sigma} = 
  \prod_{r \in \tilde v}\widetilde r : \Rts{}{}{\envR(r)}$.

  \begin{align}
  \label{pt:r-par}
  \AxiomC{}
  \LeftLabel{\scriptsize Sess}
      \UnaryInfC{$\Gt{\Gamma \sigma}\cat \envPropR;\es;\dual{\prop_k}:\btout{\cdot}\tinact
  			\proves \propout{k}{} \inact$}
  \AxiomC{\eqref{eq:typerec-ih}}
  \LeftLabel{\scriptsize Par}
  \BinaryInfC{$\Gt{\Gamma \sigma}\cat \envPropR;\es;\Gt{\Delta \sigma} \cat 
  				\Theta \cat \dual{\prop_k}:\btout{\cdot}\tinact \proves \propout{k}{} \inact \Par \B{k}{\epsilon}{P\sigma}$}
  				\DisplayProof 
  \end{align}
  
  The following tree proves this case: 
  \begin{align}
  \AxiomC{\eqref{pt:prod}} 
  \AxiomC{\eqref{pt:r-par}}
  \LeftLabel{\scriptsize Par}
  	\BinaryInfC{
  	\begin{tabular}{c}
  	$\Gt{\Gamma \sigma}\cat \envPropR;\es;\Gt{\Delta \sigma} \cat 
  	\Gt{\envR \sigma} \cat
  				\Theta \cat \dual{\prop_k}:\btout{\cdot}\inact \proves$  \\
  		$\prod_{r \in \tilde v} \binp{\prop^r}{b}(\appl{b}{\widetilde r}) \Par
  		\propout{k}{} \inact \Par \B{k}{\epsilon}{P\sigma}$
  		\end{tabular}} 
  	\LeftLabel{\scriptsize PolyRes}
  	\UnaryInfC{
  	\begin{tabular}{c}
  	$\Gt{\Gamma \sigma};\es;\Gt{\Delta \sigma} \cat \Gt{\envR \sigma} \cat
  				\Theta \cat \dual{\prop_k}:\btout{\cdot}\tinact \proves$ \\ 
  				$\news{\widetilde c_r}(
  		\prod_{r \in \tilde v} \binp{\prop^r}{b}(\appl{b}{\widetilde r}) \Par
  		\propout{k}{} \inact \Par \B{k}{\epsilon}{P\sigma})$
  		\end{tabular}}
  	\LeftLabel{\scriptsize PolyResS}
  	\UnaryInfC{$\Gt{\Gamma \sigma};\es;\Gt{\Delta \sigma} \cat 
  	\Gt{\envR \sigma}
  	\proves \news{\widetilde \prop}\news{\widetilde c_r}(
  		\prod_{r \in \tilde v} \binp{\prop^r}{b}(\appl{b}{\widetilde r}) \Par
  		\propout{k}{} \inact \Par \B{k}{\epsilon}{P\sigma})$}	
  	\DisplayProof
  \end{align}

\end{proof}

\end{document}